\documentclass[a4paper,12pt]{article}
\usepackage{amsmath,amssymb,amsthm}
\usepackage{abstract}
\usepackage{threeparttable}
\usepackage[table]{xcolor}
\usepackage{adjustbox}
\usepackage{authblk}
\usepackage{subcaption}
\usepackage{graphicx}
\usepackage{mathtools}
\usepackage{pdflscape}
\usepackage{colortbl}
\usepackage{comment}
\usepackage{amsfonts}
\usepackage{bbm,bm}
\usepackage{physics}
\usepackage{float}
\usepackage[labelfont=bf]{caption}
\usepackage{caption}
\usepackage[bottom]{footmisc}
\usepackage{cancel}
\usepackage{cases}
\usepackage{ascmac}
\usepackage{ulem} 
\usepackage[english]{babel}
\usepackage{stmaryrd}
\usepackage{multirow} 
\usepackage[top=25truemm,bottom=25truemm,left=20truemm,right=20truemm]{geometry}
\usepackage{setspace}
\usepackage{hyperref}
\hypersetup{
    colorlinks=true,
    linkcolor=orange,
    filecolor=orange,      
    urlcolor=orange,
    citecolor=orange,
}
\usepackage{physics} 
\usepackage{titlesec}
\titleformat{\section}
  {\normalfont\large\bfseries}{\thesection}{1em}{}
  \titleformat{\subsection}
  {\normalfont\normalsize\bfseries}{\thesubsection}{1em}{}
\newtheorem*{Estimation and Inference}{Estimation and Inference}
\newtheorem*{Repeated cross sections}{Repeated cross sections}
\newtheorem*{Comparing Fuzzy DID}{Comparing Fuzzy DID}
\newtheorem*{Unbalanced case}{Unbalanced case}
\newtheorem*{Cross section comparison}{Cross section comparison}
\newtheorem*{Supplementary of Theorem 3}{Supplementary of Theorem 3}
\newtheorem*{Non-binary, ordered treatment}{Non-binary, ordered treatment}
\newtheorem*{Particular DID-IV design}{Particular DID-IV design}
\newtheorem*{Random assignment of the adoption date}{Random assignment of the adoption date}
\newtheorem*{Introducing the treatment's path}{Introducing the treatment's path}
\newtheorem*{Design based approach}{Design based approach}
\newtheorem{Proof of Theorem}{Proof of Theorem}
\newtheorem{Proof of Lemma}{Proof of Lemma}
\newtheorem{proof of}{proof of}

\theoremstyle{definition}
\newtheorem{Assumption}{Assumption}
\newtheorem*{Discussion in section 5.4}{Discussion in section 5.4}

\newtheorem{Theorem}{Theorem}

\newtheorem*{Def}{Def}
\newtheorem{Lemma}{Lemma}

\newtheorem*{No anticipation for the first stage}{No anticipation for the first stage}
\newtheorem*{Parallel Trend for the first stage}{Parallel Trend for the first stage}
\newtheorem*{Parallel Trend for the second stage}{Parallel Trend for the second stage}

\usepackage{listings}
\usepackage[T1]{fontenc} %change font encoding to T1
\usepackage[framed,numbered]{matlab-prettifier}
\usepackage[round]{natbib}

\title{Two-way fixed effects instrumental variable regressions\\ in staggered DID-IV designs.\thanks{I am grateful to Daiji Kawaguchi and Ryo Okui for their continued guidance and support. All errors are my own.}}
\author{Miyaji Sho\thanks{Graduate School of Economics, The University of Tokyo, 7-3-1 Hongo, Bunkyoku, Tokyo 113-0033, Japan; Email: \href{mailto:shomiyaji-apple@g.ecc.u-tokyo.ac.jp}{shomiyaji-apple@g.ecc.u-tokyo.ac.jp}.}}

\date{\today}
\begin{document}
\begin{titlepage}
\maketitle
\begin{abstract}
Many studies run two-way fixed effects instrumental variable (TWFEIV) regressions, leveraging variation in the timing of policy adoption across units as an instrument for treatment. This paper studies the properties of the TWFEIV estimator in staggered instrumented difference-in-differences (DID-IV) designs. We show that in settings with the staggered adoption of the instrument across units, the TWFEIV estimator can be decomposed into a weighted average of all possible two-group/two-period Wald-DID estimators. Under staggered DID-IV designs, a causal interpretation of the TWFEIV estimand hinges on the stable effects of the instrument on the treatment and the outcome over time. We illustrate the use of our decomposition theorem for the TWFEIV estimator through an empirical application.
\end{abstract} 
\vspace*{\fill}
\end{titlepage}

\clearpage

\section{Introduction}\label{sec1}
Instrumented difference-in-differences (DID-IV) is a method to estimate the effect of a treatment on an outcome, exploiting variation in the timing of policy adoption across units as an instrument for the treatment. In a simple setting with two groups and two periods, some units become exposed to the policy shock in the second period (exposed group), whereas others are not over two periods (unexposed group). The estimator is constructed by running the following IV regression with the group and post-time dummies as included instruments and the interaction of the two as the excluded instrument (e.g., \cite{Duflo2001-nh}, \cite{Field2007-yc}):
\begin{align*}
Y_{i,t}=\beta_{0}+\beta_{i,.}\text{Exposed}_{i}+\beta_{,.t}\text{POST}_{t}+\beta_{IV}D_{i,t}+\epsilon_{i,t}.
\end{align*}
The resulting IV estimand $\beta_{IV}$ scales the DID estimand of the outcome by the DID estimand of the treatment, the so-called Wald-DID estimand (\cite{De_Chaisemartin2018-xe}, \cite{Miyaji2023}). In this two-group/two-period ($2 \times 2$) setting, DID-IV designs mainly consist of a monotonicity assumption and parallel trends assumptions in the treatment and the outcome between the two groups, and allow for the Wald-DID estimand to capture the local average treatment effect on the treated (LATET) (\cite{chasemartin2010-ch}, \cite{Hudson2017-tm}, and \cite{Miyaji2023}). DID-IV designs have gained popularity over DID designs in practice when there is no control group or the treatment adoption is potentially endogenous over time (\cite{Miyaji2023}).
\par 
In reality, however, most DID-IV applications go beyond the canonical DID-IV set up, and leverage variation in the timing of policy adoption across units in more than two periods, instrumenting for the treatment with the natural variation. The instrument is constructed, for instance from the staggered adoption of school reforms across countries or municipalities (e.g. \cite{Oreopoulos2006-bn}, \cite{Lundborg2014-gm}, \cite{Meghir2018-bk}), the phase-in introduction of head starts across states (e.g. \cite{Johnson2019-kb}), or the gradual adoption of broadband internet programs (e.g. \cite{Akerman2015-hh}, \cite{Bhuller2013-ki}). These policy changes can be viewed as some natural experiments, but not randomized in reality.\par 
Recently, \cite{Miyaji2023} formalizes the underlying identification strategy as a staggered DID-IV design. In this design, the treatment adoption is allowed to be endogenous over time, while the instrument is required to be uncorrelated with time-varying unobservables in the treatment and the outcome; the assignment of the treatment can be non-staggered across units, while the assignment of the instrument is staggered across units: they are partitioned into mutually exclusive and exhaustive cohorts by the initial adoption date of the instrument. The target parameter is the cohort specific local average treatment effect on the treated (CLATT); this parameter measures the treatment effects among the units who belong to cohort $e$ and are induced to the treatment by instrument in a given relative period $l$ after the initial adoption of the instrument. The identifying assumptions are the natural generalization of those in $2 \times 2$ DID-IV designs.\par
In practice, empirical researchers commonly implement this design via linear instrumental variable regressions with time and unit fixed effects, the so-called two-way fixed effects instrumental variable (TWFEIV) regressions (e.g., \cite{Black2005-aw}, \cite{Lundborg2017-mz}, \cite{Johnson2019-kb}): 
\begin{align}
\label{intro1}
    &Y_{i,t}=\phi_{i.}+\lambda_{t.}+\beta_{IV} D_{i,t}+v_{i,t},\\
    \label{intro2}
    &D_{i,t}=\gamma_{i.}+\zeta_{t.}+\pi Z_{i,t}+\eta_{i,t}.
\end{align}
In contrast to the canonical DID-IV set up, however, the validity of running TWFEIV regressions seems less clear under staggered DID-IV designs. The IV estimate is commonly interpreted as measuring the local average treatment effect in the presence of heterogeneous treatment effects as in \cite{Imbens1994-qy}, whereas the target parameter is not stated formally. We know little about how the IV estimator is constructed by comparing the evolution of the treatment and the outcome across units and over time. Finally, we have no tools to illustrate the identifying variations in the IV estimate in a given application.\par
In this paper, we study the properties of two-way fixed effects instrumental variable estimators under staggered DID-IV designs. Specifically, we present the decomposition result for the TWFEIV estimator, and study the causal interpretation of the TWFEIV estimand under staggered DID-IV designs.\par 
First, we derive the decomposition theorem for the TWFEIV estimator with settings of the staggered adoption of the instrument across units. We show that the TWFEIV estimator is equal to a weighted average of all possible $2 \times 2$ Wald-DID estimators arising from the three types of the DID-IV design. First, in an Unexposed/Exposed design, some units are never exposed to the instrument during the sample period (unexposed group), whereas some units start exposed at a particular date and remain exposed (exposed group). Second, in an Exposed/Not Yet Exposed design, some units start exposed earlier, whereas some units are not yet exposed during the design period (not yet exposed group). Finally, in an Exposed/Exposed Shift design, some units are already exposed, whereas some units start exposed later at a particular point during the design period (exposed shift group). The weight assigned to each Wald-DID estimator reflects all the identifying variations in each DID-IV design: the sample share, the variance of the instrument, and the DID estimator of the treatment between the two groups.\par
Built on the decomposition result, we next uncover the shortcomings of running TWFEIV regressions under staggered DID-IV designs. We show that the TWFEIV estimand potentially fails to summarize the treatment effects under staggered DID-IV designs due to negative weights. Specifically, we show that this estimand is equal to a weighted average of all possible cohort specific local average treatment effect on the treated (CLATT) parameters, but some weights can be negative. The negative weight problem potentially arises due to the "bad comparisons" (c.f. \cite{Goodman-Bacon2021-ej}) performed by TWFEIV regressions: the already exposed units play the role of controls in the Exposed/Exposed Shift design in the first stage and reduced form regressions. Given the negative result of using the TWFEIV estimand under staggered DID-IV designs, we also investigate the sufficient conditions for this estimand to attain its causal interpretation. We show that this estimand can be interpreted as causal only if the effects of the instrument on the treatment and the outcome are stable over time.\par
We extend our decomposition result in several directions. We first consider non-binary, ordered treatment. We also derive the decomposition result for the TWFEIV estimand in unbalanced panel settings. Lastly, we consider the case when the adoption date of the instrument is randomized across units. In all cases, we show that the TWFEIV estimand potentially fails to summarize the treatment effects under staggered DID-IV designs due to negative weights.\par
We illustrate our findings with the setting of \cite{Miller2019-ok} who estimate the effect of female police officers' share on intimate partner homicide rate, leveraging the timing variation of AA (affirmative action) plans across U.S. counties. In this application, we first assess the plausibility of the staggered DID-IV design implicitly imposed by \cite{Miller2019-ok} and confirm its validity. We then estimate TWFEIV regressions, slightly modifying the authors' setting, and apply our DID-IV decomposition theorem to the IV estimate. We find that the estimate assigns more weights to the Unexposed/Exposed design and less weights to the other two types of the DID-IV design. Despite the small weight on the Exposed/Exposed Shift design, we also find that the IV estimate suffers from the substantial downward bias arising from the bad comparisons in the Exposed/Exposed Shift design.\par
Finally, we develop simple tools to examine how different specifications affect the change in TWFEIV estimates, and illustrate these by revisiting \cite{Miller2019-ok}. In many empirical settings, researchers typically diverge from a simple TWFEIV regression as in equation \eqref{intro1} and estimate various specifications such as weighting or including time-varying covariates. We follow \cite{Goodman-Bacon2021-ej} and decompose the difference between the two specifications into the changes in Wald-DID estimates, the changes in weights, and the interaction of the two. This decomposition result enables the researchers to quantify the contribution of the changes in each term to the difference in the overall estimates. In addition, plotting the pairs of Wald-DID estimates and associated weights obtained from the two specifications allows the researchers to investigate which components have the significant impact on these contributions.\par
Overall, this paper shows the negative result of using TWFEIV estimators under staggered DID-IV designs in more than two periods, and provide tools to illustrate how serious that concern is in a given application. Specifically, our decomposition result for the TWFEIV estimator enables the researchers to quantify the bias term arising from the bad comparisons in Exposed/Exposed Shift designs in the data. Fortunately, \cite{Miyaji2023} recently proposes the alternative estimation method in staggered DID-IV designs that is robust to treatment effect heterogeneity. Using such estimation method allows the practitioners to avoid the issue of TWFEIV estimators in practice, and facilitates the credibility of their empirical findings.\par
The rest of the paper is organized as follows. The next subsection discusses the related literature. Section \ref{sec2} presents our decomposition theorem for the TWFEIV estimator. Section \ref{sec3} formally introduces staggered instrumented difference-in-differences designs. Section \ref{sec4} presents the pitfalls of running TWFEIV regressions under staggered DID-IV designs, and explores the sufficient conditions for the TWFEIV estimand to attain its causal interpretation. Section \ref{sec5} describes some of the extensions. Section \ref{sec6} presents our empirical application. Section \ref{sec7} explain how different specifications affect the difference in estimates and Section \ref{sec8} concludes. All proofs are given in the Appendix.
\subsection{Related literature}\label{sec1.1}
Our paper is related to the recent DID-IV literature (\cite{chasemartin2010-ch}; \cite{Hudson2017-tm}; \cite{De_Chaisemartin2018-xe}; \cite{Miyaji2023}). In this literature, \cite{chasemartin2010-ch} first formalizes $2 \times 2$ DID-IV designs and shows that a Wald-DID estimand identifies the local average treatment effect on the treated (LATET) if the parallel trends assumptions in the treatment and the outcome, and a monotonicity assumption are satisfied. \cite{Hudson2017-tm} also consider $2 \times 2$ DID-IV designs with non-binary, ordered treatment settings. Build on the work in \cite{chasemartin2010-ch}, however, \cite{De_Chaisemartin2018-xe} formalize $2 \times 2$ DID-IV designs differently, and call them Fuzzy DID. \cite{Miyaji2023} compares $2 \times 2$ DID-IV to Fuzzy DID designs and points out the issues embedded in Fuzzy DID designs, and extends $2 \times 2$ DID-IV design to multiple period settings with the staggered adoption of the instrument across units, which the author calls staggered DID-IV designs. \cite{Miyaji2023} also provides a reliable estimation method in staggered DID-IV designs that is robust to treatment effect heterogeneity.\par 
In this paper, we contribute to the literature by showing the properties of two-way fixed instrumental variable estimators in staggered DID-IV designs. In reality, when empirical researchers implicitly rely on the staggered DID-IV design, they commonly implement this design via TWFEIV regressions (e.g. \cite{Black2005-aw}, \cite{Lundborg2014-gm}, \cite{Meghir2018-bk}). This paper presents the issues of the conventional approach, and provides the sufficient conditions for this estimand to attain its causal interpretation.\par 
Our paper is also related to a recent DID literature on the causal interpretation of two-way fixed effects (TWFE) regressions and its dynamic specifications under heterogeneous treatment effects (\cite{Athey2022-uo}; \cite{Borusyak2021-jv}; \cite{De_Chaisemartin2020-dw}; \cite{Goodman-Bacon2021-ej}; \cite{Imai2021-dn}; \cite{Sun2021-rp}).\par
Specifically, this paper is closely connected to \cite{Goodman-Bacon2021-ej}, who derives the decomposition theorem for the TWFE estimator with settings of the staggered adoption of the treatment across units. In this paper, we establish the decompose theorem for the TWFEIV estimator with settings of the staggered adoption of the instrument across units, which is a natural generalization of their theorem 1.\par 
This paper is also closely connected to \cite{De_Chaisemartin2020-dw}, who decompose the TWFE estimand and present the issue of using this estimand under DID designs: some weights assigned to the causal parameters in this estimand can be potentially negative. In their appendix, the authors also decompose the TWFEIV estimand and refer to the negative weight problem in this estimand. Specifically, they apply the decomposition theorem for the TWFE estimand to the numerator and denominator in the TWFEIV estimand respectively, and conclude that this estimand identifies the LATE as in \cite{Imbens1994-qy} only if the effects of the instrument on the treatment and outcome are constant across groups and over time. However, their decomposition result for the TWFEIV estimand has some drawbacks. First, they do not formally state the target parameter and identifying assumptions in DID-IV designs. Second, their decomposition result is not based on the target parameter in DID-IV designs. Finally, the sufficient conditions for this estimand to be interpretable causal parameter are not well investigated.\par 
In this paper, we investigate the causal interpretation of the TWFEIV estimand more clearly than that of \cite{De_Chaisemartin2020-dw}. Specifically, we first decompose the TWFEIV estimator into all possible $2 \times 2$ Wald-DID estimators. We then formally introduce the target parameter and identifying assumptions in staggered DID-IV designs, built on the recent work in \cite{Miyaji2023}. This allows us to decompose the TWFEIV estimand into a weighted average of the target parameter in staggered DID-IV designs. Finally, we assess the causal interpretation of the TWFEIV estimand under a variety of restrictions on the effects of the instrument on the treatment and outcome, which clarifies the sufficient conditions for this estimand to attain its causal interpretation.\par   
We note that this paper is distinct from the recent IV literature on the causal interpretation of two stage least square (TSLS) estimators with covariates under heterogeneous treatment effects (\cite{Sloczynski2020-uk}, \cite{Blandhol2022-lk}). These recent studies investigate the causal interpretation of the TSLS estimand with covariates under the random variation of the instrument conditional on covariates, and cast doubt on the LATE (or LATEs) interpretation of this estimand. In this literature, the identifying variations come from the assignment process of the instrument. In this paper, however, we investigate the causal interpretation of the TWFEIV estimand (where time and unit dummies can be viewed as covariates) under staggered DID-IV designs: our identifying variations mainly come from the parallel trends assumptions in the treatment and the outcome over time.

\section{Instrumented difference-in-differences decomposition}\label{sec2}
In this section, we present a decomposition result for the two-way fixed effects instrumental variable (TWFEIV) estimator in multiple time period settings with the staggered adoption of the instrument across units.\par
\subsection{Set up}\label{sec2.1}
We introduce the notation we use throughout this article. We consider a panel data setting with $T$ periods and $N$ units. For each $i \in \{1,\dots N\}$ and $t \in \{1,\dots,T\}$, let $Y_{i,t}$ denote the outcome and $D_{i,t} \in \{0,1\}$ denote the treatment status, and $Z_{i,t}\in \{0,1\}$ denote the instrument status. Let $D_i=(D_{i,1},\dots,D_{i,T})$ and $Z_i=(Z_{i,1},\dots,Z_{i,T})$ denote the path of the treatment and the path of the instrument for unit $i$, respectively. Throughout this article, we assume that $\{Y_{i,t},D_{i,t},Z_{i,t}\}_{t=1}^{T}$ are independent and identically distributed (i.i.d). \par
We make the following assumption about the assignment process of the instrument.\par
\begin{Assumption}[Staggered adoption for $Z_{i,t}$]
\label{sec2as1}
For $s < t$, $Z_{i,s} \leq Z_{i,t}$ where $s,t \in \{1,\dots T\}$.
\end{Assumption}
Assumption \ref{sec2as1} requires that once units start exposed to the instrument, they remain exposed to that instrument afterward. In the DID literature, several recent papers impose this assumption on the adoption process of the treatment and sometimes call it the "staggered treatment adoption", see, e.g., \cite{Athey2022-uo}, \cite{Callaway2021-wl} and \cite{Sun2021-rp}.\par
Given Assumption \ref{sec2as1}, we can uniquely characterize the instrument path by the time period when unit $i$ is first exposed to the instrument, denoted by $E_{i}=\min\{t: Z_{i,t}=1\}$. If unit $i$ is not exposed to the instrument for all time periods, we define $E_{i}=\infty$. Based on the initial exposure period $E_i$, we can uniquely partition units into mutually exclusive and exhaustive cohorts $e$ for $e \in \{1,2,\dots, T,\infty\}$: all the units in cohort $e$ are first exposed to the instrument at time $E_{i}=e$. Hereafter, to ease the notation, we assume that the data contain $K$ cohorts $(K \leq T)$ where $e \in \{1,\dots,k,\dots, K\}$, and define $U$ as the never exposed cohort $E_i=\infty$.\par
Let $n_e$ be the relative sample share for cohort $e$ and let $\bar{Z}_e$ be the time share of the exposure to the instrument for cohort $e$:
\begin{align*}
n_e \equiv \frac{\sum_i{\mathbf{1}\{E_i=e\}}}{N},\hspace{3mm}\bar{Z}_e \equiv \frac{\sum_t{\mathbf{1}\{t \geq e\}}}{T}.
\end{align*}
We also define $n_{ab} \equiv \displaystyle\frac{n_a}{n_a+n_b}$ to be the relative sample share between cohort $a$ and $b$.\par 
In contrast to the staggered adoption of the instrument across units, we allow the general adoption process for the treatment: the treatment can potentially turn on/off repeatedly over time. \cite{De_Chaisemartin2020-dw} and \cite{Imai2021-dn} consider the same setting in the recent DID literature.\par
The notations $PRE(a)$, $MID(a,b)$, and $POST(a)$ represent the corresponding time window, respectively: $PRE(a) \equiv [1,a)$, $MID(a,b) \equiv [a,b)$, and $POST(a) \equiv [a,T]$. Let $\bar{R}_{e}^{POST(a)}$ be the sample mean of the random variable $R_{i,t}$ in cohort $e$ during the time window $POST(a)$:
\begin{align*}
\bar{R}_{e}^{POST(a)} \equiv \frac{1}{T-(a-1)}\sum_{a}^{T}\left[\frac{\sum_{i}R_{i,t}\mathbf{1}\{E_i=e\}}{\sum_{i}\mathbf{1}\{E_i=e\}}\right].
\end{align*}
We define $\bar{R}_{e}^{PRE(a)}$ and $\bar{R}_{e}^{MID(a,b)}$ analogously, representing the sample mean of the random variable $R_{i,t}$ in cohort $e$ during the time window $PRE(a)$ and $MID(a,b)$ respectively.\par
\subsection{Decomposing the TWFEIV estimator}\label{2.2}
We consider a TWFEIV regression in multiple time period settings with the staggered adoption of the instrument across units:
\begin{align}
\label{TWFEIV1}
    &Y_{i,t}=\phi_{i.}+\lambda_{t.}+\beta_{IV} D_{i,t}+v_{i,t},\\
    \label{TWFEIV2}
    &D_{i,t}=\gamma_{i.}+\zeta_{t.}+\pi Z_{i,t}+\eta_{i,t}.
\end{align}
By substituting the first stage regression \eqref{TWFEIV2} into the structural equation \eqref{TWFEIV1}, we obtain the reduced form regression:
\begin{align}
\label{TWFEIV3}
Y_{i,t}=\phi_{i.}+\lambda_{t.}+\alpha Z_{i,t}+v_{i,t}.
\end{align}
The ratio between the first stage coefficient $\hat{\pi}$ and the reduced form coefficient $\hat{\alpha}$ yields the TWFEIV estimator $\hat{\beta}_{IV}$. By the Frisch-Waugh-Lovell theorem, the IV estimator $\hat{\beta}_{IV}$ is equal to the ratio between the coefficient from regressing $Y_{i,t}$ on the double-demeaning variable $\Tilde{Z}_{i,t}$ and the coefficient from regressing $D_{i,t}$ on the same variable:
\begin{align}
\label{sec4eq3}
\hat{\beta}_{IV}&=\frac{\frac{1}{NT}\sum_{i}\sum_{t}\Tilde{Z}_{i,t}Y_{i,t}}{\frac{1}{NT}\sum_{i}\sum_{t}\Tilde{Z}_{i,t}D_{i,t}},
\end{align}
where $\Tilde{Z}_{i,t}$ is the double demeaning variable defined below:
\begin{align*}
\Tilde{Z}_{i,t}&=Z_{i,t}-\frac{1}{T}\sum_{t=1}^T Z_{i,t}-\frac{1}{N}\sum_{i=1}^N Z_{i,t}+\frac{1}{NT}\sum_{t=1}^T\sum_{i=1}^N Z_{i,t}\\
&\equiv (Z_{i,t}-\bar{Z}_i)-(\bar{Z}_t-\bar{\bar{Z}}).
\end{align*}\par
Note that the TWFEIV regression runs the two-way fixed effects (TWFE) regression twice, as can be seen in equations \eqref{TWFEIV2} and \eqref{TWFEIV3}. Because we assume the staggered assignment of the instrument across units, if we focus on the TWFE coefficient on $Z_{i,t}$ in the first stage or the reduced form regression, we can show that it is equal to a weighted average of all possible $2 \times 2$ DID estimators of the treatment or the outcome from the decomposition result for the TWFE estimator shown by \cite{Goodman-Bacon2021-ej}.\par
Consider the simple setting where we have only two periods and two cohorts: one cohort is not exposed to the instrument during the two periods ($E_i=U$), whereas the other cohort starts exposed to the instrument in the second period ($E_i=2$). In this setting, the TWFEIV estimator takes the following form, the so-called Wald-DID estimator (\cite{De_Chaisemartin2018-xe}, \cite{Miyaji2023}):
\begin{align*}
\hat{\beta}_{IV} = \frac{\bar{Y}_{2,2}-\bar{Y}_{2,1}-(\bar{Y}_{U,2}-\bar{Y}_{U,1})}{\bar{D}_{2,2}-\bar{D}_{2,1}-(\bar{D}_{U,2}-\bar{D}_{U,1})},
\end{align*}
where $\bar{R}_{a,t}$ is the sample mean of the random variable $R_{i,t}$ for cohort $E_i=a$ in time $t$. This estimator scales the DID estimator of the outcome by the DID estimator of the treatment between cohort $E_i=U$ and $E_i=2$.\par
The above observations bring us the intuition about how we can decompose the TWFEIV estimator with settings of the staggered adoption of the instrument across units; we expect that the TWFEIV estimator can be decomposed into a weighted average of all possible $2 \times 2$ Wald-DID estimators (instead of DID-estimators).\par
To clarify this intuition, assume for now that we have only three cohorts, an early exposed cohort $k$, a middle exposed cohort $l$ $(k < l)$, and a never exposed cohort $U$ $(E_i=\infty)$. Figure \ref{Figure1} plots the simulated data for the time trends of the average treatment (first stage) and the average outcome (reduced form) in three cohorts.\par
\begin{figure}[t]
  \begin{minipage}[b]{0.5\columnwidth}
    \centering
    \includegraphics[width=\columnwidth]{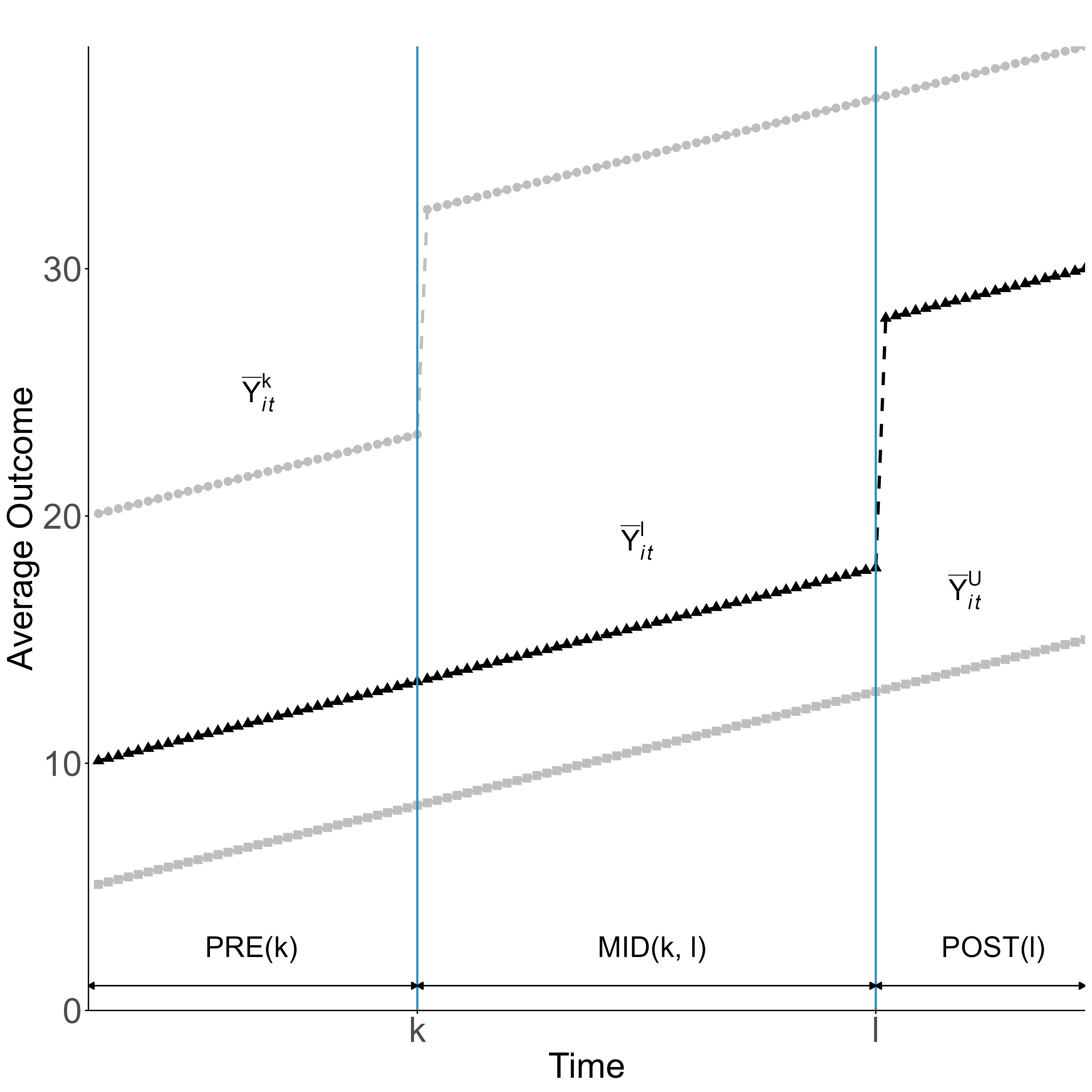}
  \end{minipage}
  \hspace{0.02\columnwidth} % ここで隙間作成
  \begin{minipage}[b]{0.5\columnwidth}
    \centering
    \includegraphics[width=\columnwidth]{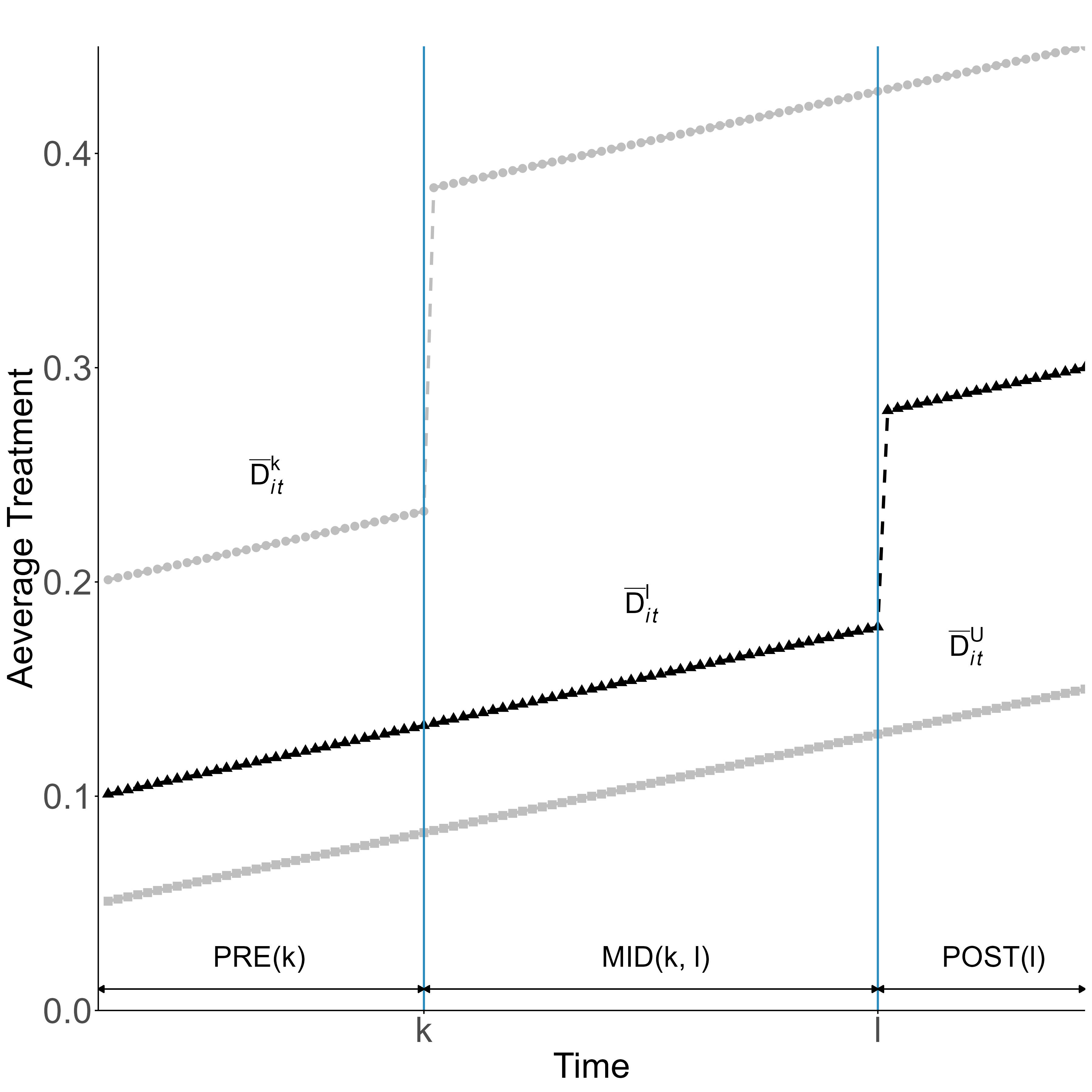}
  \end{minipage}
  \caption{Instrumented difference-in-differences with three cohorts. {\it Notes:} This figure plots the simulated data for the time trends of the average treatment (first stage) and the average outcome (reduced form) with time length $T=100$ in three cohorts: an early exposed cohort $k$, which is exposed to the instrument at $k=\frac{34}{100}T$; a middle exposed cohort $l$, which is exposed to the instrument at $l=\frac{80}{100}T$; a never exposed cohort, $U$. The $x$-axis consists of three time windows: the pre-exposed period for cohort $k$, $[1,k-1]$, denoted by $PRE(k)$; the middle exposed period when the cohort $k$ is already exposed but cohort $l$ is not yet exposed, $[k,l-1]$, denoted by $MID(k,l)$; and post-exposed period when cohort $l$ is already exposed, $[l, T]$, denoted by $POST(l)$. The effects of the instrument on the treatment and the outcome are $0.15$ and $9$ in cohort $k$ respectively; $0.1$ and $10$ in cohort $l$ respectively.}
  \label{Figure1}
\end{figure}
From the data structure, we can construct the Wald-DID estimator in three ways. First, we can compare the evolution of the treatment and the outcome between exposed cohort $j=k,l$ and never exposed cohort $U$, exploiting the time window $POST(j)$ and $PRE(j)$, which we call an Unexposed/Exposed design:
\begin{align}
\label{sec4eq4}
\hat{\beta}_{IV, jU}^{2\times2} &\equiv \frac{\left(\bar{y}_{j}^{POST(j)}-\bar{y}_{j}^{PRE(j)}\right)-\left(\bar{y}_{U}^{POST(j)}-\bar{y}_{U}^{PRE(j)}\right)}{\left(\bar{D}_{j}^{POST(j)}-\bar{D}_{j}^{PRE(j)}\right)-\left(\bar{D}_{U}^{POST(j)}-\bar{D}_{U}^{PRE(j)}\right)},\hspace{5mm}j=k,l,\\
&\equiv \frac{\hat{\beta}_{jU}^{2\times2}}{\hat{D}_{jU}^{2\times2}},\hspace{5mm}j=k,l.\notag
\end{align}\par
Second, we can construct the Wald-DID estimator, leveraging variation in the timing of the initial exposure to the instrument between exposed cohorts. Consider an early exposed cohort $k$ and a middle exposed cohort $l$. Before period $l$, the early exposed cohort $k$ is already exposed to the instrument, while the middle exposed cohort $l$ is not yet exposed to the instrument. In this setting, we can view that the middle exposed cohort $l$ plays the role of the control group in both the first stage and the reduced form. From this observation, we can compare the evolution of the treatment and the outcome between the early exposed cohort $k$ and middle exposed cohort $l$, exploiting the time window $MID(k,l)$ and $PRE(k)$, which we call an Exposed/Not Yet Exposed design:
\begin{align}
\label{sec4eq5}
\hat{\beta}_{IV, kl}^{2\times2,k} &\equiv \frac{\left(\bar{y}_{k}^{MID(k,l)}-\bar{y}_{k}^{PRE(k)}\right)-\left(\bar{y}_{l}^{MID(k,l)}-\bar{y}_{l}^{PRE(k)}\right)}{\left(\bar{D}_{k}^{MID(k,l)}-\bar{D}_{k}^{PRE(k)}\right)-\left(\bar{D}_{l}^{MID(k,l)}-\bar{D}_{l}^{PRE(k)}\right)}\\
&\equiv \frac{\hat{\beta}_{kl}^{2\times2,k}}{\hat{D}_{kl}^{2\times2,k}}\notag.
\end{align}\par
Finally, if we focus on the middle exposed cohort $l$, which changes the exposure status from being unexposed to being exposed at time $l$, we can regard the early exposed cohort $k$ as the control group after time $l$ because this cohort is already exposed to the instrument at time $l$. We can compare the evolution of the treatment and the outcome between early exposed cohort $k$ and middle exposed cohort $l$, exploiting the time window $MID(k,l)$ and $POST(l)$, which we call an Exposed/Exposed Shift design:
\begin{align}
\label{sec4eq6}
\hat{\beta}_{IV, kl}^{2\times2,l} &\equiv \frac{\left(\bar{y}_{l}^{POST(l)}-\bar{y}_{l}^{MID(k,l)}\right)-\left(\bar{y}_{k}^{POST(l)}-\bar{y}_{k}^{MID(k,l)}\right)}{\left(\bar{D}_{l}^{POST(l)}-\bar{D}_{l}^{MID(k,l)}\right)-\left(\bar{D}_{k}^{POST(l)}-\bar{D}_{k}^{MID(k,l)}\right)}\\
&\equiv \frac{\hat{\beta}_{kl}^{2\times2,l}}{\hat{D}_{kl}^{2\times2,l}}.\notag
\end{align}\par
In each type of the DID-IV design, we have three sources of variation. First, each design exploits the subsample from all $NT$ observations. The Unexposed/Exposed DID-IV design in \eqref{sec4eq4} uses two cohorts and all time periods, indicating that the relative sample share is $n_k+n_u$. The Exposed/Not Yet Exposed DID-IV design in \eqref{sec4eq5} uses two cohorts but exploits only the time periods before period $l$, so the relative sample share is $(1-\bar{Z}_l)(n_k+n_l)$. The Exposed/Exposed Shift DID-IV design in \eqref{sec4eq6} uses two cohorts but exploits only the time periods after period $k$, so the relative sample share is $\bar{Z}_k(n_k+n_l)$.\par
Second, the variation in each type of the DID-IV design partly comes from the variation of the instrument in its subsample. It is equal to the variance of the double demeaning variable $\Tilde{Z}_{i,t}$ in each design:
\begin{align}
\label{sec4eq7}
&\hat{V}_{jU}^Z \equiv n_{jU}(1-n_{jU})\bar{Z}_j(1-\bar{Z}_j),\hspace{5mm}j=k,l,\\
\label{sec4eq8}
&\hat{V}_{kl}^{Z,k} \equiv n_{kl}(1-n_{kl})\left(\frac{\bar{Z}_{k}-\bar{Z}_{l}}{1-\bar{Z}_{l}}\right)\left(\frac{1-\bar{Z}_{k}}{1-\bar{Z}_{l}}\right),\\
\label{sec4eq9}
&\hat{V}_{kl}^{Z,l} \equiv n_{kl}(1-n_{kl})\left(\frac{\bar{Z}_{l}}{\bar{Z}_{k}}\right)\left(\frac{\bar{Z}_{k}-\bar{Z}_{l}}{\bar{Z}_{k}}\right),
\end{align}
where the $\hat{V}_{jU}^Z$, $\hat{V}_{kl}^{Z,k}$ and $\hat{V}_{kl}^{Z,l}$ represent the variance of the double demeaning variable $\Tilde{Z}_{i,t}$ in Unexposed/Exposed, Exposed/Not Yet Exposed, and Exposed/Exposed Shift DID-IV designs, respectively. In the staggered DID set up, \cite{Goodman-Bacon2021-ej} also describes the two variations, that is, the relative sample share and the variance of the double demeaning treatment variable in each type of the DID designs.\par
Unlike the staggered DID set up, however, each DID-IV design has an additional source of the variation; the effect of the instrument on the treatment in the first stage. This comes from the fact that each DID-IV design allows the noncompliance of receiving the treatment when units are exposed to the instrument. The amount of this variation is equal to the $2\times2$ DID estimator of the treatment in each DID-IV design:
\begin{align*}
&\hat{D}_{jU}^{2\times2}\equiv \left(\bar{D}_{j}^{POST(j)}-\bar{D}_{j}^{PRE(j)}\right)-\left(\bar{D}_{U}^{POST(j)}-\bar{D}_{U}^{PRE(j)}\right)\hspace{5mm}j=k,l,\\
&\hat{D}_{kl}^{2\times2,k}\equiv\left(\bar{D}_{k}^{MID(k,l)}-\bar{D}_{k}^{PRE(k)}\right)-\left(\bar{D}_{l}^{MID(k,l)}-\bar{D}_{l}^{PRE(k)}\right),\\
&\hat{D}_{kl}^{2\times2,l}\equiv\left(\bar{D}_{l}^{POST(l)}-\bar{D}_{l}^{MID(k,l)}\right)-\left(\bar{D}_{k}^{POST(l)}-\bar{D}_{k}^{MID(k,l)}\right).
\end{align*}\par
Note that the denominator of the TWFEIV estimator $\hat{\beta}_{IV}$ in \eqref{sec4eq3}, which we denote $\hat{C}^{D,Z}$ hereafter, measures the covariance between the instrument $Z_{i,t}$ and the treatment $D_{i,t}$ in whole samples. By some calculations (see the proof of Theorem \ref{sec3thm1} below), one can show that $\hat{C}^{D,Z}$ is equal to a weighted average of all possible $2\times2$ DID estimators of the treatment in each DID-IV design:
\begin{align*}
\hat{C}^{D,Z}=\sum_{k \neq U}\hat{w}_{kU}\hat{D}_{kU}^{2\times2}+\sum_{k \neq U}\sum_{l >k}[\hat{w}_{kl}^{k}\hat{D}_{kl}^{2\times2,k}+\hat{w}_{kl}^{l}\hat{D}_{kl}^{2\times2,l}],
\end{align*}
where the weights are:
\begin{align*}
&\hat{w}_{kU}=(n_k+n_u)^2\hat{V}_{kU}^{Z},\\
&\hat{w}_{kl}^{k}=((n_k+n_l)(1-\bar{Z}_l))^2\hat{V}_{kl}^{Z,k},\\
&\hat{w}_{kl}^{l}=((n_k+n_l)\bar{Z}_k)^2\hat{V}_{kl}^{Z,l}.
\end{align*}
Hereafter, we refer to $\hat{w}_{kU}$, $\hat{w}_{kl}^{k}$, and $\hat{w}_{kl}^{l}$ as the first stage weights. This decomposition result for $\hat{C}^{D,Z}$ is almost identical to that of \cite{Goodman-Bacon2021-ej} for the TWFE estimator under staggered DID designs, but the slight difference here is that each weight is not scaled by the variance of the double demeaning variable $\Tilde{Z}_{it}$ in whole samples.\par
We now present the decomposition theorem for the TWFEIV estimator under the staggered assignment of the instrument across units. Theorem \ref{sec3thm1} below is a generalization of the decomposition result for the TWFE estimator with settings of the staggered assignment of the treatment across units in \cite{Goodman-Bacon2021-ej}.\par
\begin{Theorem}[Instrumented Difference-in-Differences Decomposition Theorem]
\label{sec3thm1}
Suppose that there exist $K$ cohorts, $e=1,\dots,k,\dots,K$. The data may also contain a never exposed cohort $U$. Then, the two-way fixed effects instrumental variable estimator $\hat{\beta}_{IV}$ in \eqref{sec4eq3} is a weighted average of all possible $2 \times 2$ Wald-DID estimators.
\begin{align*}
\hat{\beta}_{IV}=\bigg[\sum_{k \neq U}\hat{w}_{IV, kU}\hat{\beta}_{IV, kU}^{2\times2}+\sum_{k \neq U}\sum_{l >k}\hat{w}_{IV, kl}^{k}\hat{\beta}_{IV, kl}^{2\times2,k}+\hat{w}_{IV, kl}^{l}\hat{\beta}_{IV, kl}^{2\times2,l}\bigg].
\end{align*}
The $2\times2$ Wald-DID estimators are:
\begin{align*}
&\hat{\beta}_{IV, kU}^{2\times2} \equiv \frac{\left(\bar{y}_{k}^{POST(k)}-\bar{y}_{k}^{PRE(k)}\right)-\left(\bar{y}_{U}^{POST(k)}-\bar{y}_{U}^{PRE(k)}\right)}{\left(\bar{D}_{k}^{POST(k)}-\bar{D}_{k}^{PRE(k)}\right)-\left(\bar{D}_{U}^{POST(k)}-\bar{D}_{U}^{PRE(k)}\right)},\\
&\hat{\beta}_{IV, kl}^{2\times2,k} \equiv \frac{\left(\bar{y}_{k}^{MID(k,l)}-\bar{y}_{k}^{PRE(k)}\right)-\left(\bar{y}_{l}^{MID(k,l)}-\bar{y}_{l}^{PRE(k)}\right)}{\left(\bar{D}_{k}^{MID(k,l)}-\bar{D}_{k}^{PRE(k)}\right)-\left(\bar{D}_{l}^{MID(k,l)}-\bar{D}_{l}^{PRE(k)}\right)},\\
&\hat{\beta}_{IV, kl}^{2\times2,l} \equiv \frac{\left(\bar{y}_{l}^{POST(l)}-\bar{y}_{l}^{MID(k,l)}\right)-\left(\bar{y}_{k}^{POST(l)}-\bar{y}_{k}^{MID(k,l)}\right)}{\left(\bar{D}_{l}^{POST(l)}-\bar{D}_{l}^{MID(k,l)}\right)-\left(\bar{D}_{k}^{POST(l)}-\bar{D}_{k}^{MID(k,l)}\right)}.
\end{align*}
The weights are:
\begin{align*}
&\hat{w}_{IV, kU}=\frac{\hat{w}_{kU}\hat{D}_{kU}^{2\times2}}{\hat{C}^{D,Z}}\\
&\hat{w}_{IV, kl}^{k}=\frac{\hat{w}_{kl}^{k}\hat{D}_{kl}^{2\times2,k}}{\hat{C}^{D,Z}}\\
&\hat{w}_{IV, kl}^{l}=\frac{\hat{w}_{kl}^{l}\hat{D}_{kl}^{2\times2,l}}{\hat{C}^{D,Z}}.
\end{align*}
and sum to one, that is, we have $\sum_{k \neq U}w_{IV, kU}+\sum_{k \neq U}\sum_{l >k}[w_{IV, kl}^{k}+w_{IV, kl}^{l}]=1$.
\end{Theorem}
\begin{proof}
See Appendix \ref{ApeA}.
\end{proof}
Theorem \ref{sec3thm1} shows that when the assignment of the instrument is staggered across units, the TWFEIV estimator is a weighted average of all possible $2 \times 2$ Wald-DID estimators. If there exist $K$ cohorts in the data, we have $K^2-K$ Wald-DID estimators, which come from either Exposed/Not Yet Exposed designs as in \eqref{sec4eq5} or Exposed/Exposed shift designs as in \eqref{sec4eq6}. If the data contains a never exposed cohort $U$, we have additionally $K$ Wald-DID estimators, which come from Unexposed/Exposed designs as in \eqref{sec4eq4}. If both situations occur, the TWFEIV estimator equals a weighted average of $K^2$ Wald-DID estimators.\par
The weight assigned to each Wald-DID estimator consists of three parts: the relative sample share squared, the variance of the double demeaning variable $\Tilde{Z}_{i,t}$, and the DID estimator of the treatment in each DID-IV design. The first part depends on the sample share of two cohorts and the timing of the initial exposure date. The second part reflects the variation of the instrument in the subsample, represented by \eqref{sec4eq7}-\eqref{sec4eq9}, and depends on the relative sample share between two cohorts and the timing of the initial exposure date. Finally, the remaining part reflects variation in the evolution of the treatment between the two cohorts. Note that the weight is not guaranteed to be non-negative in finite sample settings: although the first and second parts are always non-negative, the DID estimator of the treatment can be potentially negative in the data.\par
Theorem \ref{sec3thm1} also shows that if we subset the data containing only two cohorts (cohorts $k$ and $l$), the TWFEIV estimator $\beta_{IV,kl}^{2 \times 2}$ in the subsample can be written as:
\begin{align*}
\beta_{IV,kl}^{2 \times 2}=\frac{\hat{w}_{kl}^{k}\hat{D}_{kl}^{2 \times 2,k}}{\hat{w}_{kl}^{k}\hat{D}_{kl}^{2 \times 2,k}+\hat{w}_{kl}^{l}\hat{D}_{kl}^{2 \times 2,l}}\beta_{IV,kl}^{2 \times 2,k}+\frac{\hat{w}_{kl}^{l}\hat{D}_{kl}^{2 \times 2,l}}{\hat{w}_{kl}^{k}\hat{D}_{kl}^{2 \times 2,k}+\hat{w}_{kl}^{l}\hat{D}_{kl}^{2 \times 2,l}}\beta_{IV,kl}^{2 \times 2,l}.
\end{align*}
The TWFEIV estimator $\beta_{IV,kl}^{2 \times 2}$ is a weighted average of the Wald-DID estimators which come from either Exposed/Not Yet Exposed design or Exposed/Exposed Shift design, and the weight assigned to each Wald-DID estimator reflects the first stage weight and the DID estimator of the treatment in each DID-IV design.\par
To make the DID-IV decomposition theorem concrete, we provide a simple numerical example. Suppose we have three cohorts with equal sample size, as shown in Figure \ref{Figure1}. In this figure, we set an early exposed period $k$ and a middle exposed period $l$ such that $\bar{Z}_k=0.67$ and $\bar{Z}_l=0.21$. We assume that the effect of the instrument on the treatment is $0.15$ in cohort $k$ and $0.1$ in cohort $l$ over time. This means that the units in cohort $k$ are more induced to the treatment by the instrument than those in cohort $l$ and the effects are stable in both cohorts. The DID estimates of the treatment are $\{\hat{D}_{kU}^{2 \times 2},\hat{D}_{lU}^{2 \times 2},\hat{D}_{kl}^{2 \times 2,k},\hat{D}_{kl}^{2 \times 2,l}\}=\{0.15,0.1,0.15,0.1\}$. We also assume that the effect of the instrument on the outcome through treatment is $9$ in cohort $k$ and $10$ in cohort $l$ over time. The DID estimates of the outcome are $\{\hat{Y}_{kU}^{2 \times 2},\hat{Y}_{lU}^{2 \times 2},\hat{Y}_{kl}^{2 \times 2,k},\hat{Y}_{kl}^{2 \times 2,l}\}=\{9,10,9,10\}$. Dividing the DID estimate of the treatment by the DID estimate of the outcome yields the Wald-DID estimate: $\{\hat{\beta}_{kU}^{2 \times 2},\hat{\beta}_{lU}^{2 \times 2},\hat{\beta}_{kl}^{2 \times 2,k},\hat{\beta}_{kl}^{2 \times 2,l}\}=\{60,100,60,100\}$. The Wald-DID estimate is larger in cohort $l$ than that of cohort $k$, though as we already noted, the effect of the instrument on the treatment is larger in cohort $k$ than that of cohort $l$.\par
The DID estimates of the treatment and the exposure timing determine the amount of the weight assigned to each Wald-DID estimate, holding the sample size equal across cohorts. In the above setting, the resulting weights are $\{\hat{w}_{IV,kU},\hat{w}_{IV,lU},\hat{w}_{IV,kl}^{k},\hat{w}_{IV,kl}^{l}\}=\{0.28,0.12,0.40,0.20\}$. In Unexposed/Exposed designs, we have $\hat{w}_{IV, kU} > \hat{w}_{IV, lU}$ for two reasons. First, the DID estimate of the treatment is larger in cohort $k$ than that of cohort $l$, that is, we have $\hat{D}_{kU}^{2 \times 2}=0.15 > 0.1=\hat{D}_{lU}^{2 \times 2}$. Second, the time period $k$ is closer to the middle in the whole period than the time period $l$, that is, we have $\bar{Z}_k(1-\bar{Z}_k)=0.22 > 0.17=\bar{Z}_l(1-\bar{Z}_l)$, which implies $\hat{w}_{kU} > \hat{w}_{lU}$ in the first stage weight. By the similar argument, we have $\hat{w}_{IV, kl}^{k} > \hat{w}_{IV, kl}^{l}$ between Exposed/Not Yet Exposed and Exposed/Exposed Shift designs: we have $\hat{D}_{kl}^{2 \times 2,k}=0.15 > 0.1=\hat{D}_{kl}^{2 \times 2,l}$ and $\hat{w}_{kl}^{k} > \hat{w}_{kl}^{l}$ in the first stage weight. If the DID estimates of the treatment are equal between the two designs, the exposure timing matters: we have $\hat{w}_{IV,kU} < \hat{w}_{IV, kl}^{k}$ and $\hat{w}_{IV,lU} < \hat{w}_{IV,kl}^{l}$. The DD estimates are the same in each comparison, that is, we have $\hat{D}_{kU}^{2 \times 2}=\hat{D}_{kl}^{2 \times 2,k}$ and $\hat{D}_{lU}^{2 \times 2}=\hat{D}_{kl}^{2 \times 2,l}$. However, the different initial exposure date yields different weights in the first stage, that is, we have $\hat{w}_{kU}<\hat{w}_{kl}^{k}$ and $\hat{w}_{lU}<\hat{w}_{kl}^{l}$, which make the difference above the two comparisons.\par
In this numerical example, the simple average of the Wald-DID estimates is $80$ and the weighted average is $100 \times \frac{3}{5}+60 \times \frac{2}{5}=84$ where the weight assigned to the Wald-DID estimate reflects the relative amount of the DID estimate of the treatment. The TWFEIV estimate, however, is $\hat{\beta}_{IV}=60 \times (0.28+0.40)+100 \times (0.12+0.20)=72.8$ because it assigns more weights on the smaller Wald-DID estimate.\par
Theorem \ref{sec3thm1} is a decomposition result for the TWFEIV estimator and not for the estimand. Related to the work in this paper, \cite{De_Chaisemartin2020-dw} decompose the TWFE estimand and present the issue regarding the use of this estimand under DID designs: some weights assigned to the causal parameters in this estimand can be potentially negative. In their appendix, the authors also decompose the TWFEIV estimand, and refer to the negative weight problem in this estimand. Specifically, they apply their decomposition theorem for the TWFE estimand to the numerator and the denominator of the TWFEIV estimand respectively, and conclude that this estimand identifies the local average treatment effect as in \cite{Imbens1994-qy} only if the effects of the instrument on the treatment and the outcome are homogeneous across groups and over time. In fact, the population coefficients on the instrument in the first stage and the reduced form regressions take the form of the TWFE estimand and their decomposition theorem for the TWFE estimand is also applicable to the analysis of the TWFEIV estimand. However, the way of their decomposition for the TWFEIV estimand has some drawbacks. First, they do not formally state the target parameter and identifying assumptions in DID-IV designs. Second, their decomposition for the TWFEIV estimand is not based on the target parameter in DID-IV designs. Finally, the sufficient conditions for this estimand to have its causal interpretation are not well explored.\par
In the following section, we explore the causal interpretation of the TWFEIV estimand under staggered DID-IV designs. In section \ref{sec3}, we first define the target parameter and identifying assumptions in staggered DID-IV designs. In section \ref{sec4}, based on the decomposition theorem for the TWFEIV estimator, we then provide the causal interpretation of the TWFEIV estimand under staggered DID-IV designs. Finally, we investigate the sufficient conditions for this estimand to attain its causal interpretation under staggered DID-IV designs.

\section{Staggered instrumented difference-in-differences}\label{sec3}
In this section, we formalize the staggered instrumented difference-in-differences (DID-IV), built on the recent work in \cite{Miyaji2023}. We first introduce the additional notation. We then define the target parameter and identifying assumptions in staggered DID-IV designs. 
\subsection{Notation}\label{sec3.1}
First, we introduce the potential outcomes framework. Let $Y_{i,t}(d,z)$ denote the potential outcome in period $t$ when unit $i$ receives the treatment path $d \in \mathcal{S}(D)$ and the instrument path $z \in \mathcal{S}(Z)$. Similarly, let $D_{i,t}(z)$ denote the potential treatment status in period $t$ when unit $i$ receives the instrument path $z \in \mathcal{S}(Z)$.\par
Assumption \ref{sec2as1} allows us to rewrite $D_{i,t}(z)$ by the initial adoption date $E_i=e$. Let $D_{i,t}^{e}$ denote the potential treatment status in period $t$ if unit $i$ is first exposed to the instrument in period $e$. Let $D_{i,t}^{\infty}$ denote the potential treatment status in period $t$ if unit $i$ is never exposed to the instrument. Hereafter, we call $D_{i,t}^{\infty}$ the "never exposed treatment". Since the adoption date of the instrument uniquely pins down one's instrument path, we can write the observed treatment status $D_{i,t}$ for unit $i$ at time $t$ as
\begin{align*}
D_{i,t}=D_{i,t}^{\infty}+\sum_{1 \leq e \leq T}(D_{i,t}^{e}-D_{i,t}^{\infty})\cdot \mathbf{1}\{E_i=e\}.
\end{align*}\par
We define $D_{i,t}-D_{i,t}^{\infty}$ to be the effect of the instrument on the treatment for unit $i$ at time $t$, which is the difference between the observed treatment status $D_{i,t}$ to the never exposed treatment status $D_{i,t}^{\infty}$. Hereafter, we refer to $D_{i,t}-D_{i,t}^{\infty}$ as the individual exposed effect in the first stage. In the DID literature, \cite{Callaway2021-wl} and \cite{Sun2021-rp} define the effect of the treatment on the outcome in the same fashion.\par
Next, we introduce the group variable which describes the type of unit $i$ at time $t$,  based on the reaction of potential treatment choices at time $t$ to the instrument path $z$. Let $G_{i,e,t} \equiv (D_{i,t}^{\infty},D_{i,t}^{e}) (t \geq e)$ be the group variable at time $t$ for unit $i$ and the initial exposure date $e$. Specifically, the first element $D_{i,t}^{\infty}$ represents the treatment status at time $t$ if unit $i$ is never exposed to the instrument $E_i=\infty$ and the second element $D_{i,t}^{e}$ represents the treatment status at time $t$ if unit $i$ starts exposed to the instrument at $E_i=e$. Following to the terminology in \cite{Imbens1994-qy}, we define $G_{i,e,t}=(0,0) \equiv NT_{e,t}$ to be the never-takers, $G_{i,e,t}=(1,1) \equiv AT_{e,t}$ to be the always-takers, $G_{i,e,t}=(0,1) \equiv CM_{e,t}$ to be the compliers and $G_{i,e,t}=(1,0) \equiv DF_{e,t}$ to be the defiers at time $t$ and the initial exposure date $e$.\par
Finally, we make a no carryover assumption on potential outcomes $Y_{i,t}(d,z)$. 
\begin{Assumption}[No carryover assumption]
\label{sec2as2}
\begin{align*}
\forall z \in \mathcal{S}(Z), \forall d\in \mathcal{S}(D), \forall t\in \{1,\dots,T\},Y_{i,t}(d,z)=Y_{i,t}(d_t,z),
\end{align*}
where $d=(d_1,\dots,d_T)$ is the generic element of the treatment path $D_i$.
\end{Assumption}
This assumption requires that potential outcomes $Y_{i,t}(d,z)$ depend only on the current treatment status $d_t$ and the instrument path $z$. In the DID literature, several recent papers impose this assumption with settings of a non-staggered treatment; see, e.g., \cite{De_Chaisemartin2020-dw} and \cite{Imai2021-dn}. Although it can be possible to weaken this assumption by introducing the treatment path $d$ in potential outcomes $Y_{i,t}(d,z)$, this requires the cumbersome notation and complicates the definition of our target parameter, thus is beyond the scope of this paper. 
\par
Henceforth, we keep Assumption \ref{sec2as1} and \ref{sec2as2}. In the next section, we define the target parameter in staggered DID-IV designs.

\subsection{Target parameter in staggered DID-IV designs}\label{sec3.2}
Our target parameter in staggered DID-IV designs is the cohort specific local average treatment effect on the treated (CLATT) defined below.
\begin{Def}
The cohort specific local average treatment effect on the treated (CLATT) at a given relative period $l$ from the initial adoption of the instrument is 
\begin{align*}
CLATT_{e,l}&=E[Y_{i,e+l}(1)-Y_{i,e+l}(0)|E_i=e, D_{i,e+l}^{e} > D_{i,e+l}^{\infty}]\\
&=E[Y_{i,e+l}(1)-Y_{i,e+l}(0)|E_i=e,CM_{e,e+l}].
\end{align*}
\end{Def}
This parameter measures the treatment effects at a given relative period $l$ from the initial instrument adoption date $E_i=e$, for those who belong to cohort $e$, and are the compliers $CM_{e,e+l}$, that is, who are induced to treatment by instrument at time $e+l$. Each CLATT$_{e,l}$ can potentially vary across cohorts and over time, as it depends on cohort $e$, relative period $l$, and the compliers $CM_{e,e+l}$.\par

\subsection{Identifying assumptions in staggered DID-IV designs}\label{sec3.3}
In this section, we state the identifying assumptions in staggered DID-IV designs based on \cite{Miyaji2023}.\par
\begin{Assumption}[Exclusion Restriction in multiple time periods]
\label{sec2as3}
\begin{align*}
\forall z \in \mathcal{S}(Z),\forall d_t \in \mathcal{S}(D_t),\forall t \in \{1,\dots,T\}, Y_{i,t}(d,z)=Y_{i,t}(d)\hspace{3mm}a.s.
\end{align*}
\end{Assumption}
Assumption \ref{sec2as3} requires that the path of the instrument does not directly affect the potential outcome for all time periods and its effects are only through treatment. Given Assumption \ref{sec2as2} and Assumption \ref{sec2as3}, we can write the potential outcome $Y_{i,t}(d,z)$ as $Y_{i,t}(d_t)=D_{i,t}Y_{i,t}(1)+(1-D_{i,t})Y_{i,t}(0)$.\par 
Here, we introduce the potential outcomes at time $t$ if unit $i$ is assigned to the instrument path $z \in \mathcal{S}(Z)$: 
\begin{align*}
Y_{i,t}(D_{i,t}(z)) \equiv D_{i,t}(z)Y_{i,t}(1)+(1-D_{i,t}(z))Y_{i,t}(0).
\end{align*}
Since the exposure timing $E_i$ completely determines the path of the instrument, we can write the potential outcomes for cohort $e$ and cohort $\infty$ as $Y_{i,t}(D_{i,t}^{e})$ and $Y_{i,t}(D_{i,t}^{\infty})$, respectively. The potential outcome $Y_{i,t}(D_{i,t}^{e})$ represents the outcome status at time $t$ if unit $i$ is first exposed to the instrument at time $e$ and the potential outcome $Y_{i,t}(D_{i,t}^{\infty})$ represents the outcome status at time $t$ if unit $i$ is never exposed to the instrument. Hereafter, we refer to $Y_{i,t}(D_{i,t}^{\infty})$ as the "never exposed outcome".\par
\begin{Assumption}[Monotonicity Assumption in multiple time periods]
\label{sec2as4}
\begin{align*}
Pr(D_{i,e+l}^{e} \geq D_{i,e+l}^{\infty})=1\hspace{2mm}\text{or}\hspace{2mm}Pr(D_{i,e+l}^{e} \leq D_{i,e+l}^{\infty})=1\hspace{2mm}\text{for all}\hspace{2mm}e\in \mathcal{S}(E_i)\hspace{2mm}\text{and for all}\hspace{2mm}l \geq 0.
\end{align*}
\end{Assumption}
This assumption requires that the instrument path affects the treatment adoption behavior in a monotone way for all relative periods after the initial exposure. Recall that we define $D_{i,t}-D_{i,t}^{\infty}$ to be the effect of the instrument on the treatment for unit $i$ at time $t$. Assumption \ref{sec2as4} requires that the individual exposed effect in the first stage is non-negative (or non-positive) for all $i$ and all the time periods after the initial exposure. This assumption implies that the group variable $G_{i,e,t} \equiv (D_{i,t}^{\infty},D_{i,t}^{e})$ can take three values with non-zero probability for all $e$ and all $t \geq e$. Hereafter, we consider the type of the monotonicity assumption that rules out the existence of the defiers $DF_{e,t}$ for all $t \geq e$ in any cohort $e$.\par
\begin{Assumption}[No anticipation in the first stage]
\label{sec2as5}
\begin{align*}
D_{i,e+l}^{e}=D_{i,e+l}^{\infty}\hspace{2mm}a.s.\hspace{3mm}\text{for all units $i$,}\hspace{3mm}\text{for all}\hspace{2mm}e\in \mathcal{S}(E_i)\hspace{2mm}\text{and for all}\hspace{2mm}l<0.
\end{align*}
\end{Assumption}
Assumption \ref{sec2as5} requires that the potential treatment choice for the treatment in any $l$ period before the initial exposure to the instrument is equal to the never exposed treatment. This assumption restricts the anticipatory behavior before the initial exposure in the first stage.\par
\begin{Assumption}[Parallel Trends Assumption in the treatment in multiple time periods]
\label{sec2as6}
\begin{align*}
\hspace{2mm}\text{For all}\hspace{2mm}s \neq t,\hspace{2mm} E[D_{i,t}^{\infty}-D_{i,s}^{\infty}|E_i=e]\hspace{3mm}\text{is same for all}\hspace{2mm}e\in \mathcal{S}(E_i).
\end{align*}
\end{Assumption}
Assumption \ref{sec2as6} is a parallel trends assumption in the treatment in multiple periods and multiple cohorts. This assumption requires that the trends of the treatment across cohorts would have followed the same path, on average, if there is no exposure to the instrument. Assumption \ref{sec2as6} is analogous to that of \cite{Callaway2021-wl} and \cite{Sun2021-rp} in DID designs: both papers impose the same type of the parallel trends assumption on untreated outcomes with settings of multiple periods and multiple cohorts. 
\begin{Assumption}[Parallel Trends Assumption in the outcome in multiple time periods]
\label{sec2as7}
\begin{align*}
\hspace{2mm}\text{For all}\hspace{2mm}s < t,\hspace{2mm} E[Y_{i,t}(D_{i,t}^{\infty})-Y_{i,s}(D_{i,s}^{\infty})|E_i=e]\hspace{3mm}\text{is same for all}\hspace{2mm}e\in \mathcal{S}(E_i).
\end{align*}
\end{Assumption}
Assumption \ref{sec2as7} is a parallel trends assumption in the outcome with settings of multiple periods and multiple cohorts. This assumption requires that the expectation of the never exposed outcome across cohorts would have followed the same evolution if the assignment of the instrument had not occurred. From the discussions in \cite{Miyaji2023}, we can interpret that this assumption requires the same expected time gain across cohorts and over time: the effects of time on outcome through treatment are the same on average across cohorts and over time.\par

\section{Causal interpretation of the TWFEIV estimand}\label{sec4}
In this section, we explore the causal interpretation of the TWFEIV estimand under staggered DID-IV designs. In section \ref{sec4.1}, we first define the main building block parameter in the first stage and reduced form regressions, respectively. In section \ref{sec4.2}, we then interpret the TWFEIV estimand under staggered DID-IV designs, and show that this estimand potentially fails to summarize the treatment effects. In section \ref{sec4.3}, given the negative result of using the TWFEIV estimand under staggered DID-IV designs, we describe the various restrictions on main building block parameter in each stage regression. In section \ref{sec4.4}, as a preparation, we then describe the causal interpretation of the denominator in the TWFEIV estimand under these restrictions. In section \ref{sec4.5}, we finally investigate the sufficient conditions for the TWFEIV estimand to attain its causal interpretation.

\subsection{Main building block parameter in each stage regression}\label{sec4.1}
As we already mentioned in section \ref{sec2}, the TWFEIV regression employs the TWFEIV regression twice in the first stage and reduced form regressions. In this section, we define the main building block parameter in each stage regression.\par 
In the first stage regression, our building block parameter is the average of individual exposed effect at a given relative period $l$ from the initial exposure to the instrument in cohort $e$. We call this the cohort specific average exposed effect on the treated in the first stage (CAET$_{e,l}^{1}$) defined below.
\begin{Def}
The cohort specific average exposed effect on the treated in the first stage (CAET$^{1}$) at a given relative period $l$ from the initial adoption of the instrument is 
\begin{align*}
CAET_{e,l}^{1}=E[D_{i,e+l}-D_{i,e+l}^{\infty}|E_i=e].
\end{align*}
\end{Def}
We use the superscript $1$ to make it clear that we define this parameter for the first stage regression. In the recent DID literature, \cite{Sun2021-rp} define their main building block parameter in staggered DID designs in a similar fashion and call it the cohort specific average treatment effect on the treated. \cite{Callaway2021-wl} call the same parameter the group-time average treatment effect.\par
If the treatment is binary and monotonicity assumption (Assumption \ref{sec2as4}) holds, the CAET$_{e,l}^{1}$ is equal to the share of the compliers CM$_{e,e+l}$ in cohort $e$ at period $e+l$:
\begin{align*}
CAET^{1}_{e,l}&=E[D_{i,e+l}^{e}-D_{i,e+l}^{\infty}|E_i=e]\\
&=Pr(CM_{e,e+l}|E_i=e).
\end{align*}\par
In the reduced form regression, our building block parameter is the average of individual effect of the instrument on the outcome through treatment at a given relative period $l$ from the initial exposure to the instrument in cohort $e$. We call this the cohort specific average intention to exposed effect on the treated in the reduced form (CAIET$_{e,l}$) defined below.
\begin{Def}
The cohort specific average intention to exposed effect on the treated in the reduced form (CAIET) at a given relative period $l$ from the initial adoption of the instrument is 
\begin{align*}
CAIET_{e,l}=E[Y_{i,e+l}(D_{i,e+l})-Y_{i,e+l}(D_{i,e+l}^{\infty})|E_i=e].
\end{align*}
\end{Def}
If we assume the identifying assumptions in staggered DID-IV designs (Assumptions \ref{sec2as1} to \ref{sec2as7}), this parameter is equal to a product of the CLATT$_{e,l}$ and CAET$_{e,l}^{1}$:
\begin{align}
CAIET_{e,l}&=E[Y_{i,e+l}(D^{e}_{i,e+l})-Y_{i,e+l}(D_{i,e+l}^{\infty})|E_i=e]\notag\\
&=E[(D^{e}_{i,e+l}-D_{i,e+l}^{\infty})(Y_{i,e+l}(1)-Y_{i,e+l}(0))|E_i=e]\notag\\
&=E[Y_{i,e+l}(1)-Y_{i,e+l}(0)|E_i=e,CM_{e,e+l}]\cdot Pr(CM_{e,e+l}|E_i=e)\notag\\
\label{sec4.1eq1}
&=CLATT_{e,l}\cdot CAET_{e,l}^{1}.
\end{align}
In other words, if we scale the CAIET$_{e,l}$ in the reduced form by the CAET$_{e,l}^{1}$ in the first stage, we obtain the CLATT$_{e,l}$, which is the reason why we call this the cohort specific average "intention to exposed effect" on the treated in the reduced form.\par
\subsection{Interpreting the TWFEIV estimand under staggered DID-IV designs}\label{sec4.2}
We now interpret the TWFEIV estimand under staggered DID-IV designs based on the DID-IV decomposition theorem derived in section \ref{sec2} and the main building block parameters defined in the previous section. This section presumes the monotonicity assumption (Assumption \ref{sec2as4}) to clarify the interpretation of each notation defined below.\par
First, we introduce the additional notation. Let CLATT$^{CM}_k(W)$ denote a weighted average of each CLATT$_{k,t}$ in the time window $W$ (with $T_W$ periods) where the weight reflects the relative amount of the exposed effect in the first stage in cohort $k$ at period $t$:
\begin{align*}
CLATT^{CM}_k(W) &\equiv \sum_{t \in W}\frac{CAET^{1}_{k,t}}{\sum_{t \in W}CAET^{1}_{k,t}}CLATT_{k,t}\\
&=\sum_{t \in W}\frac{Pr(CM_{k,t}|E_i=k)}{\sum_{t \in W}Pr(CM_{k,t}|E_i=k)}CLATT_{k,t}.
\end{align*}
The first equality holds because we have a binary treatment and assume the monotonicity assumption (Assumption \ref{sec2as4}). Each weight assigned to each $CLATT_{k,t}$ reflects the relative share of the compliers at period $t$ in cohort $k$ during the time window $W$. We call this the compliers weighted scheme. This would be one of the reasonable weighting schemes for two reasons. First, the weight is designed to be larger in the period when the proportion of the compliers is higher in cohort $k$. Second, the sum of the weight is one by construction: the proportion of the compliers in each period in cohort $k$ is divided by the total amount of the compliers in the time window $W$ in cohort $k$.\par 
We also define the similar notation $CLATT_k(W)$, in which the proportion of the compliers in cohort $k$ at period $t$ is divided by the time length $T_{W}$:
\begin{align*}
CLATT_k(W) &\equiv \frac{1}{T_W}\sum_{t \in W}CAET^{1}_{k,t}CLATT_{k,t}\\
&=\frac{1}{T_W}\sum_{t \in W}Pr(CM_{k,t}|E_i=k)CLATT_{k,t}.
\end{align*}
We call this the time-corrected weighting scheme. In contrast to $CLATT^{CM}_k(W)$, the weight assigned to each $CLATT_{k,t}$ can be inappropriate: each weight does not reflect the relative share of the compliers in cohort $k$ at period $t$. In addition, the sum of each weight is not equal to one in general.\par
Theorem \ref{sec4thm2} below shows the probability limit of the TWFEIV estimator $\hat{\beta}_{IV}$ under staggered DID-IV designs (Assumptions \ref{sec2as1}-\ref{sec2as7}).\par
\begin{Theorem}
\label{sec4thm2}
Suppose Assumptions \ref{sec2as1}-\ref{sec2as7} hold. Then, the TWFEIV estimand $\beta_{IV}$ consists of two terms:
\begin{align*}
\hat{\beta}_{IV}&=\bigg[\sum_{k \neq U}\hat{w}_{IV, kU}\hat{\beta}_{IV, kU}^{2\times2}+\sum_{k \neq U}\sum_{l >k}\hat{w}_{IV, kl}^{k}\hat{\beta}_{IV, kl}^{2\times2,k}+\hat{w}_{IV, kl}^{l}\hat{\beta}_{IV, kl}^{2\times2,l}\bigg]\\
&\xrightarrow{p} WCLATT-\Delta CLATT.
\end{align*}
where we define:
\begin{align*}
WCLATT &\equiv \sum_{k \neq U}w_{IV,kU}CLATT^{CM}_{k}(POST(k))+\sum_{k \neq U}\sum_{l >k}w_{IV,kl}^{k}CLATT^{CM}_{k}(MID(k,l))\\
&+\sum_{k \neq U}\sum_{l >k}\sigma_{IV,kl}^{l}\cdot CLATT_l(POST(l))\\
\Delta CLATT &\equiv \sum_{k \neq U}\sum_{l >k}\sigma_{IV,kl}^{l}\cdot \left[CLATT_k(POST(l))-CLATT_k(MID(k,l))\right].
\end{align*}
The weights $w_{IV,kU}$ and $w_{IV,kl}^{k}$ are the probability limit of $\hat{w}_{IV,kU}$ and $\hat{w}_{IV,kl}^{k}$, respectively. The weight $\sigma_{IV,kl}^{l}$ is the probability limit of $\frac{\hat{w}_{kl}^{l}}{\hat{C}^{D,Z}} \neq \frac{\hat{w}_{kl}^{l}\hat{D}_{kl}^{2\times2,l}}{\hat{C}^{D,Z}}=\hat{w}_{IV,kl}^{l}$. The specific expressions for each weight are shown in equations \eqref{sec2thm2weight1}, \eqref{sec2thm2weight2}, and \eqref{sec2thm2weight3} in Appendix \ref{ApeB}.
\end{Theorem}\par
\begin{proof}
See Appendix \ref{ApeB}.
\end{proof}\par
Theorem \ref{sec4thm2} shows that the TWFEIV estimand $\beta_{IV}$ consists of two terms ($WCLATT$ and $\Delta CLATT$) and potentially fails to aggregate the treatment effects under staggered DID-IV designs.\par
The first term $WCLATT$ is a positively weighted average of each $CLATT_{k,t}$ for the post-exposed period in cohort $k$. We call this a weighted average cohort specific local average treatment effect on the treated ($WCLATT$) parameter. The first and the second terms in the $WCLATT$ use the compliers weighted scheme, but the third term in $WCLATT$ uses the time-corrected one.\par
Although $WCLATT$ can be a causal parameter, the amount of this parameter may be difficult to interpret in practice for two reasons. First, the weight $\sigma_{IV,kl}^{l}$ assigned to $ CLATT_l(POST(l))$ reflects only the sample share and the variation of the instrument, and does not reflect the variation of the treatment $D_{kl}^{2\times2,l}$ in the first stage. Because the other weights, $w_{IV,kU}$ and $w_{IV,kl}^{k}$ precisely reflect all the variations in each DID-IV design, this asymmetry can break the implication of the magnitude of this parameter in a given application. Second, the $CLATT_l(POST(l))$ in the third term is a weighted average of $CLATT_{k,t}$ for the post exposed periods in cohort $k$, but the weight assigned to each $CLATT_{k,t}$ seems not reasonable: it does not reflect the relative share of the compliers in period $t$ in cohort $k$ and the sum of the weight is not equal to one.\par 
The problem of the $WCLATT$ is due to the "bad comparisons" in the first stage TWFE regression: when we compare the evolution of the treatment in Exposed/Exposed Shift designs, we use already exposed cohorts as controls. In these comparisons, we should offset the DID estimator of the treatment in each weight in Exposed/Exposed Shift designs by the one appeared in the denominator of the corresponding Wald-DID estimator, which produces the weight $\sigma_{IV,kl}^{l}$ and  $CLATT_l(POST(l))$ in the third term.\par
The second term $\Delta CLATT$ is a weighted sum of the differences in the positively weighted average of each $CLATT_{k,t}$ from the exposed period $k$ to before period $l (k <l)$ and after period $l$ in the already exposed cohort $k$. This term fails to properly aggregate the treatment effects because the $CLATT_k(POST(l))$ is canceled out by the $CLATT_k(MID(k,l))$ in each cohort $k$. This problem arises due to the "bad comparisons" in the reduced form TWFE regression: when we compare the evolution of the outcome in Exposed/Exposed Shift designs, we use already exposed cohorts as controls. In these comparisons, we subtract their expected trends of unexposed potential outcomes and average intention exposed effects, which yields the $\Delta CLATT$.\par
Overall, this section shows that the TWFEIV estimand potentially fails to summarize the treatment effects under staggered DID-IV designs. In the next section, we first describe various restrictions on main building block parameters in the first stage and the reduced form regressions. Given these restrictions on exposed effect heterogeneity, we then explore the sufficient conditions for the TWFEIV estimand to be causally interpretable parameter.

\subsection{Restrictions on exposed effect heterogeneity}\label{sec4.3}
First, we describe the restrictions on the CAET$^{1}_{e,l}$ in the first stage regression. 
\begin{Assumption}[Exposed effect homogeneity across cohorts in the first stage]
\label{sec4as1}
For each relative period $l$, $CAET^{1}_{e,l}$ does not depend on cohort $e$ and is equal to $AET_{l}^{1}$.
\end{Assumption}
Assumption \ref{sec4as1} requires that the exposed effects in the first stage depend on only the relative time period $l$ after the initial exposure to the instrument and do not depend on the cohort $e$. This assumption does not exclude the dynamic effects of the instrument on the treatment, but requires that the exposed effects are the same across cohorts for all relative periods.

\begin{Assumption}[Stable exposed effect over time within cohort in the first stage]
\label{sec4as2}
For each cohort $e$, $CAET^{1}_{e,l}$ does not depend on the relative time period $l$ and is equal to $CAET^{1}_{e}$.
\end{Assumption}
Assumption \ref{sec4as2} rules out the dynamic effects of the instrument on the treatment within cohort $e$ in the first stage regression. Assumption \ref{sec4as2} permits the heterogeneous exposed effects across cohort $e$, but requires the homogeneous exposed effects over time after the initial adoption of the instrument within cohort $e$.\par
The recent DID literature imposes the similar restrictions as in Assumption \ref{sec4as1} and Assumption \ref{sec4as2} on treatment effects. \cite{Sun2021-rp} assume that "each cohort experiences the same path of treatment effects", which is in line with Assumption \ref{sec4as1}. \cite{Goodman-Bacon2021-ej} requires heterogeneous treatment effects to either be "constant over time but vary across units" or "vary over time but not across units". The former corresponds to Assumption \ref{sec4as2} and the latter corresponds to Assumption \ref{sec4as1}.\par
Next, we describe the restrictions on the CAIET$_{e,l}$ in the reduced form regression. Following to Assumption \ref{sec4as1} and Assumption \ref{sec4as2} on the  CAET$^{1}_{e,l}$, we consider Assumption \ref{sec4as3} and Assumption \ref{sec4as4} below.
\begin{Assumption}[Exposed effect homogeneity across cohorts in the reduced form]
\label{sec4as3}
For each relative period $l$, $CAIET_{e,l}$ does not depend on cohort $e$ and is equal to $AIET_{l}$.
\end{Assumption}  
\begin{Assumption}[Stable exposed effect over time within cohort in the reduced form]
\label{sec4as4}
For each cohort $e$, $CAIET_{e,l}$ does not depend on the relative time period $l$ and is equal to $CAIET_{e}$.
\end{Assumption} 
Assumption \ref{sec4as3} requires that the evolution of the average intention to exposed effect after the initial exposure is the same across cohorts. Assumption \ref{sec4as4} requires that the average intention to exposed effects are stable over time in all relative periods within cohort $e$.\par 
Note that given Assumption \ref{sec4as1} and Assumption \ref{sec4as3}, we have the following restriction on the CLATT$_{e,l}$, which follows from equation \eqref{sec4.1eq1} in section \ref{sec4.1}.
\begin{Assumption}[Treatment effect homogeneity across cohorts for $CLATT_{e,l}$]
\label{sec4as6}
For each relative period $l$, $CLATT_{e,l}$ does not depend on cohort $e$ and is equal to $LATT_{l}$.
\end{Assumption}
Similarly, given Assumption \ref{sec4as2} and Assumption \ref{sec4as4}, we have the following restriction on the CLATT$_{e,l}$.
\begin{Assumption}[Stable treatment effect over time within cohort for $CLATT_{e,l}$]
\label{sec4as7}
For each cohort $e$, $CLATT_{e,l}$ does not depend on the relative time period $l$ and is equal to $CLATT_{e}$.
\end{Assumption}

\subsection{The denominator in the TWFEIV estimand}\label{sec4.4}
In this section, we first interpret the denominator in the TWFEIV estimand under various restrictions considered in section \ref{sec4.3}. This section is a preparation for the next section, in which we analyze the TWFEIV estimand itself.\par
As we already noted, the denominator in the TWFEIV estimator (see equation \eqref{sec4eq3}), $\hat{C}^{D,Z}$ can be decomposed into a weighted average of all possible $2 \times 2$ DID estimators of the treatment. In the following discussion, we show that this estimand can potentially fail to aggregate the effects of the instrument on the treatment in the first stage regression without additional restrictions. We then briefly describe the interpretation of this estimand by imposing Assumption \ref{sec4as1} or Assumption \ref{sec4as2}, and state the implications.\par 
First, we introduce the additional notation. Let CAET$^{1}_k(W)$ denote an equally weighted average of the CAET$^{1}_{k,t}$ in the time window $W$ (with $T_W$ period length):
\begin{align*}
CAET^{1}_k(W) &\equiv \frac{1}{T_W}\sum_{t \in W}CAET^{1}_{k,t}.
\end{align*}
If we assume Assumption \ref{sec2as4} (monotonicity assumption), CAET$^{1}_k(W)$ is an equally weighted average of the fraction of the compliers in cohort $k$ in the time window $W$. For instance, the CAET$^{1}_k(POST(k))$ is an equally weighted average of the CAET$^{1}_{k}$ during the periods after the initial exposure date $k$ and rewritten as 
\begin{align*}
CAET^{1}_k(POST(k))&=\frac{1}{T-(k-1)}\sum_{t=k}^{T}CAET_{k,t}^{1}\\
&=\frac{1}{T-(k-1)}\sum_{t=k}^{T}Pr(CM_{k,t}|E_i=k).
\end{align*}\par
Lemma \ref{sec4lemma1} below shows the probability limit of the denominator $\hat{C}^{D,Z}$ under staggered DID-IV designs. This lemma is mainly based on the result of \cite{Goodman-Bacon2021-ej}, who shows the probability limit of the two-way fixed effects estimator under staggered DID designs. The slight difference here is that each weight assigned to each CAET$^{1}_{k}(W)$ in $C^{D,Z}$ is not divided by the probability limit of the grand mean $\frac{1}{NT}\sum_{i}\sum_{t}\Tilde{Z}_{it}$.
\begin{Lemma}
\label{sec4lemma1}
Suppose Assumptions \ref{sec2as1}-\ref{sec2as7} hold. Then, the probability limit of the denominator of the TWFEIV estimator, $C^{D,Z}$ consists of two terms:
\begin{align*}
\hat{C}^{D,Z}&=\sum_{k \neq U}\hat{w}_{kU}\hat{D}_{kU}^{2\times2}+\sum_{k \neq U}\sum_{l >k}[\hat{w}_{kl}^{k}\hat{D}_{kl}^{2\times2,k}+\hat{w}_{kl}^{l}\hat{D}_{kl}^{2\times2,l}]\\
&\xrightarrow{p} WCAET-\Delta CAET^{1}.
\end{align*}
where we define:
\begin{align*}
&WCAET \equiv \sum_{k \neq U}w_{kU}CAET_{k}^1(POST(k))+\sum_{k \neq U}\sum_{l > k}w_{kl}^{k}CAET_{k}^1(MID(k,l))+w_{kl}^{l}CAET_{l}^{1}(POST(l)),\\
&\Delta CAET^{1} \equiv \sum_{k \neq U}\sum_{l > k}w_{kl}^{l}\left[CAET_{k}^{1}(POST(l))-CAET_{k}^{1}(MID(k,l))\right].
\end{align*}
The weights $w_{kU}$,$w_{kl}^{k}$ and $w_{kl}^{l}$ are the probability limit of $\hat{w}_{kU},\hat{w}_{kl}^{k}$ and $\hat{w}_{kl}^{l}$ defined in section \ref{sec3} respectively, and are non-negative. The specific expressions in each weight are shown in equations \eqref{Lemma1weight1}-\eqref{Lemma1weight3} in Appendix \ref{ApeB}.
\end{Lemma}\par
\begin{proof} See Appendix \ref{ApeB}.
\end{proof}\par
Lemma \ref{sec4lemma1} shows that we can decompose $C^{D,Z}$ into two terms. The first term is a positively weighted average of each CAET$_{k,t}^{1}$ during the periods after the initial exposure in exposed cohorts, allowing for its causal interpretation. Following the terminology in \cite{Goodman-Bacon2021-ej}, we call this a weighted average cohort specific exposed effect on the treated (WCAET) parameter.\par 
The second term $\Delta$CAET$^{1}$ is equal to the sum of the difference in the positively weighted average of exposed effect CAET$_{k,t}^{1}$ from the exposed period $k$ to before period $l$ ($k <l$) and after period $l$ in the already exposed cohort $k$. This term fails to properly aggregate the causal parameter in the first stage because some exposed effects are canceled out by other exposed effects.\par
Lemma \ref{sec4lemma1} implies that if we assume only Assumptions \ref{sec2as1}-\ref{sec2as7}, the probability limit of the denominator in the TWFEIV estimand, $C^{D,Z}$ generally fails to properly summarize the exposed effects in the first stage due to the second term $\Delta$CAET$^{1}$. This problem arises from the "bad comparisons" performed by the TWFE regression in the first stage: we treat the already exposed cohorts as control groups in the Exposed/ Exposed Shift designs. In these comparisons, we should subtract their expected trends of unexposed potential treatment choices and their expected exposed effects, which yields the second term $\Delta$CAET$^{1}$. In the DID literature, \cite{Borusyak2021-jv}, \cite{De_Chaisemartin2020-dw}, and \cite{Goodman-Bacon2021-ej} point out the same issue for the TWFE estimand in staggered DID designs.\par  
Based on the negative result shown in Lemma \ref{sec4lemma1}, we consider the restrictions on exposed effect heterogeneity in the first stage regression. The conclusion here is that $C^{D,Z}$ properly aggregates each $CAET_{k,t}^{1}$ only if Assumption \ref{sec4as2} holds, that is, the exposed effects are stable over time within cohort $e$. Because \cite{Goodman-Bacon2021-ej} have already made the same point for the TWFE estimand, we briefly summarize the interpretation of $C^{D,Z}$ under Assumption \ref{sec4as1} or Assumption \ref{sec4as2} in the following. For the more detailed discussions, see section 3.1 in \cite{Goodman-Bacon2021-ej}.
\subsubsection*{Interpreting $C^{D,Z}$ under Assumption \ref{sec4as1} only}
Even when Assumption \ref{sec4as1} holds, that is, the exposed effects are the same across cohorts but vary over time in the first stage, we have $\Delta$CAET$^{1} \neq 0$ in general. This implies that if we impose only Assumption \ref{sec4as1}, we cannot generally interpret the $C^{D,Z}$ as measuring the positively weighted average of exposed effects in the first stage.
\subsubsection*{Interpreting $C^{D,Z}$ under Assumption \ref{sec4as2} only}
If Assumption \ref{sec4as2} holds, that is, the exposed effects are stable over time within cohort $e$ in the first stage, we have $CAET_k^{1}(W)=CAET_k^{1}$. This implies that the second term $\Delta$CAET$^{1}$ is equal to zero:
\begin{align*}
\Delta CAET^{1} &=\sum_{k \neq U}\sum_{l > k}w_{kl}^{l}\left[CAET_{k}^{1}-CAET_{k}^{1}\right]\\
&=0.
\end{align*}
Thus, $C^{D,Z}$ simplifies to:
\begin{align*}
C^{D,Z}&=WCAET\\
&=\sum_{k \neq U}CAET^{1}_k\underbrace{\left[w_{kU}+\sum_{j=1}^{k-1}w_{jk}^{k}+\sum_{j=k+1}^{K}w_{kj}^{k}\right]}_{\equiv w_k}.
\end{align*}
$C^{D,Z}$ weights each CAET$^{1}_k$ positively across cohorts under Assumption \ref{sec4as2} only. We note, however, that each weight assigned to each CAET$^{1}_k$, $w_k$ is not equal to the sample share in cohort $k$, but is a function of the sample share and the timing of the initial exposure date. 
\vskip\baselineskip\par
In this section, we have considered whether the denominator in the TWFEIV estimand properly aggregates the exposed effects in the first stage. We have two implications. First, if we do not impose Assumption \ref{sec4as2}, the weight assigned to each $2 \times 2$ Wald-DID in the TWFEIV estimand may not be properly normalized because the numerator in each weight is divided by $C^{D,Z}$, and the denominator potentially fail to aggregate the exposed effects in the first stage. Second, if we do not impose Assumption \ref{sec4as2}, some weights assigned to $2 \times 2$ Wald-DID estimands can be potentially negative. This is because the DID estimand of the treatment forms the part of each weight and can be negative due to the "bad comparisons" in the first stage regression.\par
From the discussion so far, hereafter, we impose Assumption \ref{sec4as2} when we consider the restrictions on exposed effect heterogeneity in the first stage. 

\subsection{Interpreting the TWFEIV estimand under additional restrictions}\label{sec4.5}
We now describe the interpretation of the TWFEIV estimand under additional restrictions.
\subsubsection*{Interpretation under Assumption \ref{sec4as2} only}
 First, we consider imposing Assumption \ref{sec4as2} only, that is, we assume only the stable exposed effect over time in the first stage. If Assumption \ref{sec4as2} holds, the $CLATT^{CM}_k(W)$ simplifies to an equally weighted average of $CLATT_{k,t}$:
\begin{align*}
CLATT^{CM}_k(W) &=\sum_{t \in W}\frac{Pr(CM_{k,t}|E_i=k)}{\sum_{t \in W}Pr(CM_{k,t}|E_i=k)}CLATT_{k,t}\\
&=\frac{1}{T_W}\sum_{t \in W}CLATT_{k,t}\\
&\equiv CLATT^{eq}_k(W).
\end{align*}
The $CLATT^{eq}_k(W)$ weights each $CLATT_{k,t}$ equally in the time window $W$ and the weight sum to one by construction. We call this an equal weighting scheme.\par
Lemma \ref{sec4lemma2} presents the interpretation of the TWFEIV estimand under staggered DID-IV designs and Assumption \ref{sec4as2}.
\begin{Lemma}
\label{sec4lemma2}
Suppose Assumptions \ref{sec2as1}-\ref{sec2as7} hold. If Assumption \ref{sec4as2} holds additionally, the TWFEIV estimand $\beta_{IV}$ consists of two terms:
\begin{align*}
\beta_{IV}= WCLATT-\Delta CLATT.
\end{align*}
where we define:
\begin{align*}
WCLATT &\equiv \sum_{k \neq U}w_{IV,kU}CLATT^{eq}_{k}(POST(k))+\sum_{k \neq U}\sum_{l >k}w_{IV,kl}^{k}CLATT^{eq}_{k}(MID(k,l))\\
&+\sum_{k \neq U}\sum_{l >k}w_{IV, kl}^{l} CLATT^{eq}_l(POST(l)),\\
\Delta CLATT &\equiv \sum_{k \neq U}\sum_{l >k}\sigma_{IV,kl}^{l}\cdot \left[CLATT_k(POST(l))-CLATT_k(MID(k,l))\right].
\end{align*}
The weights $w_{IV,kU}$, $w_{IV,kl}^{k}$ and $w_{IV,kl}^{l}$ are the probability limit of $\hat{w}_{IV,kU}$, $\hat{w}_{IV,kl}^{k}$ and $\hat{w}_{IV,kl}^{l}$ respectively, and are non-negative. The specific expressions for these weights are shown in equations \eqref{lemma2weight1}, \eqref{lemma2weight2}, and \eqref{lemma2weight3} in Appendix \ref{ApeB}. The weight $\sigma_{IV,kl}^{l}$ is already defined in Theorem \ref{sec4thm2}.
\end{Lemma}\par
\begin{proof}
See Appendix \ref{ApeB}.
\end{proof}\par
Lemma \ref{sec4lemma2} shows that Assumption \ref{sec4as2} is not sufficient for the TWFEIV estimand to attain its causal interpretation. If the exposed effects in the first stage are stable over time, we can interpret the first term $WCLATT$ causally and its interpretation seems clear: this parameter is a positively weighted average of each $CLATT^{eq}_{k}(W)$ and each weight assigned to each $CLATT^{eq}_{k}(W)$ reflects all the variations in each DID-IV design. However, the second term $\Delta CLATT$ still remains, which contaminates the causal interpretation of the TWFEIV estimand.

\subsubsection*{Interpretation under Assumption \ref{sec4as2} and Assumption \ref{sec4as3}}
Next, we assume Assumption \ref{sec4as2} and Assumption \ref{sec4as3} additionally. Even in this case, we still have the second term $\Delta CLATT \neq 0$ in general. This implies that the TWFEIV estimand identifies $WCLATT-\Delta CLATT$, that is, this estimand does not generally attain its causal interpretation.

\subsubsection*{Interpretation under Assumption \ref{sec4as2} and Assumption \ref{sec4as4}}
As we already noted in section \ref{sec4.2}, if we assume Assumption \ref{sec4as2} and Assumption \ref{sec4as4} additionally, we have Assumption \ref{sec4as7}, that is, $CLATT_{e,t}=CLATT_{e}$ holds. Then, we obtain the following Lemma.
\begin{Lemma}
\label{sec4lemma3}
Suppose Assumptions \ref{sec2as1}-\ref{sec2as7} hold. In addition, if Assumption \ref{sec4as2} and Assumption \ref{sec4as4} hold, the TWFEIV estimand $\beta_{IV}$ is:
\begin{align*}
\beta_{IV}=\sum_{k \neq U}CLATT_{k}\underbrace{\Bigg[w_{IV,kU}+\sum_{j=1}^{k-1}w_{IV,jk}^{k}
+\sum_{j=k+1}^{K}w_{IV,kj}^{k}\Bigg]}_{\equiv w_{k,IV}}.
\end{align*}
where the weights $w_{IV,kU}$, $w_{IV,kj}^{k}$ and $w_{IV,jk}^{k}$ are the probability limit of $\hat{w}_{IV,kU}$, $\hat{w}_{IV,kj}^{k}$ and $\hat{w}_{IV,jk}^{k}$ respectively.\par
\end{Lemma}\par
\begin{proof}
See Appendix \ref{ApeB}.
\end{proof}
If Assumption \ref{sec4as2} and Assumption \ref{sec4as4} are satisfied, the TWFEIV estimand is a positively weighted average of each $CLATT_k$ across exposed cohorts, which implies that we can interpret this estimand causally. However, at the same time, we also note that the weight $w_{k,IV}$ assigned to each $CLATT_k$ does not reflect only the cohort share and the fraction of the compliers, but is a function of the cohort share, the fraction of the compliers, and the timing of the initial exposure to the instrument.\par

\section{Extensions}\label{sec5}
This section briefly describes the extensions in section \ref{sec4}. We consider a non-binary, ordered treatment and unbalanced panel settings. It also includes the case when the adoption date of the instrument is randomized across units. For the proofs and the specific discussions, see Appendix \ref{ApeC}.
\subsubsection*{Non-binary, ordered treatment}
Up to now, we have considered only the case of a binary treatment. When treatment takes a finite number of ordered values, $D_{i,t} \in \{0,1,\dots,J\}$, our target parameter in staggered DID-IV design is the cohort specific average causal response on the treated (CACRT) defined below.
\begin{Def} The cohort specific average causal response on the treated (CACRT) at a given relative period $l$ from the initial adoption of the instrument is
\begin{align*}
CACRT_{e,l} \equiv \sum_{j=1}^{J}w^{e}_{e+l,j} \cdot E[Y_{i,e+l}(j)-Y_{i,e+l}(j-1)|E_i=e, D_{i,e+l}^{e} \geq j > D_{i,e+l}^{\infty}]
\end{align*}
where the weights $w^{e}_{e+l,j}$ are:
\begin{align*}
w^{e}_{e+l,j}=\frac{Pr(D_{i,e+l}^{e} \geq j > D_{i,e+l}^{\infty}|E_i=e)}{\sum_{j=1}^{J} Pr(D_{i,e+l}^{e} \geq j > D_{i,e+l}^{\infty}|E_i=e)}.
\end{align*}
\end{Def}
The CACRT is a weighted average of the effect of a unit increase in treatment on outcome, for those who are in cohort $e$ and induced to increase treatment by instrument at a relative period $l$ after the initial exposure. This parameter is similar to the average causal response (ACR) considered in \cite{Angrist1995-ij}, but the difference here is that there exist dynamic effects in the first stage, and each weight $w^{e}_{e+l,j}$ and the associated causal parameters in CACRT are conditioned on $E_i=1$.\par
If we have a non-binary, ordered treatment, one can show that we have Theorem \ref{sec4thm2} and Lemmas \ref{sec4lemma2}-\ref{sec4lemma3} in section \ref{sec4}, which replace $CLATT_{e,k}$ with $CACRT_{e,k}$. Note that our decomposition result for the TWFEIV estimator is unchanged under non-binary, ordered treatment settings.

\subsubsection*{Unbalanced panel case}
Throughout sections \ref{sec2} to \ref{sec4}, we have considered a balanced panel setting. If we assume an unbalanced panel (or repeated cross section) setting, we obtain the following theorem.
\begin{Theorem}
\label{sec5thm1}
    Suppose Assumptions \ref{sec2as1}-\ref{sec2as7} hold. If we assume a binary treatment and an unbalanced panel setting, the population regression coefficient $\beta_{IV}$ is a weighted average of each $CLATT_{e,t}$ in all relative periods after the initial exposure across cohorts with potentially some negative weights:
    \begin{align*}
    \beta_{IV}=\sum_{e}\sum_{t \geq e}w_{e,t}\cdot CLATT_{e,t}.
    \end{align*}
    where the weight $w_{e,t}$ is:
    \begin{align*}
    w_{e,t}=\frac{E[\hat{Z}_{i,t}|E_i=e]\cdot n_{e,t} \cdot CAET^{1}_{e,t}}{\sum_{e}\sum_{t \geq e}E[\hat{Z}_{i,t}|E_i=e]\cdot n_{e,t} \cdot CAET^{1}_{e,t}},
    \end{align*}
    where $E[\hat{Z}_{i,t}|E_i=e]$ is the population residuals from regression $Z_{i,t}$ on unit and time fixed effects in cohort $e$ and $n_{e,t}$ is the population share for cohort $e$ at time $t$. The weights sum to one.
\end{Theorem}
\begin{proof}
See Appendix \ref{ApeC}.
\end{proof}
Theorem \ref{sec5thm1} shows that the population regression coefficient $\beta_{IV}$ is a weighted average of all possible $CLATT_{e,t}$ across cohorts, but some weights can be negative. Theorem \ref{sec5thm1} is related to \cite{De_Chaisemartin2020-dw}, who show the decomposition theorem for the TWFEIV estimand when the assignment of the instrument is non-staggered and a no carry over assumption is satisfied in the first stage regression. Theorem \ref{sec5thm1} instead considers the case when the assignment of the instrument is staggered and there exist dynamic effects in the first stage. Theorem \ref{sec5thm1} assumes a binary treatment, but a non-binary, ordered treatment case is easy to extend: one can obtain the theorem which replaces $CLATT_{e,t}$ with $CACRT_{e,t}$.\par 
If one wants to check the validity of the TWFEIV estimator in a given application, one can estimate each weight by constructing the consistent estimator for $CAET^{1}_{e,t}$. If there does not exist a never exposed cohort, however, it is not feasible to obtain the consistent estimator for $CAET^{1}_{l,t}$ in the last exposed cohort $l=\max\{E_i\}$. In Appendix \ref{ApeC}, we provide another representation of the decomposition theorem, in which we can estimate each weight consistently and quantify the bias term arising from the bad comparisons performed by TWFEIV regressions.

\subsubsection*{Random assignment of the adoption date} 
In practice, researchers may use the TWFEIV regression when the adoption date of the instrument is randomized across units (e.g., Randomized control trial). In Appendix \ref{ApeC}, we consider the causal interpretation of the TWFEIV estimand under the random assignment assumption. In the DID literature, a similar issue is analyzed in \cite{Athey2022-uo}: they investigate the causal interpretation of the TWFE estimand when the adoption date of the treatment is randomized across units.\par
First, we define the random assignment assumption of the adoption date $E_i$.
\begin{Assumption}[Random assignment assumption of adoption date $E_i$]
\label{sec5as1}
For all $t \in \{1,\dots,T\}$ and all $z \in \mathcal{S}(Z)$, $E_i$ is independent of potential outcomes:
\begin{align*}
(Y_{i,t}(1),Y_{i,t}(0),D_{i,t}(z))  \mathop{\perp\!\!\!\!\perp}  E_i.
\end{align*}
\end{Assumption}
When the assignment of the adoption date is totally randomized, our target parameter is the local average treatment effect (LATE) defined below.
\begin{Def}
The local average treatment effect (LATE) at a given relative period $l$ from the initial adoption of the instrument is 
\begin{align*}
LATE_{e,l}=E[Y_{i,e+l}(1)-Y_{i,e+l}(0)|CM_{e,e+l}].
\end{align*}
\end{Def}\par
Unlike the CLATT, this parameter is not conditioned on the adoption date $E_i$ due to the independence assumption. The causal parameter in the first stage, $CAET^{1}_{e,l}$, is also simplified to the average exposed effect ($AE^{1}_{e,l}$) defined below:
\begin{align*}
CAET^{1}_{e,l}&=E[D_{i,e+l}-D_{i,e+l}^{\infty}]\\
&\equiv AE^{1}_{e,l}.
\end{align*}\par
If Assumptions \ref{sec2as1}- \ref{sec2as7} and Assumption \ref{sec5as1} hold, one can obtain the theorem and lemmas in section \ref{sec4}, which replace $CAET^{1}_{k,t}$ and $CLATT_{k,t}$ with $AE^{1}_{k,t}$ and $LATE_{k,t}$, respectively. This implies that even when the adoption date of the instrument is randomized, we cannot interpret the TWFEIV estimand causally in general, and the causal interpretation requires the stable exposed assumptions in both the first stage and reduced form regressions.

\section{Application}\label{sec6}
In this section, we illustrate our DID-IV decomposition theorem in the setting of \cite{Miller2019-ok}. We first explain our dataset. We then assess the plausibility of the staggered DID-IV identification strategy implicitly imposed by \cite{Miller2019-ok}. Finally, we present the DID-IV decomposition result and state the implication.\par
\cite{Miller2019-ok} study the effect of an increase in the share of female police officers on intimate partner homicide (IPH) rates among women in the United States between $1977$ and $1991$. The increase was in line with a shift in gender norms during these periods and there was growing interest in whether the female integration improved police quality in addressing violence against women.\par
To establish the causal relationship, \cite{Miller2019-ok} first regress the IPH rates on the lagged female officers' share with county and year fixed effects. In the second part of their analysis, \cite{Miller2019-ok} exploit "plausibly exogenous variation in female integration from externally imposed AA (affirmative action) following employment discrimination cases against particular departments in different years" across $255$ counties. Specifically, \cite{Miller2019-ok} use the two-way fixed effects instrumental variable regression, instrumenting the lagged female officers' share with the exposure years of AA plans.\par
\cite{Miller2019-ok} implicitly rely on staggered DID-IV designs to estimate the causal effects: \cite{Miller2019-ok} concern that "AA itself might have occurred following increasing trends" in the share of female officers or the IPH rates. To address this concern, \cite{Miller2019-ok} check the trends of these variables before AA introduction using event study regressions in the first stage and reduced form.\par
In this application, we slightly modify the authors' setting for simplicity. Specifically, unlike \cite{Miller2019-ok}, we use the staggered adoption of AA plans as our instrument instead of the exposure years. In the authors' setting, AA plans were terminated in some counties during the sample period, which is probably the reason why \cite{Miller2019-ok} use the exposure years of AA plans as their instrument. We instead drop such counties from our sample and make the instrument assignment staggered. Although it reduces our sample size, it allows us to have a clearer staggered DID-IV identification strategy. In addition, it enables us to apply our DID-IV decomposition theorem to the TWFEIV estimate in the authors' setting.

\subsubsection*{Data}\label{sec6.1}
\begin{table}[t]
\centering
\footnotesize
\begin{threeparttable}
\renewcommand{\arraystretch}{1.2}
\caption{The staggered AA introduction: exposure year, cohort sizes.}
\begin{tabular*}{11cm}{l@{\hspace{2cm}}l}
\hline
Start year of AA plans & Number of counties \\
\hline
Unexposed counties & 159\\
1976 & 6\\
1977 & 3\\
1978 & 3\\
1979 & 4\\
1980 & 3\\
1981 & 5\\
1982 & 4\\
1983 & 3\\
1984 & 2\\
1985 & 1\\
1986 & 1\\
1987 & 3\\
1988 & 1\\
1990 & 1\\
\hline
\end{tabular*}
\begin{tablenotes}[flushleft]\footnotesize
\item \textit{Notes}: This table presents the initial exposure year of AA plans and the number of counties in each year in our final sample.
\end{tablenotes}
\label{table1}
\end{threeparttable}
\vspace{5mm}
\begin{threeparttable}
\renewcommand{\arraystretch}{1.2}
\caption{Summary statistics}
\begin{tabular*}{16cm}{l@{\hspace{1cm}}cccc}
\hline
&& All counties & Unexposed counties & Exposed counties\\
\hline
IPH per $100000$ population & 1977-91 & 0.544 & 0.521 & 0.638\\
& 1977 & 0.549 & 0.526 & 0.641\\
& 1991 & 0.489 & 0.461 & 0.599\\
Lagged female officer share & 1977-91 & 0.053 & 0.050 & 0.066\\
& 1977 & 0.033 & 0.033 & 0.032\\
& 1991 & 0.077 & 0.071 & 0.101\\
Counties & & 199 & 159 & 40\\
Observations & & 2985 & 2385 & 600\\
\hline
\end{tabular*}
\begin{tablenotes}[flushleft]\footnotesize
\item \textit{Notes}: This table presents summary statistics on our final sample from $1977$ to $1991$. The sample consists of $199$ counties.
\end{tablenotes}
\label{table2}
\end{threeparttable}
\end{table}
The data come from \cite{Miller2019-ok}. Our final sample differs from their main analysis sample in two ways. First, unlike \cite{Miller2019-ok}, we only include the counties whose variables are observable for all sample periods. This restriction excludes $20$ counties and allows us to create the balanced panel data set. Second, as we already noted, we construct an instrument that takes one after the AA introduction. \cite{Miller2019-ok} use data on AA plans from \cite{Miller2012-fo} and define the instrument as the difference between the current year and the start year of AA introduction\footnote{As one can see in this construction, \cite{Miller2019-ok} create the lagged instrument in line with the lagged female officers' share. Therefore, we construct the lagged staggered instrument instead of the current one.}; see \cite{Miller2012-fo}, \cite{Miller2019-ok} for details. We identify the initial year of AA plans in each county, and discard the counties whose AA plans ended between $1976$ and $1990$ ($8$ counties dropped) and whose AA plans were already implemented before $1976$ ($23$ counties dropped). Table \ref{table1} shows the timing of AA adoption across $199$ counties between $1976$ and $1990$.\par
Summary statistics for county characteristics are reported in Table \ref{table2}. We have a smaller sample size, but otherwise have a similar sample to that of \cite{Miller2019-ok}. Counties are separated into exposed and unexposed counties based on whether the county experienced AA introduction. In both types of counties, the lagged female officers' share increased over time. However, it increases more in counties who are exposed to AA plans during sample periods. The IPH rates had downward trends in all counties, but it seems that there are no systematic differences in the trends between exposed and unexposed counties.\par

\subsubsection*{Assessing the identifying assumptions in staggered DID-IV design}
In this section, we discuss the validity of the staggered DID-IV identification strategy implicitly imposed by \cite{Miller2019-ok}. Note that in the authors' setting, our target parameter is the cohort specific average causal response on the treated (CACRT) as female officer share is a non-binary, ordered treatment. We therefore expect that we can identify each CACRT if the underlying staggered DID-IV identification strategy seems plausible, which we will check below. Here, we presume the no carry over assumption (Assumption \ref{sec2as2}).\par

 \subparagraph{Exclusion restriction (Assumption \ref{sec2as3}).}
 It would be plausible, given that the AA plans (instrument) did not affect IPH rates other than by increasing the female officers' share. This assumption may be violated for instance if the AA plans increased both the black and female officer shares and changes in IPH rates reflect both effects. \cite{Miller2019-ok} conduct the robustness check and confirm that this is not the case; see footnote $42$ in \cite{Miller2019-ok} for details.
 \subparagraph{Monotonicity assumption (Assumption \ref{sec2as4}).}
It would be automatically satisfied in the authors' setting: the AA plans (instrument) were imposed on departments with the intent to increase the share of female police officers. This ensures that the dynamic effects of the instrument on female police officers should be non-negative after the AA introduction.
 \subparagraph{No anticipation in the first stage (Assumption \ref{sec2as5}).}
It would be plausible that there is no anticipatory behavior, given that the treatment status, i.e., the female officers share before the AA plans is equal to the one in the absence of the AA introduction across counties. This assumption may be violated if the police departments in some counties had private knowledge about the probability of the AA introduction and manipulated their treatment status before the implementation.\par
\begin{figure}[t]
  \begin{minipage}[b]{0.5\columnwidth}
    \centering
    (a) First stage DID
    \includegraphics[width=\columnwidth]{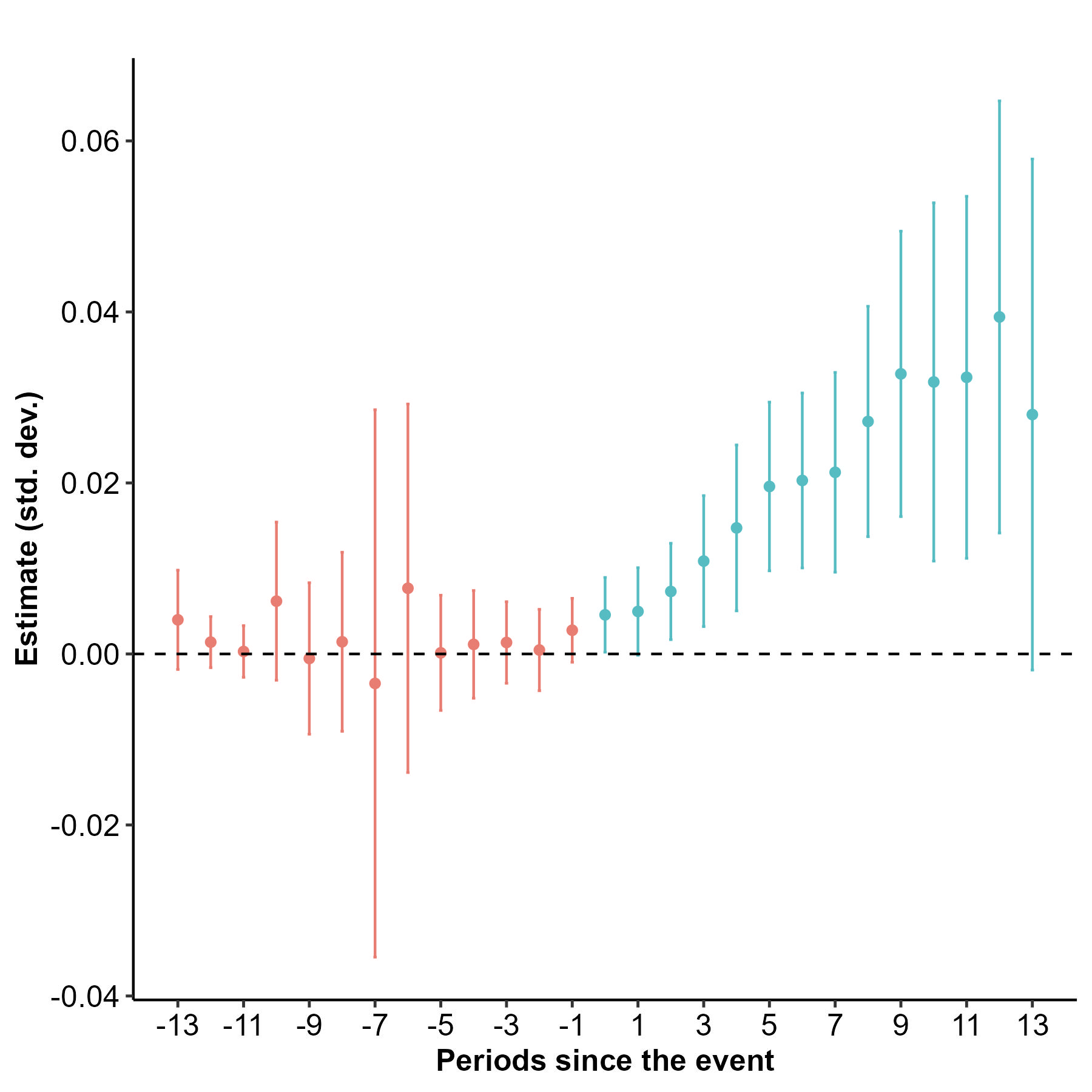}
  \end{minipage}
  \hspace{0.02\columnwidth} % ここで隙間作成
  \begin{minipage}[b]{0.5\columnwidth}
    \centering
    (b) Reduced form DID
    \includegraphics[width=\columnwidth]{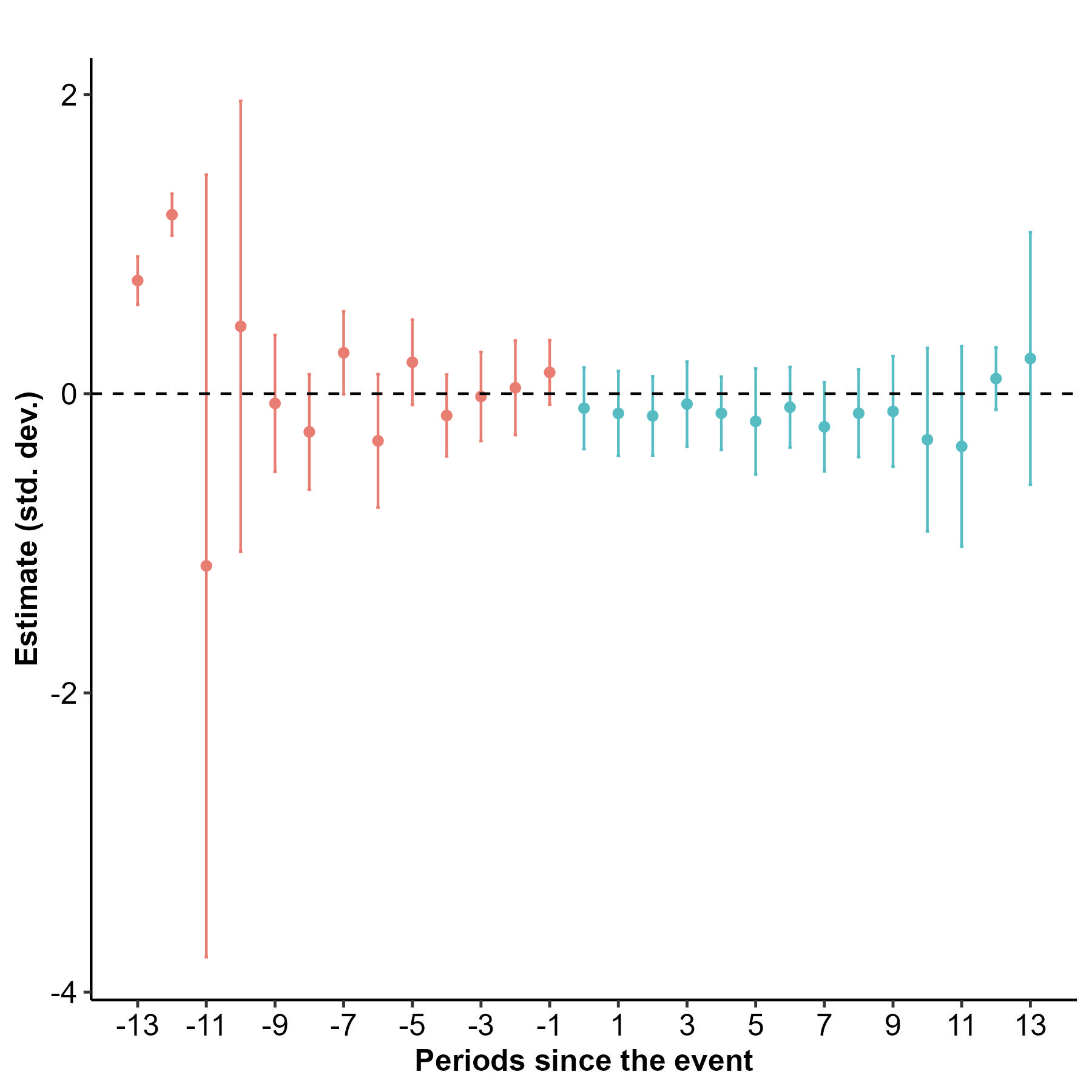}
  \end{minipage}
  \caption{Weighted average of the effects of AA introduction on female officer share and IPH rates in \cite{Miller2019-ok}. \textit{Notes}: The results for the effects of AA plans on the female officers share (Panel (a)) and on the IPH rates (Panel (b)) under the staggered DID-IV identification strategy. The red line represents the weighted average of the estimates with simultaneous $95\%$ confidence intervals for pre-exposed periods in both panels where the weight reflects cohort size in each period. The control group is a never exposed cohort and the reference period is $t-1$ for period $t$ estimate. These should be equal to zero under the null hypothesis that parallel trends assumptions in the treatment and the outcome hold. The blue line represents the weighted average of the estimates with simultaneous $95\%$ confidence intervals for post exposed periods in both panels where the weight reflects cohort size in each period. The control group is a never exposed cohort and the reference period is $t=-1$ for all post-exposed period estimates. All the standard errors are clustered at county level.}
  \label{Figure2}
\end{figure}\par
\vskip\baselineskip
Next, we assess the plausibility of the parallel trends assumptions in the treatment and the outcome. To do so, we apply the method proposed by \cite{Callaway2021-wl} to the first stage and reduced form, respectively\footnote{Unfortunately, in the presence of heterogeneous treatment effects, the coefficients on event study regression face a contamination bias shown by \cite{Sun2021-rp}.}. Specifically, we estimate the weighted average of the effects of the instrument on the treatment and outcome in each relative period where the weight reflects the cohort size. We depict the results in Figure \ref{Figure2}. The plots report estimates for the effects before and after AA plans with a simultaneous $95\%$ confidence interval in each stage. The confidence intervals account for clustering at the county level.\par
 
\subparagraph{Parallel trends assumption in the treatment (Assumption \ref{sec2as6}).}
It requires that if the AA plans had not occurred, the average time trends of the female officers share would have been the same across counties and over time. The pre-exposed estimates in Panel (a) in Figure \ref{Figure2} seem consistent with the parallel trends assumption in the treatment: the pre-exposed estimates around AA plans are not significantly different from zero. 

\subparagraph{Parallel trends assumption in the outcome (Assumption \ref{sec2as7}).}
It would be plausible if the AA plans 
had not been implemented, the average time trends of the IPH rates would have been the same across counties and over time. Panel (b) in Figure \ref{Figure2} presents that the pre-exposed estimates around AA introduction are not significantly different from zero, which indicates that the parallel trends assumption in the outcome is also plausible.\par
\vskip\baselineskip
Figure \ref{Figure2} also sheds light on the dynamic effects of the AA plans on the female officer share and IPH rates during the post-exposed periods. The figure indicates that the effect of the AA plans on the female officer share increases over time, whereas the effect on IPH rates through the female officer share has downward trends during the post-exposed periods. We note that the estimated effects in the reduced form are not scaled by the ones in the first stage, i.e., these estimates do not capture each CACRT after the AA shock.

\subsubsection*{Illustrating the weights in TWFEIV regression}
First, we estimate the two-way fixed effects instrumental variable regression in the authors' setting. To clearly illustrate the shortcomings of the TWFEIV regression, we modify the authors' specification in two ways: \cite{Miller2019-ok} include some covariates and weight their regression with county population, whereas we exclude such covariates and do not apply their weights to our regression.\par
The result is shown in Table \ref{table3}. The two-way fixed effects instrumental variable estimate is $-0.646$ and it is not significantly different from zero\footnote{Although \cite{Miller2019-ok} do not report the TWFEIV estimate without weights and covariates, when we run such a TWFEIV regression in their final analysis sample, the IV estimate is $-1.445$ and is not significantly different from zero. This implies that we reach the same conclusion as in \cite{Miller2019-ok} in our data.}. However, as we already noted in section \ref{sec4}, we cannot generally interpret the IV estimate as measuring a properly weighted average of each CACRT if the effect of the AA introduction on female officer share or IPH rates is not stable over time.\par
Our DID-IV decomposition theorem (Theorem \ref{sec3thm1}) allows us to visualize the source of variations in the three types of the DID-IV design: Unexposed/Exposed, Exposed/Not Yet Exposed, and Exposed/Exposed Shift designs.
Panel (a) in Figure \ref{Figure3} plots the weights and the corresponding Wald-DID estimates for all designs and Panels (b), (c), and (d) in Figure \ref{Figure3} plot them for each type of the DID-IV design, respectively. Table \ref{table4} reports the total weight, total Wald-DID estimate, and weighted average of Wald-DID estimates in each type of the DID-IV design. The total weight and total Wald-DID estimate are calculated by summing the weights and Wald-DID estimates respectively, and the weighted average of Wald-DID estimates is calculated by summing the products of the weight and the associated Wald-DID estimate. Summing all the weighted average of Wald-DID estimates yields the two-way fixed instrumental variable estimate ($-0.646$).\par
Panel (a) in Figure \ref{Figure3} shows that the weights are heavily assigned to the Wald-DID estimates in Unexposed/Exposed designs. This is due to the large sample size of the unexposed cohort in the authors' setting. Panels (b), (c), and (d) in Figure \ref{Figure3} highlight that some weights in each type can be negative: $2$ out of $14$ weights are negative in Unexposed/Exposed designs, $29$ out of $91$ weights are negative in Exposed/Not Yet Exposed designs and $50$ out of $91$ weights are negative in Exposed/Exposed Shift designs. The negative weights arise because some DID estimates of the treatment in the first stage are negative in each type of the DID-IV design.\par
The TWFEIV estimate suffers from a downward bias due to the bad comparisons arising from the Exposed/Exposed shift designs. As we already mentioned in section \ref{sec4}, the TWFEIV estimand potentially fails to summarize the causal effects if the effect of the instrument on the treatment or the outcome evolves over time. Table \ref{table4} indicates that the estimated bias occurring from the Exposed/Exposed shift designs is quantitatively not negligible: the weighted average of the Wald-DID estimates in the Exposed/Exposed shift designs is $-0.093$, which accounts for one-seventh of our IV estimate.
\begin{figure}[H]
\centering
\footnotesize
\includegraphics[width=0.8\columnwidth]{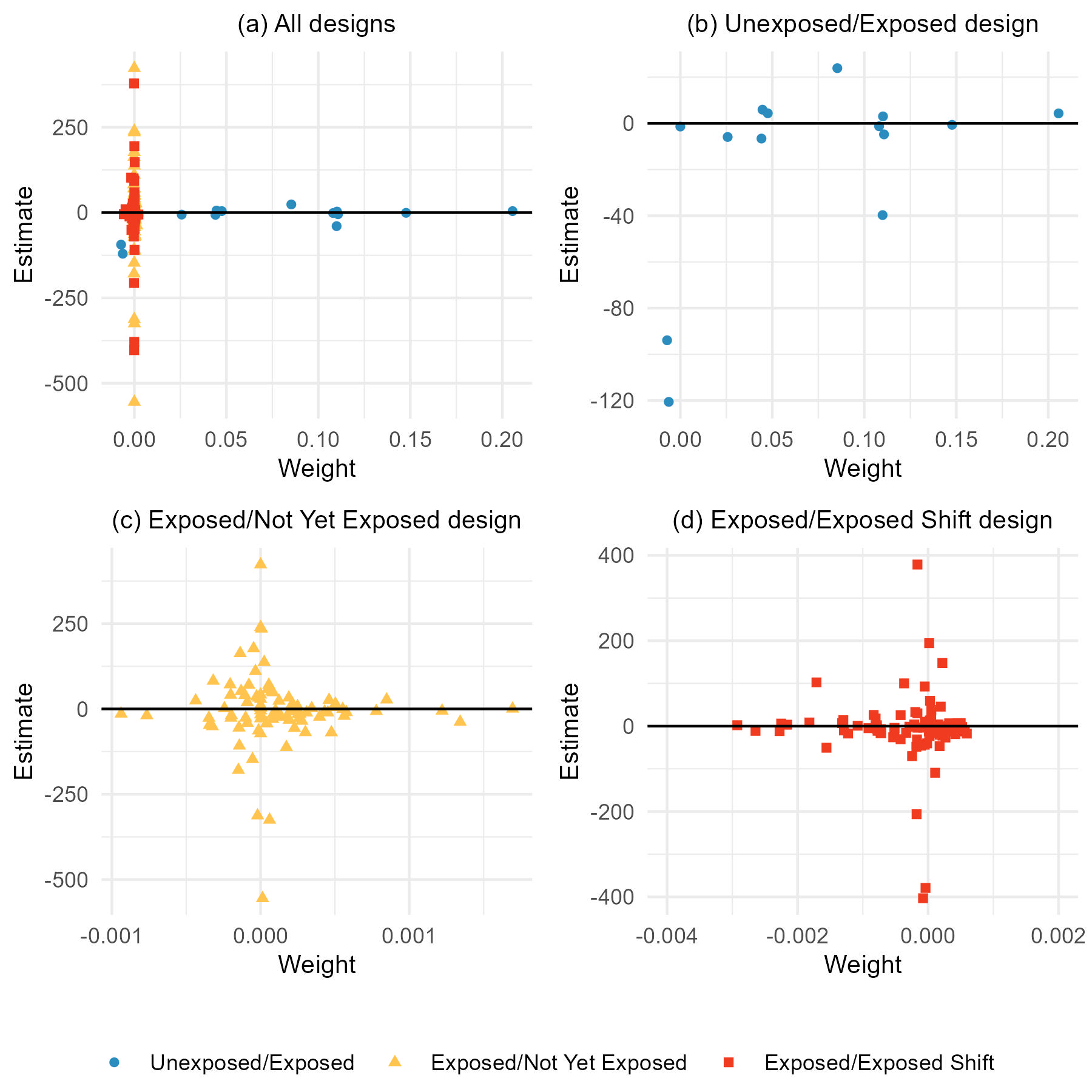}
\caption{Instrumented difference-in-differences decomposition result in the setting of \cite{Miller2019-ok}. \textit{Notes}: Panel (a) plots the weights and the corresponding Wald-DID estimates for all DID-IV designs and Panels (b), (c), and (d) plot them for each type of the DID-IV design, respectively. Unexposed/Exposed designs yield blue circles, Exposed/Not Yet Exposed designs yield yellow triangles and Exposed/Exposed Shift designs yield red squares.}
  \label{Figure3}
\end{figure}
\begin{table}[H]
\centering
\footnotesize
\begin{threeparttable}
\renewcommand{\arraystretch}{1.1}
\caption{Estimate for the effect of female officers share on IPH rates.}\label{table3}
\begin{tabular*}{14cm}{c@{\hspace{2cm}}c@{\hspace{1cm}}c@{\hspace{1cm}}c}
\hline
& Estimate & Standard Error & 95\% CI \\
\hline
TSLS with fixed effects & -0.646 & 3.284 & [-7.594, 6.301]\\
\hline
\end{tabular*}
\begin{tablenotes}[flushleft]\footnotesize
\item \textit{Notes}: Sample consists of $199$ counties. Confidence intervals account for clustering at the county level.
\end{tablenotes}
\end{threeparttable}
\begin{threeparttable}
\renewcommand{\arraystretch}{1.1}
\caption{Total weight, Total and Weighted WDD estimates in each type of the DID-IV design.}\label{table4}
\begin{tabular*}{16cm}{c@{\hspace{1cm}}ccc}
\hline
& Total weight & Total WDD estimate & Weighted WDD estimate\\
\hline
Unexposed/Exposed & 1.026 & -233.198 & -0.399\\
Exposed/Not Yet Exposed & 0.010 & -490.665 & -0.154\\
Exposed/Exposed Shift & -0.036 & -569.426 & -0.093\\
\hline
\end{tabular*}
\begin{tablenotes}[flushleft]\footnotesize
\item \textit{Notes}: This table presents the total weight, total Wald-DID estimate, and weighted average of Wald-DID estimates in each type of the DID-IV design: Unexposed/Exposed, Exposed/Not Yet Exposed, and Exposed/Exposed Shift designs.
\end{tablenotes}
\end{threeparttable}
\end{table}
\section{Alternative specifications}\label{sec7}
So far, we have considered simple TWFEIV regressions as in equation \eqref{intro1}. However, many studies routinely estimate various specifications, such as weighting or introducing covariates, to check the robustness of their findings. In this section, we extend our DID-IV decomposition theorem to the settings with weighting and covariates, and provide simple tools to examine how different specifications affect differences in estimates. We illustrate these by revisiting \cite{Miller2019-ok}.\par
The tools we provide here are based on \cite{Goodman-Bacon2021-ej}. Recall that our DID-IV decomposition theorem shows that the TWFEIV estimator can be written as the product of a vector of $2 \times 2$ Wald-DID estimators ($\hat{\bm{\beta}}_{IV}^{2 \times 2}$) and a vector of weights ($\bm{s}$), that is, $\hat{\beta}^{IV}=\bm{s'}\hat{\bm{\beta}}_{IV}^{2 \times 2}$. When a TWFEIV estimator generated from different specification ($\hat{\beta}_{IV,alt}$) can also be written as the product of a vector of $2 \times 2$ Wald-DID estimators ($\hat{\bm{\beta}}_{IV,alt}^{2 \times 2}$) and a vector of their associated weights ($\bm{s_{alt}}$), one can decompose the difference between the two specifications as
\begin{align*}
\hat{\beta}_{IV,alt}-\hat{\beta}_{IV}=\underbrace{\bm{s'}(\hat{\bm{\beta}}_{IV,alt}^{2 \times 2}-\hat{\bm{\beta}}_{IV}^{2 \times 2})}_{\text{Due to}\hspace{1mm}2 \times 2\hspace{1mm}\text{Wald-DIDs}}+\underbrace{(\bm{s'_{alt}}-\bm{s'})\hat{\bm{\beta}}_{IV}^{2 \times 2}}_{\text{Due to}\hspace{1mm}2 \times 2\hspace{1mm}\text{weights}}+\underbrace{(\bm{s'_{alt}}-\bm{s'})(\hat{\bm{\beta}}_{IV,alt}^{2 \times 2}-\hat{\bm{\beta}}_{IV}^{2 \times 2})}_{\text{Due to the interaction of the two}}.
\end{align*}
It takes the form of a Oaxaca-Blinder-Kitagawa decomposition (\cite{Oaxaca1973-zy}, \cite{Blinder1973-os}, \cite{Kitagawa1955-gz}) and indicates that the difference comes from changes in $2 \times 2$ Wald-DID estimators, changes in weights, and the interaction of the two. Dividing both sides by $\hat{\beta}_{IV,alt}-\hat{\beta}_{IV}$, one can measure the proportional contribution of each term on the difference.  Plotting each pair in ($\hat{\bm{\beta}}_{IV,alt}^{2 \times 2}, \hat{\bm{\beta}}_{IV}^{2 \times 2})$ and $(\bm{s'_{alt}}, \bm{s'})$, one can also examine which elements in each term have a significant impact on the difference. 
\subsection{Weighted TWFEIV regression}
When researchers use weighted TWFEIV regression instead of unweighted one, it potentially changes the influence of Wald-DID estimators ($\hat{\bm{\beta}}_{IV,WLS}^{2 \times 2}$) by replacing the DIDs of the treatment and the outcome with the weighted ones. It also potentially change the influence of weights ($\bm{s'_{WLS}}$) by replacing the sample share with the relative amount of the specified weight and the DIDs of the treatment with the weighted ones. Table \ref{table5} shows the result of our TWFEIV regression weighted by county population in \cite{Miller2019-ok}: the estimate changes from $-0.646$ to $-0.386$. The decomposition result indicates that the contribution of the changes in $2 \times 2$ Wald-DIDs is negative, whereas the contributions of the changes in weights and the interaction are positive.\par
Figure \ref{Figure4} plots the $2 \times 2$ Wald-DIDs and the associated weights in WLS against those in OLS. Panel (a) shows that most comparisons of the Wald-DID between OLS and WLS are located at the $45$-degree line, but some comparisons generated from Exposed/Not Yet Exposed and Exposed/Exposed Shift designs are away from the 45-degree line. In addition, this figure indicates that the Wald-DID generated from the comparison between $1978$ and $1991$ counties ($1991$ counties are the controls) is much more negative in WLS than in OLS, which drives the overall negative impact of the changes in $2 \times 2$ Wald-DIDs on the difference between the two specifications. Panel (b) shows that most comparisons of the decomposition weight between OLS and WLS are near the $45$-degree line and the origin, but some comparisons generated from Unexposed/Exposed designs are away from the 45-degree line and the origin. This figure also indicates that the decomposition weight generated from the comparison between $1982$ and unexposed counties is much more positive in WLS than in OLS, which causes the overall positive impact of the changes in weights on the difference between the two specifications. 
\begin{table}[t]
\centering
\footnotesize
\renewcommand{\arraystretch}{1.1}
\caption{Estimate for the effect of female officers share on IPH rates.}\label{table5}
\begin{tabular*}{16cm}{l@{\hspace{3cm}}c@{\hspace{2cm}}c@{\hspace{2cm}}c}
\hline
& $(1)$ & $(2)$ & $(3)$ \\
\hline
& Baseline & WLS & Covariates \\
\hline
Estimate & -0.646 & -0.386 & -0.868 \\
Standard Error & 3.284 & 2.452 & 3.968 \\
Difference from baseline & & 0.260 & -0.222\\
Difference comes from: & & &\\
\multicolumn{1}{c}{$ 2 \times 2$ Wald-DIDs} & & -4.048 & 0.370\\
\multicolumn{1}{c}{Weights} & & 2.341 & 16.503\\
\multicolumn{1}{c}{Interaction} & & 1.966 & -17.107\\
\multicolumn{1}{c}{Within term} & & 0 & 0.012\\
\hline
\end{tabular*}
\begin{tablenotes}[flushleft]\footnotesize
\item \textit{Notes}: This table presents TWFEIV estimates in the setting of \cite{Miller2019-ok}. Column $(1)$ is a simple TWFEIV estimate from Eq. \eqref{intro1}. Column $(2)$ is a TWFEIV estimate weighted by county population in $1977$. Column $(3)$ is a TWFEIV estimate with time-varying covariates which include the lagged local area controls, the county's non-IPH rate, and the state-level crack cocaine index. All the standard errors are clustered at county level.
\end{tablenotes}
\end{table}
\subsection{TWFEIV regression with time-varying covariates}
In most applications of thr DID-IV method, researchers typically estimate TWFEIV models that include time-varying covariates, in addition to the simple ones, based on the belief that it enhances the validity of the parallel trends assumptions in the first stage and reduced form regressions:
\begin{align}
\label{TWFEIV-covariate1}
    &Y_{i,t}=\phi_{i.}+\lambda_{t.}+\beta_{IV}^{X} D_{i,t}+\psi X_{i,t}+v_{i,t},\\
\label{TWFEIV-covariate2}
    &D_{i,t}=\gamma_{i.}+\zeta_{t.}+\pi^{X} Z_{i,t}+\tilde{\psi} X_{i,t}+\eta_{i,t}.
\end{align}
In this section, we derive a DID-IV decomposition result for the case when we introduce the time-varying covariates into TWFEIV regressions. Our decomposition result in this section is based on \cite{Goodman-Bacon2021-ej}, who decomposes TWFE estimators with time-varying covariates. Appendix \ref{ApeD} further considers the causal interpretation of the covariate-adjusted TWFEIV estimand under additional conditions.\par
First, consider the coefficient on instrument ($\alpha^{X}$) in the reduced form regression:
\begin{align}
\label{TWFEIV-covariate3}
    &Y_{i,t}=\phi_{i.}+\lambda_{t.}+\alpha^{X} Z_{i,t}+\xi X_{i,t}+v_{i,t}.
\end{align}
Let $\Tilde{Z}_{i,t}$ and $\Tilde{X}_{i,t}$ denote the double demeaning variables of $Z_{i,t}$ and $X_{i,t}$ respectively, obtained from regressing $Z_{i,t}$ and $X_{i,t}$ on time and unit fixed effects. Let $\Tilde{z}_{i,t}$ denote the residuals obtained from regressing $\Tilde{Z}_{i,t}$ on $\Tilde{X}_{i,t}$:
\begin{align*}
\Tilde{Z}_{i,t}=\overbrace{\hat{\Gamma}\Tilde{X}_{i,t}}^{\Tilde{p}_{i,t}}+\Tilde{z}_{i,t}.
\end{align*}
Here, we define the linear projection as $\Tilde{p}_{i,t} \equiv \hat{\Gamma}\Tilde{X}_{i,t}$. The specific expression for $\Tilde{z}_{i,t}$ is:
\begin{align*}
\Tilde{z}_{i,t}&=[(Z_{i,t}-\bar{Z_i})-(\hat{\Gamma}X_{i,t}-\hat{\Gamma}\bar{X_{i}})]-[(\bar{Z}_t-\bar{\bar{Z}})-(\hat{\Gamma}\bar{X}_{t}-\hat{\Gamma}\bar{\bar{X}})]\\
&\equiv (z_{i,t}-\bar{z}_{i})-(\bar{z}_{t}-\bar{\bar{z}}).
\end{align*}
By the FWL theorem, we then obtain the following expression for $\hat{\alpha}^{X}$:
\begin{align*}
\hat{\alpha}^{X}=\frac{\hat{C}(Y_{i,t},\Tilde{z}_{i,t})}{\hat{V}^{\Tilde{z}}}=\frac{\hat{C}(Y_{i,t},\Tilde{Z}_{i,t}-\Tilde{p}_{i,t})}{\hat{V}^{\Tilde{z}}},
\end{align*}
where $\hat{V}^{\Tilde{z}}$ is the variance of $\Tilde{z}_{i,t}$.
By symmetry, we can also express the first stage coefficient on instrument $\hat{\pi}^{X}$ as follows:
\begin{align*}
\hat{\pi}^{X}=\frac{\hat{C}(D_{i,t},\Tilde{z}_{i,t})}{\hat{V}^{\Tilde{z}}}=\frac{\hat{C}(D_{i,t},\Tilde{Z}_{i,t}-\Tilde{p}_{i,t})}{\hat{V}^{\Tilde{z}}}.
\end{align*}
Because the IV estimator $\hat{\beta}_{IV}^{X}$ is the ration between the first stage coefficient $\hat{\pi}^{X}$ and the reduced form coefficient $\hat{\alpha}^{X}$, we obtain the following expression for $\hat{\beta}_{IV}^{X}$:
\begin{align}
\label{TWFEIVcovariates_IV1}
\hat{\beta}_{IV}^{X}=\frac{\hat{C}(Y_{i,t},\Tilde{z}_{i,t})}{\hat{C}(D_{i,t},\Tilde{z}_{i,t})}=\frac{\hat{C}(Y_{i,t},\Tilde{Z}_{i,t}-\Tilde{p}_{i,t})}{\hat{C}(D_{i,t},\Tilde{Z}_{i,t}-\Tilde{p}_{i,t})}.
\end{align}
In contrast to the unconditional TWFEIV estimator $\hat{\beta}_{IV}$, the covariate-adjusted TWFEIV estimator exploits the variation in both $\Tilde{Z}_{i,t}$ and $\Tilde{p}_{i,t}$. $\Tilde{Z}_{i,t}$ varies at cohort and time level, but $\Tilde{p}_{i,t}$ varies at unit and time level because $X_{i,t}$ varies at unit and time level.\par
To decompose the covariate-adjusted TWFEIV estimator $\hat{\beta}_{IV}^{X}$, we first partition $\Tilde{z}_{i,t}$ into "within" and "between" terms as in \cite{Goodman-Bacon2021-ej}. Let $\bar{z}_{k,t}-\bar{z}_{k}=(\bar{Z}_{k,t}-\bar{Z}_{k})-(\hat{\Gamma}\bar{X}_{k,t}-\hat{\Gamma}\bar{X}_{k})$ be the average of $z_{i,t}-\bar{z}_{i}$ in cohort $k$. By adding and subtracting $\bar{z}_{k,t}-\bar{z}_{k}$, we can decompose $\Tilde{z}_{i,t}$ into two terms:
\begin{align}
\label{TWFEIVwithin1}
\Tilde{z}_{i,t}=\underbrace{[(z_{i,t}-\bar{z}_{i})-(\bar{z}_{k,t}-\bar{z}_{k})]}_{\Tilde{z}_{i(k),t}}+\underbrace{[(\bar{z}_{k,t}-\bar{z}_{k})-(\bar{z}_{t}-\bar{\bar{z}})]}_{\Tilde{z}_{k,t}}.
\end{align}
The first term $\Tilde{z}_{i(k),t}$ measures the deviation of $z_{i,t}-\bar{z}_{i}$ from the average $\bar{z}_{k,t}-\bar{z}_{k}$ in cohort $k$, which we call the within term of $\Tilde{z}_{i,t}$. The second term $\Tilde{z}_{k,t}$ measures the deviation of $\bar{z}_{k,t}-\bar{z}_{k}$ from the average $\bar{z}_{t}-\bar{\bar{z}}$ in whole sample, which we call the between term of $\Tilde{z}_{i,t}$. The within term $\Tilde{z}_{i(k),t}$ varies at unit and time level because of $\Tilde{p}_{i,t}$, whereas the between term $\Tilde{z}_{k,t}$ varies at cohort and time level.\par
By substituting \eqref{TWFEIVwithin1} into \eqref{TWFEIVcovariates_IV1}, we obtain
\begin{align}
\label{TWFEIVcovariates_IV2}
\hat{\beta}_{IV}^{X}=\frac{\hat{C}(Y_{i,t},\Tilde{z}_{i(k),t})+\hat{C}(Y_{i,t},\Tilde{z}_{k,t})}{\hat{C}(D_{i,t},\Tilde{z}_{i(k),t})+\hat{C}(D_{i,t},\Tilde{z}_{k,t})}&=\frac{\hat{V}_{w}^{z} \hat{\beta}_{w}^{p,y}+\hat{V}_{b}^{z} \hat{\beta}_{b}^{z,y}}{\hat{V}_{w}^{z} \hat{\beta}_{w}^{p,d}+\hat{V}_{b}^{z} \hat{\beta}_{b}^{z,d}}\\
&=\underbrace{\frac{\hat{C}_{w}^{D,\Tilde{z}}}{\hat{C}_{w}^{D,\Tilde{z}}+\hat{C}_{b}^{D,\Tilde{z}}}}_{\Omega}\cdot \underbrace{\frac{\hat{\beta}_{w}^{p,y}}{ \hat{\beta}_{w}^{p,d}}}_{\hat{\beta}_{w,IV}^{p}}+\underbrace{\frac{\hat{C}_{b}^{D,\Tilde{z}}}{\hat{C}_{w}^{D,\Tilde{z}}+\hat{C}_{b}^{D,\Tilde{z}}}}_{1-\Omega}\cdot \underbrace{\frac{\hat{\beta}_{b}^{z,y}}{\hat{\beta}_{b}^{z,d}}}_{\hat{\beta}_{b,IV}^{z}}.
\end{align}
We use the subscript $w$ to denote within components and the subscript $b$ to denote between components. $\hat{V}_{w}^{z}$ and $\hat{V}_{b}^{z}$ are the variances of $\Tilde{z}_{i(k),t}$ and $\Tilde{z}_{k,t}$, respectively. $\hat{C}_{w}^{D,\Tilde{z}}$ is the covariance between $D_{i,t}$ and $\Tilde{z}_{i(k),t}$, the within term of $\Tilde{z}_{i,t}$. $\hat{C}_{b}^{D,\Tilde{z}}$ is the covariance between $D_{i,t}$ and $\Tilde{z}_{k,t}$, the between term of $\Tilde{z}_{i,t}$. The weight $\Omega=\frac{\hat{C}_{w}^{D,\Tilde{z}}}{\hat{C}_{w}^{D,\Tilde{z}}+\hat{C}_{b}^{D,\Tilde{z}}}$ measures the relative amount of the within covariance $\hat{C}_{w}^{D,\Tilde{z}}$.\par 
$\hat{\beta}_{w}^{p,y} \equiv \frac{\hat{C}(Y_{i,t},\Tilde{z}_{i(k),t})}{\hat{V}_{w}^{z}}$ measures the relationship between $Y_{i,t}$ and $\Tilde{z}_{i(k),t}$. Similarly, $\hat{\beta}_{w}^{p,d} \equiv \frac{\hat{C}(D_{i,t},\Tilde{z}_{i(k),t})}{\hat{V}_{w}^{z}}$ measures the relationship between $D_{i,t}$ and $\Tilde{z}_{i(k),t}$. We call these the within coefficients in the first stage and reduced form regressions. $\hat{\beta}_{w,IV}^{p} \equiv \frac{\hat{\beta}_{w}^{p,y}}{\hat{\beta}_{w}^{p,d}}$ scales the within coefficient in the reduced form regression by the one in the first stage regression. We call this the within IV coefficient\footnote{One can obtain this coefficient by running an IV regression of the outcome on the treatment with $\tilde{z}_{i(k),t}$ as the excluded instrument.}. This IV coefficient arises because $\Tilde{z}_{i(k),t}$ varies at unit and time level. Similar to what \cite{Goodman-Bacon2021-ej} points out for the covariate-adjusted TWFE estimator, time-varying covariates bring a new source of identifying variation in the TWFEIV estimator, within variation of $X_{i,t}$ in each cohort.
\par
$\hat{\beta}_{b}^{z,y} \equiv \frac{\hat{C}(Y_{i,t},\Tilde{z}_{k,t})}{\hat{V}_{b}^{z}}$ measures the relationship between $Y_{i,t}$ and $\Tilde{z}_{k,t}$. Similarly $\hat{\beta}_{b}^{z,d} \equiv \frac{\hat{C}(D_{i,t},\Tilde{z}_{k,t})}{\hat{V}_{b}^{z}}$ measures the relationship between $D_{i,t}$ and $\Tilde{z}_{k,t}$. We call these the between coefficients in the first stage and reduced form regressions. $\hat{\beta}_{b,IV}^{z} \equiv \frac{\hat{\beta}_{b}^{z,y}}{\hat{\beta}_{b}^{z,d}}$ divides the between coefficient in the reduced form regression by the one in the first stage regression, and have the following specific expression:
\begin{align}
\label{TWFEIV_between1}
\hat{\beta}_{b,IV}^{z}=\frac{\hat{C}^{D,Z}\hat{\beta}_{IV}-\hat{C}^{p}_{b}\hat{\beta}_{b,IV}^{p}}{\hat{C}_{b}^{D,\Tilde{z}}}.
\end{align}
$\hat{C}^{D,Z}$ and $\hat{\beta}_{IV}$ are already defined in section \ref{sec3}. $\hat{C}^{p}_{b}$ is the covariance between $D_{i,t}$ and $\Tilde{p}_{k,t}$ (the between term of $\Tilde{p}_{i,t}$). $\hat{\beta}_{b,IV}^{p}$ is the estimator, obtained from an IV regression of $Y_{i,t}$ on $D_{i,t}$ with $\Tilde{p}_{k,t}$ as the excluded instrument. We call $\hat{\beta}_{b,IV}^{z}$ the between IV coefficient, which exploits the cohort and time level variation in $\Tilde{z}_{k,t}$. This IV coefficient is not equal to the unconditional TWFEIV coefficient $\hat{\beta}_{IV}$: $\hat{\beta}_{b,IV}^{z}$ subtracts the influence of  $\hat{\beta}_{b,IV}^{p}$ from the unconditional IV estimator $\hat{\beta}_{IV}$. This indicates that time-varying covariates $X_{i,t}$ changes the identifying variation at cohort and time level through $\Tilde{p}_{k,t}$, the between term of the linear projection $\Tilde{p}_{i,t}$.\par 
We can further decompose the between IV coefficient as follows: 
\begin{align}
\label{TWFEIV_between2}
\hat{\beta}_{b,IV}^{z}=\sum_{k}\sum_{l > k}\underbrace{(n_k+n_l)^2\frac{\hat{C}_{b,kl}^{D,\Tilde{z}}}{\hat{C}_{b}^{D,\Tilde{z}}}}_{s_{b,kl}}\underbrace{\left[\frac{\hat{C}^{D,Z}_{kl}\hat{\beta}_{IV,kl}^{2 \times 2}-\hat{C}^{p}_{b,kl}\hat{\beta}_{b,IV,kl}^{p}}{\hat{C}_{b,kl}^{D,\Tilde{z}}}\right]}_{\hat{\beta}_{b,IV,kl}^{z}}.
\end{align}
The proof is given in Appendix \ref{ApeD}. Each notation is similarly defined in $(k,l)$ cell subsamples. $\hat{\beta}_{b,IV,kl}^{z}$ and $s_{b,kl}$ are the between IV coefficient and the corresponding weight in $(k,l)$ cell subsamples. Equation \eqref{TWFEIV_between2} indicates that time-varying covariates $X_{i,t}$ affect the between IV coefficient $\hat{\beta}_{b,IV}^{z}$ by changing both the $2 \times 2$ between IV coefficient and the associated weight in each $(k,l)$ cell. \par
To sum up, combining \eqref{TWFEIV_between2} with \eqref{TWFEIVcovariates_IV2}, we can decompose the covariate-adjusted TWFEIV estimator $\hat{\beta}_{IV}^{X}$ as 
\begin{align*}
\hat{\beta}_{IV}^{X}=\Omega \hat{\beta}_{w,IV}^{p} + (1-\Omega)\underbrace{\sum_{k}\sum_{l > k}s_{b,kl}\hat{\beta}_{b,IV,kl}^{z}}_{\hat{\beta}_{b,IV}^{z}}.
\end{align*}
The weight $\Omega$ is assigned to the within IV coefficient $\hat{\beta}_{w,IV}^{p}$ and the weight $1-\Omega$ is assigned to the between IV coefficient $\hat{\beta}_{b,IV}^{z}$, which is equal to a weighted average of all possible $2 \times 2$ between IV coefficients $\hat{\beta}_{b,IV,kl}^{z}$ as in Theorem \ref{sec3thm1}.\par
Table \ref{table5} presents the result of our TWFEIV regression with time-varying covariates in \cite{Miller2019-ok}. We follow \cite{Miller2019-ok} and include the lagged local area controls, the county's non-IPH rate, and the state-level crack cocaine index; see \cite{Miller2019-ok} for details. The estimate changes from $-0.646$ to $-0.868$. The decomposition result shows that the contribution of the within term is positive but negligible, whereas the contribution of the between term is negative and substantial. Specifically, in the between term, the contribution of the changes in $2 \times 2$ Wald-DIDs and weights are positive, but these are offset by the negative contribution of the interaction. This result indicates that in \cite{Miller2019-ok}, the time-varying covariates affect the IV estimate mainly through the identifying variation in cohort and time level, that is, the between term of the linear projection $\Tilde{p}_{i,t}$.\par
\begin{figure}[H]
\centering
\begin{adjustbox}{width=\columnwidth, height=0.6\columnwidth, keepaspectratio}
    \includegraphics{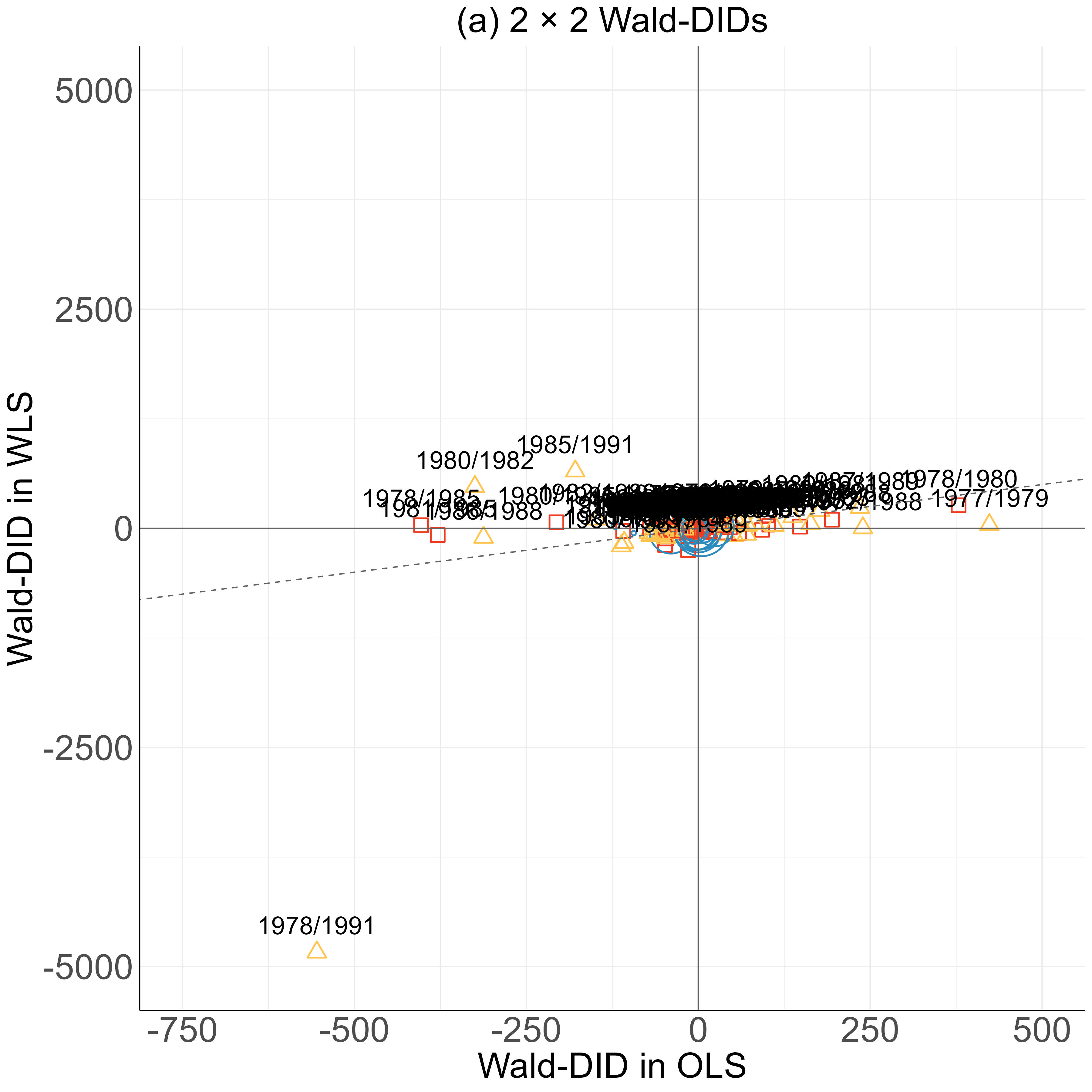}
    \end{adjustbox}
\begin{adjustbox}{width=\columnwidth, height=0.6\columnwidth, keepaspectratio}
    \includegraphics{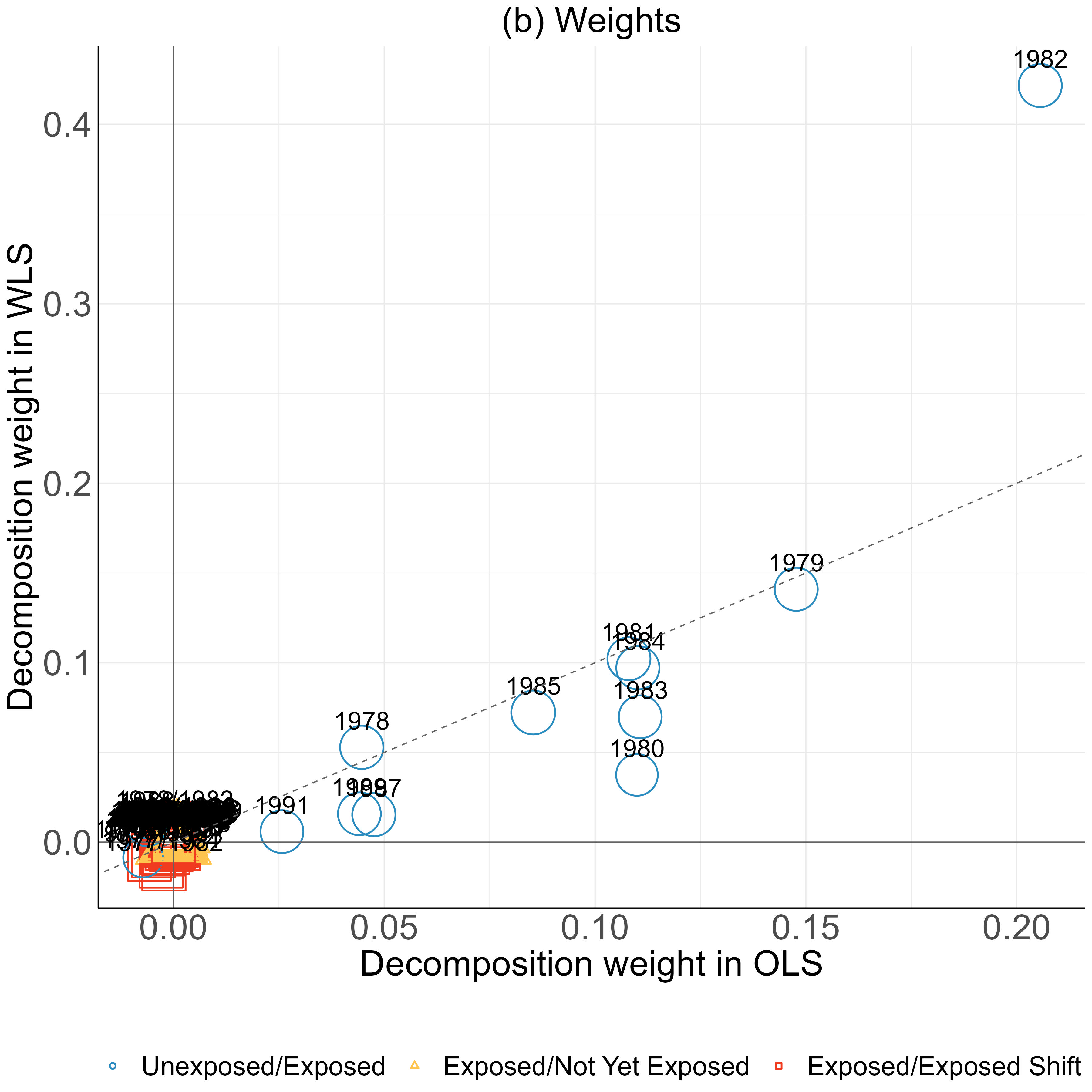}
    \end{adjustbox}
  \caption{Comparisons of $2 \times 2$ WaldDIDs and weights between OLS and WLS in the setting of \cite{Miller2019-ok}. \textit{Notes}: Panel (a) plots the $2 \times 2$ Wald-DIDs in WLS against those in OLS for all DID-IV designs. Panel (b) plots the decomposition weights in WLS against those in OLS for all DID-IV designs. In Panel (a), the size of each point is proportional to the corresponding weight in OLS. In Panel (b), the size of each point is proportional to the corresponding Wald-DID estimate in OLS. In both panels, the dotted lines represent 45-degree lines. In both panels, Unexposed/Exposed designs yield blue circles, Exposed/Not Yet Exposed designs yield yellow triangles, and Exposed/Exposed Shift designs yield red squares. In both panels, the dotted lines represent 45-degree lines.}
\label{Figure4}
  \end{figure}

\section{Conclusion}\label{sec8}
Many studies run two-way fixed effects instrumental variable (TWFEIV) regressions, leveraging variation occurring from the different timing of policy adoption across units as an instrument for the treatment. In this paper, we study the causal interpretation of the TWFEIV estimator in staggered DID-IV designs. We first show that in settings with the staggered adoption of the instrument across units, the TWFEIV estimator is equal to a weighted average of all possible $2 \times 2$ Wald-DID estimators arising from the three types of the DID-IV design: Unexposed/Exposed, Exposed/Not Yet Exposed, and Exposed/Exposed Shift designs. The weight assigned to each Wald-DID estimator is a function of the sample share, the variance of the instrument, and the DID estimator of the treatment in each DID-IV design.\par 
Based on the decomposition result, we then show that in staggered DID-IV designs, the TWFEIV estimand is equal to a weighted average of all possible cohort specific local average treatment effect on the treated parameters, but some weights can be negative. The negative weight problem arises due to the bad comparisons in the first and reduced form regressions: we use the already exposed units as controls. The TWFEIV estimand attains its causal interpretation if the effects of the instrument on the treatment and outcome are stable over time. The resulting causal parameter is a positively weighted average cohort specific local average treatment effect on the treated parameter.\par
Finally, we illustrate our findings with the setting of \cite{Miller2019-ok} who estimate the effect of female officers' share on the IPH rate, exploiting the timing variation of AA introduction across U.S. counties. We first assess the underlying staggered DID-IV identification strategy implicitly imposed by \cite{Miller2019-ok} and confirm its validity. We then apply our DID-IV decomposition theorem to the TWFEIV estimate, and find that the estimate suffers from the substantial downward bias arising from the bad comparisons in Exposed/Exposed shift DID-IV designs. We also decompose the difference between the two specifications and illustrate how different specifications affect the overall estimates in \cite{Miller2019-ok}.\par
Overall, this paper shows the negative result of using TWFEIV estimators in the presence of heterogeneous treatment effects in staggered DID-IV designs in more than two periods. This paper provides simple tools to evaluate how serious that concern is in a given application. Specifically, we demonstrate that the TWFEIV estimator is not robust to the time-varying exposed effects in the first stage and reduced form regressions. Our DID-IV decomposition theorem allows the empirical researchers to assess the impact of the bias term arising from the bad comparisons on their TWFEIV estimate. Recently, \cite{Miyaji2023} developed an alternative estimation method that is robust to treatment effects heterogeneity and proposes a weighting scheme to construct various summary measures in staggered DID-IV designs. Further developing alternative approaches and diagnostic tools will be a promising area for future work, facilitating the credibility of DID-IV design in practice. 
\newpage
\bibliography{reference} % 拡張子を外したbibファイル

@ARTICLE{De_Chaisemartin2018-xe,
  title     = "Fuzzy {Differences-in-Differences}",
  author    = "de Chaisemartin, C and D'Haultf{\oe}uille, X",
  journal   = "Rev. Econ. Stud.",
  publisher = "[Oxford University Press, The Review of Economic Studies, Ltd.]",
  volume    =  85,
  number    = "2 (303)",
  pages     = "999--1028",
  year      =  2018
}

@ARTICLE{Lundborg2017-mz,
  title     = "Can Women Have Children and a Career? {IV} Evidence from {IVF}
               Treatments",
  author    = "Lundborg, Petter and Plug, Erik and Rasmussen, Astrid Wurtz",
  journal   = "Am. Econ. Rev.",
  publisher = "Amer Economic Assoc",
  volume    =  107,
  number    =  6,
  pages     = "1611--1637",
  month     =  jun,
  year      =  2017,
  language  = "en"
}

@ARTICLE{Field2007-yc,
  title     = "Entitled to Work: Urban Property Rights and Labor Supply in Peru",
  author    = "Field, Erica",
  journal   = "Q. J. Econ.",
  publisher = "Oxford Academic",
  volume    =  122,
  number    =  4,
  pages     = "1561--1602",
  month     =  nov,
  year      =  2007
}

@ARTICLE{Akerman2015-hh,
  title     = "The Skill Complementarity of Broadband Internet",
  author    = "Akerman, Anders and Gaarder, Ingvil and Mogstad, Magne",
  journal   = "Q. J. Econ.",
  publisher = "Oxford Academic",
  volume    =  130,
  number    =  4,
  pages     = "1781--1824",
  month     =  jul,
  year      =  2015,
  language  = "en"
}

@ARTICLE{Borusyak2021-jv,
  title         = "Revisiting Event Study Designs: Robust and Efficient
                   Estimation",
  author        = "Borusyak, Kirill and Jaravel, Xavier and Spiess, Jann",
  month         =  aug,
  year          =  2021,
  archivePrefix = "arXiv",
  primaryClass  = "econ.EM",
  eprint        = "2108.12419"
}

@ARTICLE{Bhuller2013-ki,
  title     = "Broadband Internet: An Information Superhighway to Sex Crime?",
  author    = "Bhuller, Manudeep and Havnes, Tarjei and Leuven, Edwin and
               Mogstad, Magne",
  journal   = "Rev. Econ. Stud.",
  publisher = "Oxford Academic",
  volume    =  80,
  number    =  4,
  pages     = "1237--1266",
  month     =  apr,
  year      =  2013,
  language  = "en"
}

@ARTICLE{Miyaji2023,
  title        = "Instrumented Difference-in-Differences with heterogeneous treatment effects",
  author       = "Miyaji,Sho",
  year         =  2024,
  journal   = "Unpublished Manuscript", 
}

@ARTICLE{chasemartin2010-ch,
  title     = "A note on instrumented difference in differences",
  author    = "de Chaisemartin",
  journal   = "Unpublished Manuscript",
  publisher = "",
  year      =  2010
}

@ARTICLE{Oaxaca1973-zy,
  title     = "{Male-Female} Wage Differentials in Urban Labor Markets",
  author    = "Oaxaca, Ronald",
  journal   = "Int. Econ. Rev.",
  publisher = "[Economics Department of the University of Pennsylvania, Wiley,
               Institute of Social and Economic Research, Osaka University]",
  volume    =  14,
  number    =  3,
  pages     = "693--709",
  month     =  oct,
  year      =  1973
}

@ARTICLE{Kitagawa1955-gz,
  title     = "Components of a Difference Between Two Rates",
  author    = "Kitagawa, Evelyn M",
  journal   = "J. Am. Stat. Assoc.",
  publisher = "[American Statistical Association, Taylor \& Francis, Ltd.]",
  volume    =  50,
  number    =  272,
  pages     = "1168--1194",
  month     =  dec,
  year      =  1955
}

@ARTICLE{Blinder1973-os,
  title     = "Wage Discrimination: Reduced Form and Structural Estimates",
  author    = "Blinder, Alan S",
  journal   = "J. Hum. Resour.",
  publisher = "[University of Wisconsin Press, Board of Regents of the
               University of Wisconsin System]",
  volume    =  8,
  number    =  4,
  pages     = "436--455",
  month     =  oct,
  year      =  1973
}

@MISC{Hudson2017-tm,
  title        = "Interpreting Instrumented Difference-in-Differences",
  author       = "Hudson, Sally and Hull, Peter and Liebersohn, Jack",
  year         =  2017,
  howpublished = "Available at \url{http://www.mit.edu/~liebers/DDIV.pdf}"
}

@ARTICLE{Sun2021-rp,
  title     = "Estimating dynamic treatment effects in event studies with
               heterogeneous treatment effects",
  author    = "Sun, Liyang and Abraham, Sarah",
  journal   = "J. Econom.",
  publisher = "Elsevier Science Publishers",
  volume    =  225,
  number    =  2,
  pages     = "175--199",
  month     =  dec,
  year      =  2021
}

@ARTICLE{Imbens1994-qy,
  title     = "Identification and Estimation of Local Average Treatment Effects",
  author    = "Imbens, Guido W and Angrist, Joshua D",
  journal   = "Econometrica",
  publisher = "[Wiley, Econometric Society]",
  volume    =  62,
  number    =  2,
  pages     = "467--475",
  year      =  1994
}

@ARTICLE{Goodman-Bacon2021-ej,
  title    = "Difference-in-differences with variation in treatment timing",
  author   = "Goodman-Bacon, Andrew",
  journal  = "J. Econom.",
  volume   =  225,
  number   =  2,
  pages    = "254--277",
  month    =  dec,
  year     =  2021
}

@ARTICLE{De_Chaisemartin2020-dw,
  title     = "{Two-Way} Fixed Effects Estimators with Heterogeneous Treatment
               Effects",
  author    = "de Chaisemartin, Cl{\'e}ment and D'Haultf{\oe}uille, Xavier",
  journal   = "Am. Econ. Rev.",
  publisher = "American Economic Association",
  volume    =  110,
  number    =  9,
  pages     = "2964--2996",
  month     =  sep,
  year      =  2020,
  language  = "en"
}

@UNPUBLISHED{Blandhol2022-lk,
  title    = "When is {TSLS} Actually {LATE}?",
  author   = "Blandhol, Christine and Bonney, John and Mogstad, Magne and
              Torgovitsky, Alexander",
  journal  = "SSRN Electron. J.",
  month    =  feb,
  year     =  2022,
  language = "en"
}

@ARTICLE{Callaway2021-wl,
  title     = "{Difference-in-Differences} with multiple time periods",
  author    = "Callaway, Brantly and Sant'Anna, Pedro H C",
  journal   = "J. Econom.",
  publisher = "Elsevier Science Publishers",
  volume    =  225,
  number    =  2,
  pages     = "200--230",
  month     =  dec,
  year      =  2021
}

@ARTICLE{Angrist1995-ij,
  title     = "{Two-Stage} Least Squares Estimation of Average Causal Effects
               in Models with Variable Treatment Intensity",
  author    = "Angrist, Joshua D and Imbens, Guido W",
  journal   = "J. Am. Stat. Assoc.",
  publisher = "Taylor \& Francis",
  volume    =  90,
  number    =  430,
  pages     = "431--442",
  month     =  jun,
  year      =  1995
}

@ARTICLE{Athey2022-uo,
  title     = "Design-based analysis in {Difference-In-Differences} settings
               with staggered adoption",
  author    = "Athey, Susan and Imbens, Guido W",
  journal   = "J. Econom.",
  publisher = "Elsevier Sequoia S.A",
  volume    =  226,
  number    =  1,
  pages     = "62--79",
  month     =  jan,
  year      =  2022
}

@ARTICLE{Sloczynski2020-uk,
  title         = "When Should We (Not) Interpret Linear {IV} Estimands as
                   {LATE}?",
  author        = "S{\l}oczy{\'n}ski, Tymon",
  month         =  nov,
  year          =  2020,
  archivePrefix = "arXiv",
  primaryClass  = "econ.EM",
  eprint        = "2011.06695"
}

@ARTICLE{Meghir2018-bk,
  title     = "Education and mortality: Evidence from a social experiment",
  author    = "Meghir, Costas and Palme, M{\aa}rten and Simeonova, Emilia",
  journal   = "Am. Econ. J. Appl. Econ.",
  publisher = "American Economic Association",
  volume    =  10,
  number    =  2,
  pages     = "234--256",
  month     =  apr,
  year      =  2018
}

@ARTICLE{Oreopoulos2006-bn,
  title     = "Estimating Average and Local Average Treatment Effects of
               Education when Compulsory Schooling Laws Really Matter",
  author    = "Oreopoulos, Philip",
  journal   = "Am. Econ. Rev.",
  publisher = "American Economic Association",
  volume    =  96,
  number    =  1,
  pages     = "152--175",
  month     =  mar,
  year      =  2006
}

@ARTICLE{Black2005-aw,
  title     = "Why the Apple Doesn't Fall Far: Understanding Intergenerational
               Transmission of Human Capital",
  author    = "Black, Sandra E and Devereux, Paul J and Salvanes, Kjell G",
  journal   = "Am. Econ. Rev.",
  publisher = "American Economic Association",
  volume    =  95,
  number    =  1,
  pages     = "437--449",
  year      =  2005
}

@ARTICLE{Miller2012-fo,
  title     = "Does temporary affirmative action produce persistent effects? A
               study of black and female employment in law enforcement",
  author    = "Miller, Amalia R and Segal, Carmit",
  journal   = "Rev. Econ. Stat.",
  publisher = "MIT Press - Journals",
  volume    =  94,
  number    =  4,
  pages     = "1107--1125",
  month     =  nov,
  year      =  2012,
  language  = "en"
}

@ARTICLE{Lundborg2014-gm,
  title     = "Parental education and offspring outcomes: Evidence from the
               Swedish compulsory school reform",
  author    = "Lundborg, Petter and Nilsson, Anton and Rooth, Dan-Olof",
  journal   = "Am. Econ. J. Appl. Econ.",
  publisher = "American Economic Association",
  volume    =  6,
  number    =  1,
  pages     = "253--278",
  month     =  jan,
  year      =  2014
}

@ARTICLE{Imai2021-dn,
  title     = "On the Use of {Two-Way} Fixed Effects Regression Models for
               Causal Inference with Panel Data",
  author    = "Imai, Kosuke and Kim, In Song",
  journal   = "Polit. Anal.",
  publisher = "Cambridge University Press",
  volume    =  29,
  number    =  3,
  pages     = "405--415",
  month     =  jul,
  year      =  2021
}

@ARTICLE{Miller2019-ok,
  title     = "Do female officers improve law enforcement quality? Effects on
               crime reporting and domestic violence",
  author    = "Miller, Amalia R and Segal, Carmit",
  journal   = "Rev. Econ. Stud.",
  publisher = "Oxford University Press (OUP)",
  volume    =  86,
  number    =  5,
  pages     = "2220--2247",
  month     =  oct,
  year      =  2019,
  copyright = "https://academic.oup.com/journals/pages/open\_access/funder\_policies/chorus/standard\_publication\_model",
  language  = "en"
}

@ARTICLE{Johnson2019-kb,
  title    = "Reducing Inequality through Dynamic Complementarity: Evidence
              from Head Start and Public School Spending",
  author   = "Johnson, Rucker C and Jackson, C Kirabo",
  journal  = "American Economic Journal: Economic Policy",
  volume   =  11,
  number   =  4,
  pages    = "310--349",
  month    =  nov,
  year     =  2019
}

@ARTICLE{Duflo2001-nh,
  title     = "Schooling and Labor Market Consequences of School Construction
               in Indonesia: Evidence from an Unusual Policy Experiment",
  author    = "Duflo, Esther",
  journal   = "Am. Econ. Rev.",
  publisher = "American Economic Association",
  volume    =  91,
  number    =  4,
  pages     = "795--813",
  month     =  sep,
  year      =  2001
}
\bibliographystyle{econ-jet}
\newpage
\appendix
\renewcommand{\thesection}{\Alph{section}}
\section{Proof of the theorem in section \ref{sec2}}\label{ApeA}
Before we proceed the proof of Theorem \ref{sec3thm1}, we provide Lemma \ref{ApeAlemma1} below. This lemma is shown by \cite{Goodman-Bacon2021-ej}.

\begin{Lemma}[Lemma $1$ in \cite{Goodman-Bacon2021-ej}]
\label{ApeAlemma1}
The sample covariance between a cohort and time specific variable $z_{kt}$ and a double demeaning variable $\Tilde{x}_{kt}=(x_{kt}-\bar{x}_k)-(\bar{x}_t-\bar{\bar{x}})$ is equal to a sum over every pair of observations of the period-by-period products of differences between cohorts in $z_{kt}$ and $\Tilde{x}_{kt}$.
\begin{align}
&\sum_{k}n_{k}\frac{1}{T}\sum_{t}z_{kt}\left[(x_{kt}-\bar{x}_{k})-(\bar{x}_t-\bar{\bar{x}})\right]\notag\\
\label{ApeCeq1}
=&\sum_{k}\sum_{l>k}n_{l}n_{k}\frac{1}{T}\sum_{t}(z_{kt}-z_{lt})[(x_{kt}-\bar{x}_{k})-(x_{lt}-\bar{x}_{l})]
\end{align}
\end{Lemma}
\begin{proof}
See the proof of Lemma $1$ in \cite{Goodman-Bacon2021-ej}.
\end{proof}

\subsection{Proof of Theorem \ref{sec3thm1}}
\begin{proof}
From the FWL theorem, the TWFEIV estimator $\hat{\beta}_{IV}$ is:
\begin{align}
\label{ApeCeq2}
\hat{\beta}_{IV}=\frac{\frac{1}{NT}\sum_{i}\sum_{t}\Tilde{Z}_{i,t}Y_{i,t}}{\frac{1}{NT}\sum_{i}\sum_{t}\Tilde{Z}_{i,t}D_{i,t}}
\end{align}
where $\Tilde{Z}_{i,t}$ is a double-demeaning variable.\par
First, we consider the numerator of \eqref{ApeCeq2}. In the following, we use $k(i)$ to express that unit $i$ belongs to cohort $k$. We define $\bar{R}_{k(i),t}$ to be the sample mean of the random variable $R_{i,t}$ in cohort $k$ at time $t$, and define $\bar{R}_{k(i)}$ to the average of $\bar{R}_{k(i),t}$ over time:
\begin{align*}
\bar{R}_{k(i),t}=\frac{\sum_{i}R_{i,t}\mathbf{1}\{E_i=k\}}{\sum_{i}\mathbf{1}\{E_i=k\}}\hspace{2mm}\text{and}\hspace{2mm}\bar{R}_{k(i)}=\frac{1}{T}\sum_{t=1}^T \bar{R}_{k(i),t}.
\end{align*}
For the numerator of \eqref{ApeCeq2}, by adding and subtracting $(\bar{Z}_{k(i),t}-\bar{Z}_{k(i)})$, we obtain
\begin{align*}
&\frac{1}{NT}\sum_{i}\sum_{t}\Tilde{Z}_{i,t}Y_{i,t}\notag\\
&=\frac{1}{NT}\sum_{i}\sum_{t}Y_{i,t}\left[(Z_{i,t}-\bar{Z}_i)-(\bar{Z}_t-\bar{\bar{Z}})\right]\notag\\
\label{2}
&=\frac{1}{NT}\sum_{i}\sum_{t}Y_{i,t}\left[\underbrace{(Z_{i,t}-\bar{Z}_i)-(\bar{Z}_{k(i),t}-\bar{Z}_{k(i)})}_{=0}+(\bar{Z}_{k(i),t}-\bar{Z}_{k(i)})-(\bar{Z}_t-\bar{\bar{Z}})\right]\\
&=\frac{1}{NT}\sum_{i}\sum_{t}Y_{i,t}\left[(\bar{Z}_{k(i),t}-\bar{Z}_{k(i)})-(\bar{Z}_t-\bar{\bar{Z}})\right]\\
&=\sum_{k}n_{k}\frac{1}{T}\sum_{t}\bar{Y}_{k,t}\left[(\bar{Z}_{k,t}-\bar{Z}_{k})-(\bar{Z}_t-\bar{\bar{Z}})\right],
\end{align*}
where the third equality follows from the fact that $Z_{i,t}=\bar{Z}_{k(i),t}$ and $\bar{Z}_i=\bar{Z}_{k(i)}$ because all the units in cohort $k$ have the same assignment of the instrument. The forth equality follows because the expression only depends on cohort $k$ and time $t$.\par
To further develop the expression, we use Lemma \ref{ApeAlemma1}:
\begin{equation}
\sum_{k}n_{k}\frac{1}{T}\sum_{t}\bar{Y}_{kt}[(\bar{Z}_{kt}-\bar{Z}_{k})-(\bar{Z}_t-\bar{\bar{Z}})]\notag\\
\end{equation}
\begin{equation}
\label{eq330}
=\sum_{k}\sum_{l > k}n_{l}n_{k}\frac{1}{T}\sum_{t}(\bar{Y}_{kt}-\bar{Y}_{lt})[(\bar{Z}_{kt}-\bar{Z}_{k})-(\bar{Z}_{lt}-\bar{Z}_{l})].
\end{equation}
Next, we consider all possible expressions of \eqref{eq330}.
When $e=U$, that is, cohort $e$ is never exposed cohort, we have $\bar{Z}_{Ut}-\bar{Z}_{U}=0$. From this observation, for the pair $(k,U)$, we have:
\begin{align*}
&\frac{1}{T}\sum_{t}(\bar{Y}_{kt}-\bar{Y}_{Ut})\left[(\bar{Z}_{kt}-\bar{Z}_{k})-(\bar{Z}_{Ut}-\bar{Z}_{U})\right]\\
&=-\frac{1}{T}\sum_{t < k}(\bar{Y}_{kt}-\bar{Y}_{Ut})\bar{Z}_{k}+\frac{1}{T}\sum_{t \geq k}(\bar{Y}_{kt}-\bar{Y}_{Ut})(1-\bar{Z}_{k})\\
&=\left[(\bar{Y}_{kt}^{POST(k)}-\bar{Y}_{kt}^{PRE(k)})-(\bar{Y}_{Ut}^{POST(k)}-\bar{Y}_{Ut}^{PRE(k)})\right]\bar{Z}_{k}(1-\bar{Z}_{k}).
\end{align*}
By the similar argument, for the pair $(k,l)$ where $k < l <T$, we obtain
\begin{align*}
&=-\frac{1}{T}\sum_{t < k}(\bar{Y}_{kt}-\bar{Y}_{lt})(\bar{Z}_{k}-\bar{Z}_{l})+\frac{1}{T}\sum_{t \in [k,l)}(\bar{Y}_{kt}-\bar{Y}_{lt})(1-\bar{Z}_{k}+\bar{Z}_{l})-\frac{1}{T}\sum_{t \geq l}(\bar{Y}_{kt}-\bar{Y}_{lt})(\bar{Z}_{k}-\bar{Z}_{l})\\
&=-\left[(\bar{Y}_{kt}^{PRE(k)}-\bar{Y}_{lt}^{PRE(k)})\right](\bar{Z}_{k}-\bar{Z}_{l})(1-\bar{Z}_{k})+\left[(\bar{Y}_{kt}^{MID(k,l)}-\bar{Y}_{lt}^{MID(k,l)})\right](\bar{Z}_{k}-\bar{Z}_{l})(1-\bar{Z}_{k}+\bar{Z}_{l})\\
&-\left[(\bar{Y}_{kt}^{POST(l)}-\bar{Y}_{lt}^{POST(l)})\right](\bar{Z}_{k}-\bar{Z}_{l})\bar{Z}_{l}\\
&=\left[(\bar{Y}_{kt}^{MID(k,l)}-\bar{Y}_{kt}^{PRE(k)})-(\bar{Y}_{lt}^{MID(k,l)}-\bar{Y}_{lt}^{PRE(k)})\right](\bar{Z}_{k}-\bar{Z}_{l})(1-\bar{Z}_{k})\\
&+\left[(\bar{Y}_{lt}^{POST(l)}-\bar{Y}_{lt}^{MID(k,l)})-(\bar{Y}_{kt}^{POST(l)}-\bar{Y}_{kt}^{MID(k,l)})\right](\bar{Z}_{k}-\bar{Z}_{l})\bar{Z}_{l}.
\end{align*}
To sum up, for the numerator of \eqref{ApeCeq2}, we have
\begin{align}
&\frac{1}{NT}\sum_{i}\sum_{t}\Tilde{Z}_{it}Y_{it}\notag\\
&=\Big[\sum_{k \neq U}(n_k+n_u)^2n_{kU}(1-n_{kU})\bar{Z}_{k}(1-\bar{Z}_{k})\hat{\beta}_{kU}^{2\times2}\notag\\
&+\sum_{k \neq U}\sum_{l >k}[((n_k+n_l)(1-\bar{Z}_l))^2n_{kl}(1-n_{kl})\left(\frac{\bar{Z}_{k}-\bar{Z}_{l}}{1-\bar{Z}_{l}}\right)\left(\frac{1-\bar{Z}_{k}}{1-\bar{Z}_{l}}\right)\hat{\beta}_{kl}^{2\times2,k}\notag\\
&+((n_k+n_l)\bar{Z}_k)^2n_{kl}(1-n_{kl})\left(\frac{\bar{Z}_{l}}{\bar{Z}_{k}}\right)\left(\frac{\bar{Z}_{k}-\bar{Z}_{l}}{\bar{Z}_{k}}\right)\hat{\beta}_{kl}^{2\times2,l}]\Big]\notag\\
\label{ApeCeq4}
&=\bigg[\sum_{k \neq U}(n_k+n_u)^2\hat{V}_{kU}^{Z}\hat{\beta}_{kU}^{2\times2}+\sum_{k \neq U}\sum_{l >k}[((n_k+n_l)(1-\bar{Z}_l))^2\hat{V}_{kl}^{Z,k}\hat{\beta}_{kl}^{2\times2,k}+((n_k+n_l)\bar{Z}_k)^2\hat{V}_{kl}^{Z,l}\hat{\beta}_{kl}^{2\times2,l}\bigg].
\end{align}\par
Next, we consider the denominator of \eqref{ApeCeq2}. We note that the structure of the denominator is completely same as the one of the numerator in \eqref{ApeCeq2}. Therefore, by the completely same calculations, we obtain
\begin{align}
&\frac{1}{NT}\sum_{i}\sum_{t}\Tilde{Z}_{it}D_{it}\notag\\
&\equiv \hat{C}^{D,Z}\notag\\
\label{ApeCeq5}
&=\bigg[\sum_{k \neq U}(n_k+n_u)^2\hat{V}_{kU}^{Z}\hat{D}_{kU}^{2\times2}+\sum_{k \neq U}\sum_{l >k}[((n_k+n_l)(1-\bar{Z}_l))^2\hat{V}_{kl}^{Z,k}\hat{D}_{kl}^{2\times2,k}+((n_k+n_l)\bar{Z}_k)^2\hat{V}_{kl}^{Z,l}\hat{D}_{kl}^{2\times2,l}\bigg].
\end{align}
Combining \eqref{ApeCeq4} with \eqref{ApeCeq5}, we obtain
\begin{align*}
&\hat{\beta}_{IV}\\
&=\frac{\bigg[\sum_{k \neq U}(n_k+n_u)^2\hat{V}_{kU}^{Z}\hat{\beta}_{kU}^{2\times2}+\sum_{k \neq U}\sum_{l >k}[((n_k+n_l)(1-\bar{Z}_l))^2\hat{V}_{kl}^{Z,k}\hat{\beta}_{kl}^{2\times2,k}+((n_k+n_l)\bar{Z}_k)^2\hat{V}_{kl}^{Z,l}\hat{\beta}_{kl}^{2\times2,l}\bigg]}{\bigg[\sum_{k \neq U}(n_k+n_u)^2\hat{V}_{kU}^{Z}\hat{D}_{kU}^{2\times2}+\sum_{k \neq U}\sum_{l >k}[((n_k+n_l)(1-\bar{Z}_l))^2\hat{V}_{kl}^{Z,k}\hat{D}_{kl}^{2\times2,k}+((n_k+n_l)\bar{Z}_k)^2\hat{V}_{kl}^{Z,l}\hat{D}_{kl}^{2\times2,l}\bigg]}\\
&=\frac{\bigg[\sum_{k \neq U}(n_k+n_u)^2\hat{V}_{kU}^{Z}\hat{\beta}_{kU}^{2\times2}+\sum_{k \neq U}\sum_{l >k}[((n_k+n_l)(1-\bar{Z}_l))^2\hat{V}_{kl}^{Z,k}\hat{\beta}_{kl}^{2\times2,k}+((n_k+n_l)\bar{Z}_k)^2\hat{V}_{kl}^{Z,l}\hat{\beta}_{kl}^{2\times2,l}\bigg]}{\hat{C}^{D,Z}}\\
&=\bigg[\sum_{k \neq U}\hat{w}_{IV, kU}\hat{\beta}_{IV, kU}^{2\times2}+\sum_{k \neq U}\sum_{l >k}\hat{w}_{IV, kl}^{k}\hat{\beta}_{IV, kl}^{2\times2,k}+\hat{w}_{IV, kl}^{l}\hat{\beta}_{IV, kl}^{2\times2,l}\bigg],
\end{align*}
where the weights are:
\begin{align*}
\hat{w}_{IV, kU}=\frac{(n_k+n_u)^2\hat{V}_{kU}^{Z}\hat{D}_{kU}^{2 \times 2}}{\hat{C}^{D,Z}},
\end{align*}
\begin{align*}
\hat{w}_{IV, kl}^{k}=\frac{((n_k+n_l)(1-\bar{Z}_l))^2 \hat{V}_{kl}^{Z,k}\hat{D}_{kl}^{2 \times 2,k}}{\hat{C}^{D,Z}},
\end{align*}
\begin{align*}
\hat{w}_{IV, kl}^{l}=\frac{((n_k+n_l)\bar{Z}_k)^2\hat{V}_{kl}^{Z,l}\hat{D}_{kl}^{2\times2,l}}{\hat{C}^{D,Z}}.
\end{align*}
We note that $\hat{C}^{D,Z}$ is the sum of the numerator in each weight as one can see in \eqref{ApeCeq5}. This implies that the weights sum to one:
\begin{align*}
\sum_{k \neq U}w_{IV, kU}+\sum_{k \neq U}\sum_{l >k}[w_{IV, kl}^{k}+w_{IV, kl}^{l}]=1.
\end{align*}\par
Completing the proof.
\end{proof}
\setcounter{Proof of Lemma}{1}

\section{Proofs of the theorem and lemma in section \ref{sec4}.}\label{ApeB}
In this section, we first prove Lemma \ref{sec4lemma1} as a preparation.
\subsection{Proof of Lemma \ref{sec4lemma1}}
\begin{proof}
As one can see in the proof of Theorem \ref{sec3thm1}, we have the following expression for the numerator of the TWFEIV estimand:
\begin{align*}
\hat{C}^{D,Z}=\sum_{k \neq U}(n_k+n_u)^2 \hat{V}_{kU}^Z\hat{D}_{kU}^{2 \times 2}+\sum_{k \neq U}\sum_{l > k}\left[((n_k+n_l)(1-\bar{Z}_l))^2\hat{V}_{kl}^{Z,k}\hat{D}_{kl}^{2 \times 2,k}+((n_k+n_l)\bar{Z}_k)^2\hat{V}_{kl}^{Z,l}\hat{D}_{kl}^{2 \times 2,l}\right].
\end{align*}\par
We fix $T$ and consider $N \rightarrow \infty$. We first derive the probability limit of $(n_k+n_u)^2 \hat{V}_{kU}^Z\hat{D}_{kU}^{2 \times 2}$.
By definition, we can rewrite $(n_k+n_u)^2 \hat{V}_{kU}^Z\hat{D}_{kU}^{2 \times 2}$ as follows:
\begin{align*}
(n_k+n_u)^2 \hat{V}_{kU}^Z\hat{D}_{kU}^{2 \times 2}
&=n_k n_u \bar{Z}_k(1-\bar{Z}_k)\cdot\left[\left(\bar{D}_{k}^{POST(k)}-\bar{D}_{k}^{PRE(k)}\right)-\left(\bar{D}_{U}^{POST(k)}-\bar{D}_{U}^{PRE(k)}\right)\right].
\end{align*}
By the law of large number (LLN), as $N \rightarrow \infty$, we obtain
\begin{align}
&\left(\bar{D}_{k}^{POST(k)}-\bar{D}_{k}^{PRE(k)}\right)-\left(\bar{D}_{U}^{POST(k)}-\bar{D}_{U}^{PRE(k)}\right) \notag\\
&\xrightarrow{p} \frac{1}{T-(k-1)}\left[\sum_{t=k}^{T}E[D_{it}|E_i=k]-E[D_{it}|E_i=U]\right]-\frac{1}{k-1}\left[\sum_{t=1}^{k-1}E[D_{it}|E_i=k]-E[D_{it}|E_i=U]\right]\notag\\
&=\frac{1}{T-(k-1)}\left[\sum_{t=k}^{T}E[D_{it}^{k}-D_{it}^{\infty}|E_i=k]\right]\notag\\
&+\frac{1}{T-(k-1)}\left[\sum_{t=k}^{T}E[D_{it}^{\infty}|E_i=k]-\sum_{t=k}^{T}E[D_{it}^{\infty}|E_i=U]\right]\notag\\
&-\frac{1}{k-1}\left[\sum_{t=1}^{k-1}E[D_{it}^{\infty}|E_i=k]-E[D_{it}^{\infty}|E_i=U]\right]\notag\\
&=\frac{1}{T-(k-1)}\left[\sum_{t=k}^{T}E[D_{it}^{k}-D_{it}^{\infty}|E_i=k]\right]\notag\\
\label{ApeCeq6}
&=CAET_{k}^1(POST(k)).
\end{align}
The second equality follows from the simple algebra and Assumption \ref{sec2as5} (No anticipation for the first stage). The third equality follows from Assumption \ref{sec2as6} (Parallel trend assumption in the treatment).\par
Next, we consider the probability limit of $n_k n_u \bar{Z}_k(1-\bar{Z}_k)$. By the LLN and the Slutsky's theorem, we obtain
\begin{align}
\label{ApeCeq7}
n_k n_u \bar{Z}_k(1-\bar{Z}_k) \xrightarrow{p} Pr(E_i=k)Pr(E_i=U)\frac{T-(k-1)}{T}\frac{k-1}{T}.
\end{align}
Combining the result \eqref{ApeCeq6} with \eqref{ApeCeq7}, by the Slutsky's theorem, we have
\begin{align}
(n_k+n_u)^2 \hat{V}_{kU}^Z\hat{D}_{kU}^{2 \times 2}&\xrightarrow{p} Pr(E_i=k)Pr(E_i=U)\frac{T-(k-1)}{T}\frac{k-1}{T}CAET_{k,t}^1(POST(k))\notag\\
\label{ApeCeq8}
&=w_{kU}CAET_{k}^1(POST(k)).
\end{align}
where the weight $w_{kU}$ is:
\begin{align*}
w_{kU}=Pr(E_i=k)Pr(E_i=U)\frac{T-(k-1)}{T}\frac{k-1}{T}.
\end{align*}
By the completely same calculations, we also have
\begin{align}
\label{ApeCeq9}
((n_k+n_l)(1-\bar{Z}_l))^2\hat{V}_{kl}^{Z,k}\hat{D}_{kl}^{2 \times 2,k} \xrightarrow{p} w_{kl}^{k}CAET_{k}^1(MID(k,l)).
\end{align}
where the weight $w_{kl}^{k}$ is:
\begin{align*}
w_{kl}^{k}=Pr(E_i=k)Pr(E_i=l)\frac{k-1}{T}\frac{l-k}{T}.
\end{align*}
Next, we consider the probability limit of $((n_k+n_l)\bar{Z}_k)^2\hat{V}_{kl}^{Z,l}\hat{D}_{kl}^{2 \times 2,l}$. By definition, we have
\begin{align*}
((n_k+n_l)\bar{Z}_k)^2\hat{V}_{kl}^{Z,l}\hat{D}_{kl}^{2 \times 2,l}&=
n_k n_l(\bar{Z}_k-\bar{Z}_l)\bar{Z}_l\cdot\left[\left(\bar{D}_{l}^{POST(l)}-\bar{D}_{l}^{MID(k,l)}\right)-\left(\bar{D}_{k}^{POST(l)}-\bar{D}_{k}^{MID(k,l)}\right)\right].
\end{align*}
By the law of large number (LLN), as $N \rightarrow \infty$, we have
\begin{align}
&\left(\bar{D}_{l}^{POST(l)}-\bar{D}_{l}^{MID(k,l)}\right)-\left(\bar{D}_{k}^{POST(l)}-\bar{D}_{k}^{MID(k,l)}\right) \notag\\
&\xrightarrow{p} \frac{1}{T-(l-1)}\left[\sum_{t=l}^{T}E[D_{it}|E_i=l]-E[D_{it}|E_i=k]\right]-\frac{1}{l-k}\left[\sum_{t=k}^{l-1}E[D_{it}|E_i=l]-E[D_{it}|E_i=k]\right]\notag\\
&=\frac{1}{T-(l-1)}\left[\sum_{t=l}^{T}E[D_{it}^{l}-D_{it}^{\infty}|E_i=l]\right]-\frac{1}{T-(l-1)}\left[\sum_{t=l}^{T}E[D_{it}^{k}-D_{it}^{\infty}|E_i=k]\right]\notag\\
&+\frac{1}{T-(l-1)}\left[\sum_{t=l}^{T}E[D_{it}^{\infty}|E_i=l]-\sum_{t=l}^{T}E[D_{it}^{\infty}|E_i=k]\right]\notag\\
&+\frac{1}{l-k}\left[\sum_{t=k}^{l-1}E[D_{it}^{k}|E_i=k]-E[D_{it}^{\infty}|E_i=k]\right]\notag\\
&+\frac{1}{l-k}\left[\sum_{t=k}^{l-1}E[D_{it}^{\infty}|E_i=k]-E[D_{it}^{\infty}|E_i=l]\right]\notag\\
&=\frac{1}{T-(l-1)}\left[\sum_{t=l}^{T}E[D_{it}^{l}-D_{it}^{\infty}|E_i=l]\right]-\frac{1}{T-(l-1)}\left[\sum_{t=l}^{T}E[D_{it}^{k}-D_{it}^{\infty}|E_i=k]\right]\notag\\
&+\frac{1}{l-k}\left[\sum_{t=k}^{l-1}E[D_{it}^{k}|E_i=k]-E[D_{it}^{\infty}|E_i=k]\right]\notag\\
\label{ApeCeq10}
&=CAET_{l}^{1}(POST(l))-CAET_{k}^{1}(POST(l))+CAET_{k}^{1}(MID(k,l))
\end{align}
The second equality follows from the simple algebra and Assumption \ref{sec2as5}. The third equality follows from Assumption.\par
Note that the LLN and the Slutsky's theorem implies
\begin{align}
\label{ApeCeq11}
n_k n_l(\bar{Z}_k-\bar{Z}_l)\bar{Z}_l \xrightarrow{p} Pr(E_i=k)Pr(E_i=l)\frac{T-(l-1)}{T}\frac{l-k}{T}.
\end{align}
From the result of \eqref{ApeCeq10} with \eqref{ApeCeq11}, we obtain
\begin{align}
&((n_k+n_l)\bar{Z}_k)^2\hat{V}_{kl}^{Z,l}\hat{D}_{kl}^{2 \times 2,l}\notag\\
\label{ApeCeq12}
\xrightarrow{p}& w_{kl}^{l}\left[CAET_{l}^{1}(POST(l))-CAET_{k}^{1}(POST(l))+CAET_{k}^{1}(MID(k,l))\right].
\end{align}
where the weight $w_{kl}^{l}$ is:
\begin{align*}
w_{kl}^{l}=Pr(E_i=k)Pr(E_i=l)\frac{T-(l-1)}{T}\frac{l-k}{T}.
\end{align*}
To sum up, by combining \eqref{ApeCeq8},\eqref{ApeCeq9} with \eqref{ApeCeq12}, we obtain
\begin{align*}
\hat{C}^{D,Z}&=\sum_{k \neq U}(n_k+n_u)^2 \hat{V}_{kU}^Z\hat{D}_{kU}^{2 \times 2}+\sum_{k \neq U}\sum_{l > k}\left[((n_k+n_l)(1-\bar{Z}_l))^2\hat{V}_{kl}^{Z,k}\hat{D}_{kl}^{2 \times 2,k}+((n_k+n_l)\bar{Z}_k)^2\hat{V}_{kl}^{Z,l}\hat{D}_{kl}^{2 \times 2,l}\right]\\
&\xrightarrow{p}  WCAET-\Delta CAET^{1}.
\end{align*}
where we define
\begin{align*}
&WCAET \equiv \sum_{k \neq U}w_{kU}CAET_{k}^1(POST(k))+\sum_{k \neq U}\sum_{l > k}w_{kl}^{k}CAET_{k}^1(MID(k,l))+w_{kl}^{l}CAET_{l}^{1}(POST(l)),\\
&\Delta CAET^{1} \equiv w_{kl}^{l}\left[CAET_{k}^{1}(POST(l))-CAET_{k}^{1}(MID(k,l))\right].
\end{align*}
and the weights are:
\begin{align}
\label{Lemma1weight1}
&w_{kU}=Pr(E_i=k)Pr(E_i=U)\frac{T-(k-1)}{T}\frac{k-1}{T},\\
\label{Lemma1weight2}
&w_{kl}^{k}=Pr(E_i=k)Pr(E_i=l)\frac{k-1}{T}\frac{l-k}{T},\\
\label{Lemma1weight3}
&w_{kl}^{l}=Pr(E_i=k)Pr(E_i=l)\frac{T-(l-1)}{T}\frac{l-k}{T}.
\end{align}\par
Completing the proof.\par
\end{proof}
\subsection{Preparation for the proof of Theorem \ref{sec4thm2}.}
Before we present the proof of Theorem \ref{sec4thm2}, we show the following lemma.
\begin{Lemma}
\label{ApeBlemma2}
Suppose Assumptions \ref{sec2as1}-\ref{sec2as7} hold. If the treatment is binary, for all $k, l \in \{1,\dots,T\} (k \leq l)$, we have
\begin{align*}
\left[\sum_{t=k}^{l}E[Y_{i,t}(D_{i,t}^{k})-Y_{i,t}(D_{i,t}^{\infty})|E_i=k]\right]&=\sum_{t=k}^{l}E[D_{i,t}^{k}-D_{i,t}^{\infty}|E_i=k]\cdot E[Y_{i,t}(1)-Y_{i,t}(0)|E_i=k,CM_{k,t}]\\
&\equiv \sum_{t=k}^{l}CAET^{1}_{k,t}\cdot CLATT_{k,t}.
\end{align*}
\end{Lemma}
\begin{proof}
Because we assume a binary treatment, we have
\begin{align*}
\left[\sum_{t=k}^{l}E[Y_{i,t}(D_{i,t}^{k})-Y_{i,t}(D_{i,t}^{\infty})|E_i=k]\right]&=\left[\sum_{t=k}^{l}E[(D_{i,t}^{k}-D_{i,t}^{\infty})\cdot(Y_{i,t}(1)-Y_{i,t}(0))|E_i=k]\right]\\
&=\sum_{t=k}^{l}E[D_{i,t}^{k}-D_{i,t}^{\infty}|E_i=k]\cdot E[Y_{i,t}(1)-Y_{i,t}(0)|E_i=k,CM_{k,t}],
\end{align*}
where the first equality holds from Assumptions \ref{sec2as1}-\ref{sec2as3} and the second equality holds from Assumption \ref{sec2as4}.
\end{proof}
\subsection{Proof of Theorem \ref{sec4thm2}.}
\begin{proof}
We fix $T$ and consider $N \rightarrow \infty$. We first derive the probability limit of $\hat{w}_{IV, kU}\hat{\beta}_{IV, kU}^{2\times2}$.
We note that $\hat{w}_{IV, kU}\hat{\beta}_{IV, kU}^{2\times2}$ is written as follows:
\begin{align*}
\hat{w}_{IV, kU}\hat{\beta}_{IV, kU}^{2\times2}&=\frac{(n_k+n_u)^2\hat{V}_{kU}^{Z}\hat{D}_{kU}^{2\times2}}{\hat{C}^{D,Z}}\cdot\frac{\hat{\beta}_{kU}^{2\times2}}{\hat{D}_{kU}^{2\times2}}.
\end{align*}
We first consider the probability limit of $\displaystyle\frac{\hat{\beta}_{kU}^{2\times2}}{\hat{D}_{kU}^{2\times2}}$. Recall that we have already derived the probability limit of the denominator in the proof of Lemma \ref{sec4lemma1}:
\begin{align}
\label{ApeCeq13}
\hat{D}_{kU}^{2\times2} \xrightarrow{p} \frac{1}{T-(k-1)}\sum_{t=k}^T CAET^{1}_{k,t}.
\end{align}
We consider the numerator $\hat{\beta}_{kU}^{2\times2}$. By the law of large number (LLN), as $N \rightarrow \infty$, we obtain
\begin{align}
&\left(\bar{Y}_{k}^{POST(k)}-\bar{Y}_{k}^{PRE(k)}\right)-\left(\bar{Y}_{U}^{POST(k)}-\bar{Y}_{U}^{PRE(k)}\right) \notag\\
&\xrightarrow{p} \frac{1}{T-(k-1)}\left[\sum_{t=k}^{T}E[Y_{i,t}|E_i=k]-E[Y_{i,t}|E_i=U]\right]-\frac{1}{k-1}\left[\sum_{t=1}^{k-1}E[Y_{i,t}|E_i=k]-E[Y_{i,t}|E_i=U]\right]\notag\\
&=\frac{1}{T-(k-1)}\left[\sum_{t=k}^{T}E[Y_{i,t}(D_{i,t}^{k})-Y_{i,t}(D_{i,t}^{\infty})|E_i=k]\right]\notag\\
&+\frac{1}{T-(k-1)}\left[\sum_{t=k}^{T}E[Y_{i,t}(D_{i,t}^{\infty})|E_i=k]-\sum_{t=k}^{T}E[Y_{i,t}(D_{i,t}^{\infty})|E_i=U]\right]\notag\\
&-\frac{1}{k-1}\left[\sum_{t=1}^{k-1}E[Y_{i,t}(D_{i,t}^{\infty})|E_i=k]-E[Y_{i,t}(D_{i,t}^{\infty})|E_i=U]\right]\notag\\
&=\frac{1}{T-(k-1)}\left[\sum_{t=k}^{T}E[Y_{i,t}(D_{i,t}^{k})-Y_{i,t}(D_{i,t}^{\infty})|E_i=k]\right]\notag\\
\label{ApeCeq14}
&=\frac{1}{T-(k-1)}\sum_{t=k}^{T}CAET^{1}_{k,t}\cdot CLATT_{k,t}.
\end{align}
The first equality follows from the simple manipulation, Assumption \ref{sec2as3} (Exclusion restriction in multiple time periods) and Assumption \ref{sec2as5}. The second equality follows from Assumption \ref{sec2as7} (Parallel trend assumption in the outcome). The final equality follows from Lemma \ref{ApeBlemma2}.\par
Combining the result \eqref{ApeCeq13} with \eqref{ApeCeq14}, we obtain
\begin{align}
\displaystyle\frac{\hat{\beta}_{kU}^{2\times2}}{\hat{D}_{kU}^{2\times2}} &\xrightarrow{p} \sum_{t=k}^T \frac{CAET^{1}_{k,t}}{\sum_{t=k}^T CAET^{1}_{k,t}}\cdot CLATT_{k,t}\notag\\
\label{ApeCeq15}
&=CLATT^{CM}_k(POST(k)).
\end{align}
Next, we consider the probability limit of $\displaystyle\frac{(n_k+n_u)^2\hat{V}_{kU}^{Z}\hat{D}_{kU}^{2\times2}}{\hat{C}^{D,Z}}$. By the LLN and the Slutsky's theorem, we obtain
\begin{align}
\label{ApeCeq16}
\frac{(n_k+n_u)^2\hat{V}_{kU}^{Z}\hat{D}_{kU}^{2\times2}}{\hat{C}^{D,Z}} \xrightarrow{p} \frac{Pr(E_i=k)Pr(E_i=U)}{C^{D,Z}}\frac{T-(k-1)}{T}\frac{k-1}{T}CAET^1_{k}(POST(k)).
\end{align}
Here $C^{D,Z}$ is the probability limit of $\hat{C}^{D,Z}$ and its specific expression is already derived in Lemma \ref{sec4lemma1}.\par
Combining the result \eqref{ApeCeq15} with \eqref{ApeCeq16}, by the Slutsky's theorem, we have
\begin{align}
\label{ApeCeq17}
\hat{w}_{IV, kU}\hat{\beta}_{IV, kU}^{2\times2}\xrightarrow{p} 
w_{IV,kU}CLATT^{CM}_{k}(POST(k)).
\end{align}
where the weight $w_{IV,kU}$ is:
\begin{align*}
w_{IV,kU}=\frac{Pr(E_i=k)Pr(E_i=U)}{C^{D,Z}}\frac{T-(k-1)}{T}\frac{k-1}{T}\cdot CAET^1_{k}(POST(k)).
\end{align*}
By the completely same argument, we also obtain
\begin{align}
\label{ApeCeq18}
\hat{w}_{IV, kl}^{k}\hat{\beta}_{IV, kl}^{2\times2,k} \xrightarrow{p} w_{IV,kl}^{k}CLATT^{CM}_{k}(MID(k,l)).
\end{align}
where the weight $w_{IV,kl}^{k}$ is:
\begin{align*}
w_{IV,kl}^{k}=\frac{Pr(E_i=k)Pr(E_i=l)}{C^{D,Z}}\frac{k-1}{T}\frac{l-k}{T}\cdot CAET_{k}^1(MID(k,l)).
\end{align*}
Next, we derive the probability limit of $\hat{w}_{IV, kl}^{l}\hat{\beta}_{IV, kl}^{2\times2,l}$. Recall that $\hat{w}_{IV, kl}^{l}\hat{\beta}_{IV, kl}^{2\times2,l}$ is:
\begin{align*}
\hat{w}_{IV, kl}^{l}\hat{\beta}_{IV, kl}^{2\times2,l}&=\frac{((n_k+n_l)\bar{Z}_k)^2\hat{V}_{kl}^{Z,l}\hat{D}_{kl}^{2\times2,l}}{\hat{C}^{D,Z}}\cdot\frac{\hat{\beta}_{kl}^{2\times2,l}}{\hat{D}_{kl}^{2\times2,l}}.
\end{align*}
First note that in the proof of Lemma \ref{sec4lemma1}, we have already derived the probability limit of $\hat{D}_{kl}^{2\times2,l}$:
\begin{align}
\label{ApeCeq19}
\hat{D}_{kl}^{2\times2,l} &\xrightarrow{p} CAET_{l}^1(POST(l))-\left[CAET_{k}^1(POST(l))-CAET_{k}^1(MID(k,l))\right]\\
&\equiv D_{kl}^{2\times2,l}\notag.
\end{align}
Here, to ease the notation, we define $D_{kl}^{2\times2,l}$ to be the probability limit of $\hat{D}_{kl}^{2\times2,l}$.\par
Next, we consider the probability limit of $\hat{\beta}_{kl}^{2\times2,l}$.\par
By the law of large number (LLN), as $N \rightarrow \infty$, we have
\begin{align}
&\hat{\beta}_{kl}^{2\times2,l}=\left(\bar{Y}_{l}^{POST(l)}-\bar{Y}_{l}^{MID(k,l)}\right)-\left(\bar{Y}_{k}^{POST(l)}-\bar{Y}_{k}^{MID(k,l)}\right) \notag\\
&\xrightarrow{p} \frac{1}{T-(l-1)}\left[\sum_{t=l}^{T}E[Y_{i,t}|E_i=l]-E[Y_{i,t}|E_i=k]\right]-\frac{1}{l-k}\left[\sum_{t=k}^{l-1}E[Y_{i,t}|E_i=l]-E[Y_{i,t}|E_i=k]\right]\notag\\
&=\frac{1}{T-(l-1)}\left[\sum_{t=l}^{T}E[Y_{i,t}(D_{i,t}^{l})-Y_{i,t}(D_{i,t}^{\infty})|E_i=l]\right]-\frac{1}{T-(l-1)}\left[\sum_{t=l}^{T}E[Y_{i,t}(D_{i,t}^{k})-Y_{i,t}(D_{i,t}^{\infty})|E_i=k]\right]\notag\\
&+\frac{1}{T-(l-1)}\left[\sum_{t=l}^{T}E[Y_{i,t}(D_{i,t}^{\infty})|E_i=l]-\sum_{t=l}^{T}E[Y_{i,t}(D_{i,t}^{\infty})|E_i=k]\right]\notag\\
&+\frac{1}{l-k}\left[\sum_{t=k}^{l-1}E[Y_{i,t}(D_{i,t}^{k})|E_i=k]-E[Y_{i,t}(D_{i,t}^{\infty})|E_i=k]\right]\notag\\
&+\frac{1}{l-k}\left[\sum_{t=k}^{l-1}E[Y_{i,t}(D_{i,t}^{\infty})|E_i=k]-E[Y_{i,t}(D_{i,t}^{\infty})|E_i=l]\right]\notag\\
&=\frac{1}{T-(l-1)}\left[\sum_{t=l}^{T}E[Y_{i,t}(D_{i,t}^{l})-Y_{i,t}(D_{i,t}^{\infty})|E_i=l]\right]-\frac{1}{T-(l-1)}\left[\sum_{t=l}^{T}E[Y_{i,t}(D_{i,t}^{k})-Y_{i,t}(D_{i,t}^{\infty})|E_i=k]\right]\notag\\
&+\frac{1}{l-k}\left[\sum_{t=k}^{l-1}E[Y_{i,t}(D_{i,t}^{k})|E_i=k]-E[Y_{i,t}(D_{i,t}^{\infty})|E_i=k]\right]\notag\\
&=\frac{1}{T-(l-1)}\sum_{t=l}^{T}CAET_{l,t}^1\cdot CLATT_{l,t}\notag\\
\label{ApeCeq20}
&-\left[\frac{1}{T-(l-1)}\sum_{t=l}^{T}CAET_{k,t}^1\cdot CLATT_{k,t}-\frac{1}{l-k}\sum_{t=k}^{l-1}CAET_{k,t}^1\cdot CLATT_{k,t}\right].
\end{align}
The first equality follows from the simple algebra, Assumption \ref{sec2as3}, and Assumption \ref{sec2as5}. The second equality follows from Assumption \ref{sec2as7}. The final equality follows from Lemma \ref{ApeBlemma2}.\par
Note that the LLN and the Slutsky's theorem yields
\begin{align}
\label{ApeCeq21}
\frac{((n_k+n_l)\bar{Z}_k)^2\hat{V}_{kl}^{Z,l}\hat{D}_{kl}^{2\times2,l}}{\hat{C}^{D,Z}} &\xrightarrow{p} \frac{Pr(E_i=k)Pr(E_i=l)}{C^{D,Z}}\frac{T-(l-1)}{T}\frac{l-k}{T}\cdot D_{kl}^{2\times2,l}
\end{align}
From the results of \eqref{ApeCeq19} and \eqref{ApeCeq20} with \eqref{ApeCeq21}, we obtain
\begin{align}
\hat{w}_{IV, kl}^{l}\hat{\beta}_{IV, kl}^{2\times2,l} &\xrightarrow{p} \sigma_{IV,kl}^{l}\cdot \Bigg\{\frac{1}{T-(l-1)}\sum_{t=l}^{T}CAET_{l,t}^1\cdot CLATT_{l,t}\notag\\
&-\left[\frac{1}{T-(l-1)}\sum_{t=l}^{T}CAET_{k,t}^1\cdot CLATT_{k,t}-\frac{1}{l-k}\sum_{t=k}^{l-1}CAET_{k,t}^1\cdot CLATT_{k,t}\right]\Bigg\}\notag\\
&=\sigma_{IV,kl}^{l}\cdot\Bigg\{CLATT_l(POST(l))\notag\\
\label{ApeCeq22}
&-\left[CLATT_k(POST(l))-CLATT_k(MID(k,l))\right]\Bigg\}.
\end{align}
where the weight $\sigma_{IV,kl}^{l}$ is:
\begin{align}
\label{sigma1}
\sigma_{IV,kl}^{l}=\frac{Pr(E_i=k)Pr(E_i=l)}{C^{D,Z}}\frac{T-(l-1)}{T}\frac{l-k}{T}.
\end{align}
To sum up, by combining \eqref{ApeCeq17},\eqref{ApeCeq18} with \eqref{ApeCeq22}, we obtain
\begin{align*}
\hat{\beta}_{IV}&=\bigg[\sum_{k \neq U}w_{IV, kU}\hat{\beta}_{IV, kU}^{2\times2}+\sum_{k \neq U}\sum_{l >k}w_{IV, kl}^{k}\hat{\beta}_{IV, kl}^{2\times2,k}+w_{IV, kl}^{l}\hat{\beta}_{IV, kl}^{2\times2,l}\bigg]\\
&\xrightarrow{p} WCLATT-\Delta CLATT.
\end{align*}
where we define
\begin{align*}
WCLATT &\equiv \sum_{k \neq U}w_{IV,kU}CLATT^{CM}_{k}(POST(k))+\sum_{k \neq U}\sum_{l >k}w_{IV,kl}^{k}CLATT^{CM}_{k}(MID(k,l))\\
&+\sum_{k \neq U}\sum_{l >k}\sigma_{IV,kl}^{l}\cdot CLATT_l(POST(l)),\\
\Delta CLATT &\equiv \sum_{k \neq U}\sum_{l >k}\sigma_{IV,kl}^{l}\cdot \left[CLATT_k(POST(l))-CLATT_k(MID(k,l))\right].
\end{align*}
and the weights are:
\begin{align}
\label{sec2thm2weight1}
&w_{IV,kU}=\frac{Pr(E_i=k)Pr(E_i=U)}{C^{D,Z}}\frac{T-(k-1)}{T}\frac{k-1}{T}\cdot CAET^1_{k}(POST(k)),\\
\label{sec2thm2weight2}
&w_{IV,kl}^{k}=\frac{Pr(E_i=k)Pr(E_i=l)}{C^{D,Z}}\frac{k-1}{T}\frac{l-k}{T}\cdot CAET_{k}^1(MID(k,l)),\\
\label{sec2thm2weight3}
&\sigma_{IV,kl}^{l}=\frac{Pr(E_i=k)Pr(E_i=l)}{C^{D,Z}}\frac{T-(l-1)}{T}\frac{l-k}{T}.
\end{align}\par
Completing the proof.
\end{proof}

\subsection{Proof of Lemma \ref{sec4lemma2}}
\begin{proof}
We first simplify $CLATT_{k}^{CM}(W)$ and $D_{kl}^{2 \times 2,l}$ (defined in \eqref{ApeCeq19}) under Assumption \ref{sec4as2}. If we assume Assumption \ref{sec4as2}, $CLATT_{k}^{CM}(W)$ is:
\begin{align*}
CLATT^{CM}_k(W) &=\sum_{t \in W}\frac{Pr(CM_{k,t}|E_i=k)}{\sum_{t \in W}Pr(CM_{k,t}|E_i=k)}CLATT_{k,t}\\
&=\frac{1}{T_W}\sum_{t \in W}CLATT_{k,t}\\
&\equiv CLATT^{eq}_k(W).
\end{align*}
In addition, $D_{kl}^{2 \times 2,l}$ is:
\begin{align*}
D_{kl}^{2 \times 2,l}&=CAET_{l}^1(POST(l))-\left[CAET_{k}^1(POST(l))-CAET_{k}^1(MID(k,l))\right]\\
&=CAET_{l}^1.
\end{align*}
because we have $CAET_k^1(W)=CAET_k^1$.\par
We then rewrite the probability limit of $\hat{w}_{IV, kU}\hat{\beta}_{IV, kU}^{2\times2}$, $\hat{w}_{IV, kl}^{k}\hat{\beta}_{IV, kl}^{2\times2,k}$ and $\hat{w}_{IV, kl}^{l}\hat{\beta}_{IV, kl}^{2\times2,l}$ respectively. First, the probability limit of $\hat{w}_{IV, kU}\hat{\beta}_{IV, kU}^{2\times2}$ and $\hat{w}_{IV, kl}^{k}\hat{\beta}_{IV, kl}^{2\times2,k}$ is simplified to: 
\begin{align}
\label{ApeCeq23}
\hat{w}_{IV, kU}\hat{\beta}_{IV, kU}^{2\times2}&\xrightarrow{p} 
w_{IV,kU}CLATT^{CM}_{k}(POST(k))\notag\\
&=w_{IV,kU}CLATT^{eq}_{k}(POST(k)).
\end{align}
\begin{align}
\label{ApeCeq24}
\hat{w}_{IV, kl}^{k}\hat{\beta}_{IV, kl}^{2\times2,k} &\xrightarrow{p} 
w_{IV,kl}^{k}CLATT^{CM}_{k}(MID(k,l))\notag\\
&=w_{IV,kl}^{k}CLATT^{eq}_{k}(MID(k,l)).
\end{align}
where the weights $w_{IV,kU}$, $w_{IV,kl}^{k}$ are:
\begin{align}
\label{lemma2weight1}
&w_{IV,kU}=\frac{Pr(E_i=k)Pr(E_i=U)}{C^{D,Z}}\frac{T-(k-1)}{T}\frac{k-1}{T}\cdot CAET^1_{k},\\
\label{lemma2weight2}
&w_{IV,kl}^{k}=\frac{Pr(E_i=k)Pr(E_i=l)}{C^{D,Z}}\frac{k-1}{T}\frac{l-k}{T}\cdot CAET^1_{k}.
\end{align}
Next, we reconsider the probability limit of $\hat{w}_{IV, kl}^{l}\hat{\beta}_{IV, kl}^{2\times2,l}$.\par First, we note that the probability limit of $\hat{w}_{IV, kl}^{l}$ is simplified to:
\begin{align}
\hat{w}_{IV, kl}^{l}=\frac{((n_k+n_l)\bar{Z}_k)^2\hat{V}_{kl}^{Z,l}\hat{D}_{kl}^{2\times2,l}}{\hat{C}^{D,Z}} &\xrightarrow{p} \frac{Pr(E_i=k)Pr(E_i=l)}{C^{D,Z}}\frac{T-(l-1)}{T}\frac{l-k}{T}\cdot D_{kl}^{2\times2,l}\notag\\
\label{ApeCeq25}
&=\frac{Pr(E_i=k)Pr(E_i=l)}{C^{D,Z}}\frac{T-(l-1)}{T}\frac{l-k}{T}\cdot CAET_{l}^1.
\end{align}
Here the second equality follows from $D_{kl}^{2 \times 2,l}=CAET_{l}^1$.\par
Second, the probability limit of $\hat{\beta}_{IV, kl}^{2\times2,l}$ simply reduces to:
\begin{align}
\hat{\beta}_{IV, kl}^{2\times2,l}&=\frac{\hat{\beta}_{kl}^{2\times2,l}}{\hat{D}_{kl}^{2\times2,l}}\notag\\
&\xrightarrow{p} \frac{1}{CAET_{l}^1}\cdot \frac{1}{T-(l-1)}\sum_{t=l}^{T}CAET_{l}^1\cdot CLATT_{l,t}\notag\\
&-\frac{1}{CATT_{l}^1}\cdot\left[\frac{1}{T-(l-1)}\sum_{t=l}^{T}CAET_{k}^1\cdot CLATT_{k,t}-\frac{1}{l-k}\sum_{t=k}^{l-1}CAET_{k}^1\cdot CLATT_{k,t}\right]\notag\\
&=CLATT^{eq}_{l}(POST(l))\notag\\
\label{ApeCeq26}
&-\frac{1}{CAET_{l}^1}\cdot\left[CLATT_k(POST(l))-CLATT_k(MID(k,l))\right].
\end{align}
where we use \eqref{ApeCeq20} and $D_{kl}^{2 \times 2,l}=CAET_{l}^1$.\par
Combining the result \eqref{ApeCeq26} with \eqref{ApeCeq25}, by the Slutsky's theorem, we have
\begin{align}
\label{ApeCeq27}
\hat{w}_{IV, kl}^{l}\hat{\beta}_{IV, kl}^{2\times2,l} &\xrightarrow{p} w_{IV, kl}^{l}CLATT^{eq}_{l}(POST(l))-\sigma_{IV, kl}^{l}\cdot\left[CLATT_k(POST(l))-CLATT_k(MID(k,l))\right].
\end{align}
where the weight $w_{IV, kl}^{l}$ is:
\begin{align}
\label{lemma2weight3}
w_{IV, kl}^{l}=\frac{Pr(E_i=k)Pr(E_i=l)}{C^{D,Z}}\frac{T-(l-1)}{T}\frac{l-k}{T}\cdot CAET_{l}^1,
\end{align}
and $\sigma_{IV, kl}^{l}$ is already defined in \eqref{sigma1}.\par
Finally, from the result \eqref{ApeCeq23} and \eqref{ApeCeq24} with \eqref{ApeCeq27}, we obtain
\begin{align*}
\hat{\beta}_{IV}&=\bigg[\sum_{k \neq U}\hat{w}_{IV, kU}\hat{\beta}_{IV, kU}^{2\times2}+\sum_{k \neq U}\sum_{l >k}\hat{w}_{IV, kl}^{k}\hat{\beta}_{IV, kl}^{2\times2,k}+\hat{w}_{IV, kl}^{l}\hat{\beta}_{IV, kl}^{2\times2,l}\bigg]\\
&\xrightarrow{p} WCLATT-\Delta CLATT.
\end{align*}
where we define:
\begin{align*}
WCLATT &\equiv \sum_{k \neq U}w_{IV,kU}CLATT^{eq}_{k}(POST(k))+\sum_{k \neq U}\sum_{l >k}w_{IV,kl}^{k}CLATT^{eq}_{k}(MID(k,l))\\
&+\sum_{k \neq U}\sum_{l >k}w_{IV, kl}^{l} CLATT^{eq}_l(POST(l)),\\
\Delta CLATT &\equiv \sum_{k \neq U}\sum_{l >k}\sigma_{IV,kl}^{l}\cdot \left[CLATT_k(POST(l))-CLATT_k(MID(k,l))\right].
\end{align*}\par
Completing the proof.
\end{proof}
\subsection{Proof of Lemma \ref{sec4lemma3}}
\begin{proof}
First, we show $\Delta CLATT=0$. Under Assumption \ref{sec4as2} and Assumption \ref{sec4as4}, we have $CLATT_{e,t}=CLATT_{e}$. This implies:
\begin{align*}
\Delta CLATT &= \sum_{k \neq U}\sum_{l >k}\sigma_{IV,kl}^{l}\cdot \left[CLATT_k(POST(l))-CLATT_k(MID(k,l))\right]\\
&=0.
\end{align*}
because we have $CLATT_k(POST(l))-CLATT_k(MID(k,l))=0$.\par
Next, we consider $WCLATT$. Since we have $CLATT^{eq}_{k}(W)=CLATT_{k}$, we obtain:
\begin{align*}
WCLATT &= \sum_{k \neq U}w_{IV,kU}CLATT_{k}+\sum_{k \neq U}\sum_{l >k}w_{IV,kl}^{k}CLATT_{k}+\sum_{k \neq U}\sum_{l >k}w_{IV, kl}^{l} CLATT_l\\
&=\sum_{k \neq U}CLATT_{k}\Bigg[w_{IV,kU}+\sum_{j=1}^{k-1}w_{IV,jk}^{k}
+\sum_{j=k+1}^{K}w_{IV,kj}^{k}\Bigg].
\end{align*}\par
Completing the proof.
\end{proof}

\section{Extensions in section \ref{sec5}}\label{ApeC}
\subsection{Non-binary, ordered treatment}
This subsection considers a non binary, ordered treatment. We show Lemma \ref{ApeClemma1} below that is analogous to Lemma \ref{ApeBlemma2} in a binary treatment. If we use Lemma \ref{ApeClemma1} instead of Lemma \ref{ApeBlemma2} in the proof of Theorem \ref{sec4thm2} and Lemmas \ref{sec4lemma2}-\ref{sec4lemma3}, we obtain the theorem and the lemmas which replace $CLATT_{e,k}$ with $CACRT_{e,k}$.
\begin{Lemma}
\label{ApeClemma1}
Suppose Assumptions \ref{sec2as1}-\ref{sec2as7} hold. If treatment is a non-binary, ordered, for all $k, l \in \{1,\dots,T\} (k \leq l)$, we have
\begin{align*}
&\left[\sum_{t=k}^{l}E[Y_{i,t}(D_{i,t}^{k})-Y_{i,t}(D_{i,t}^{\infty})|E_i=k]\right]\\
&=\sum_{t=k}^{l}E[D_{i,t}^{k}-D_{i,t}^{\infty}|E_i=k]\cdot \sum_{j=1}^{J}w^{k}_{t,j} \cdot E[Y_{i,t}(j)-Y_{i,t}(j-1)|E_i=k, D_{i,t}^{k} \geq j > D_{i,t}^{\infty}]\\
&\equiv \sum_{t=k}^{l}CATT^{1}_{k,t}\cdot CACRT_{k,t}.
\end{align*}
where the weight $w^{k}_{t,j}$ is:
\begin{align*}
w^{k}_{t,j}=\frac{Pr(D_{i,t}^{k} \geq j > D_{i,t}^{\infty}|E_i=k)}{\sum_{j=1}^{J} Pr(D_{i,t}^{k} \geq j > D_{i,t}^{\infty}|E_i=k)}.
\end{align*}
\end{Lemma}
\begin{proof}
By the similar argument in the proof of lemma \ref{ApeBlemma2}, one can show that
\begin{align}
    &\left[\sum_{t=k}^{l}E[Y_{i,t}(D_{i,t}^{k})-Y_{i,t}(D_{i,t}^{\infty})|E_i=k]\right]\notag\\
    \label{ApeCnon1}
    &=\left[\sum_{t=k}^{l}\sum_{j=1}^{J}Pr(D_{i,t}^{k} \geq j > D_{i,t}^{\infty}|E_i=k) \cdot E[Y_{i,t}(j)-Y_{i,t}(j-1)|E_i=k, D_{i,t}^{k} \geq j > D_{i,t}^{\infty}]\right]\\
    \notag\\
    &E[D_{i,t}^{k}-D_{i,t}^{\infty}|E_i=k]\notag\\
    \label{ApeCnon2}
    &=\sum_{j=1}^{J}Pr(D_{i,t}^{k} \geq j > D_{i,t}^{\infty}|E_i=k).
\end{align}
Combining the result \eqref{ApeCnon1} with \eqref{ApeCnon2}, we obtain the desired result.
\end{proof}
\subsection{Unbalanced panel case}
In this section, we consider an unbalanced setting. We use the notation for a panel data setting, but the discussions and the results are the same if we consider an unbalanced repeated cross section setting.\par
\subsubsection*{Proof of Theorem \ref{sec5thm1}}
Let $N_{e,t}$ be the sample size for cohort $e$ at time $t$ and $N=\sum_{e}\sum_{t}N_{e,t}$ be the total number of observations. We consider the following two way fixed effects instrumental variable regression:
\begin{align*}
    &Y_{i,t}=\mu_{i.}+\delta_{t.}+\alpha Z_{i,t}+\epsilon_{i,t},\\
    &D_{i,t}=\gamma_{i.}+\zeta_{t.}+\pi Z_{i,t}+\eta_{i,t}.
\end{align*}
We define $\hat{Z}_{i,t}$ to be the residuals from regression ${Z}_{i,t}$ on the time and individual fixed effects.\par
From the FWL theorem, the TWFEIV estimator $\hat{\beta}_{IV}$ is:
\begin{align*}
\hat{\beta}_{IV}&=\frac{\sum_{i}\sum_{t}\hat{Z}_{i,t}Y_{i,t}}{\sum_{i}\sum_{t}\hat{Z}_{i,t}D_{i,t}}\\
&=\frac{\sum_{e}\sum_{t}N_{e,t}\frac{1}{N_{e,t}}\sum_{i}^{N_{e,t}}\hat{Z}_{e(i),t}Y_{e(i),t}}{\sum_{e}\sum_{t}N_{e,t}\frac{1}{N_{e,t}}\sum_{i}^{N_{e,t}}\hat{Z}_{e(i),t}D_{e(i),t}}\\
&=\frac{\sum_{e}\sum_{t}N_{e,t}\hat{Z}_{e,t}\frac{1}{N_{e,t}}\sum_{i}^{N_{e,t}}Y_{e(i),t}}{\sum_{e}\sum_{t}N_{e,t}\hat{Z}_{e,t}\frac{1}{N_{e,t}}\sum_{i}^{N_{e,t}}D_{e(i),t}},
\end{align*}
where the third equality follows from the fact that $\hat{Z}_{i,t}$ only varies across cohort and time level.\par
We note that by the definition of $\hat{Z}_{e,t}$, we have
\begin{align}
\label{ApeDeq1}
&\sum_{t}N_{e,t}\hat{Z}_{e,t}=0\hspace{3mm}\text{for all}\hspace{3mm}e \in \mathcal{S}(E_i),\\
\label{ApeDeq2}
&\sum_{e}N_{e,t}\hat{Z}_{e,t}=0\hspace{3mm}\text{for all}\hspace{3mm}t \in \{1,\dots,T\}.
\end{align}\par
To ease the notation, we define the sample mean for a random variable $R_{i,t}$ in cohort $e$ at time $t$ as follows:
\begin{align*}
R_{e,t} \equiv \frac{1}{N_{e,t}}\sum_{i}^{N_{e,t}}R_{e(i),t}.
\end{align*}
Here, we note that we can express $Y_{e,t}$ in the following:
\begin{align}
Y_{e,t}&=\frac{1}{N_{e,t}}\sum_{i}^{N_{e,t}}Y_{e(i),t}\notag\\
&=\frac{1}{N_{e,t}}\sum_{i}^{N_{e,t}}[Y_{e(i),t}(D_{i,t}^{\infty})+Z_{e,t}\cdot(Y_{e(i),t}(D_{i,t}^{e})-Y_{e(i),t}(D_{i,t}^{\infty}))]\notag\\
\label{ApeDeq3}
&=Y_{e,t}(D_{i,t}^{\infty})+Z_{e,t}\cdot(Y_{e,t}(D_{i,t}^{e})-Y_{e,t}(D_{i,t}^{\infty})).
\end{align}
where the second equality follows from Assumptions \ref{sec2as1}-\ref{sec2as3} and Assumption \ref{sec2as5}.\par
First, we consider the probability limit of the numerator in the TWFEIV estimator. By using \eqref{ApeDeq1} and \eqref{ApeDeq2}, we obtain
\begin{align}
\sum_{e}\sum_{t}N_{e,t}\hat{Z}_{e,t}\frac{1}{N_{e,t}}\sum_{i}^{N_{e,t}}Y_{e(i),t}&=\sum_{e}\sum_{t}N_{e,t}\hat{Z}_{e,t}Y_{e,t}\notag\\
\label{ApeDeq4}
&=\sum_{e}\sum_{t}N_{e,t}\hat{Z}_{e,t}\left[Y_{e,t}-Y_{e,1}-(Y_{1,t}-Y_{1,1})\right].
\end{align}
To further develop the expression, we use \eqref{ApeDeq3}:
\begin{align}
Y_{e,t}-Y_{e,1}-(Y_{1,t}-Y_{1,1})&=Y_{e,t}(D_{i,t}^{\infty})-Y_{e,1}(D_{i,1}^{\infty})-[Y_{1,t}(D_{i,t}^{\infty})-Y_{1,1}(D_{i,1}^{\infty})]\notag\\
&+Z_{e,t}\cdot(Y_{e,t}(D_{i,t}^{e})-Y_{e,t}(D_{i,t}^{\infty}))-Z_{e,1}\cdot(Y_{e,1}(D_{i,1}^{e})-Y_{e,1}(D_{i,1}^{\infty}))\notag\\
\label{ApeDeq5}
&-[Z_{1,t}\cdot(Y_{1,t}(D_{i,t}^{1})-Y_{1,t}(D_{i,t}^{\infty}))-Z_{1,1}\cdot(Y_{1,1}(D_{i,1}^{e})-Y_{1,1}(D_{i,1}^{\infty}))]
\end{align}
Substituting \eqref{ApeDeq5} into \eqref{ApeDeq4}, we obtain:
\begin{align}
&\sum_{e}\sum_{t}N_{e,t}\hat{Z}_{e,t}\left[Y_{e,t}-Y_{e,1}-(Y_{1,t}-Y_{1,1})\right]\notag\\
&=\sum_{e}\sum_{t}N_{e,t}\hat{Z}_{e,t}\left[Y_{e,t}(D_{i,t}^{\infty})-Y_{e,1}(D_{i,1}^{\infty})-[Y_{1,t}(D_{i,t}^{\infty})-Y_{1,1}(D_{i,1}^{\infty})]\right]\notag\\
\label{ApeDeq6}
&+\sum_{e}\sum_{t}N_{e,t}\hat{Z}_{e,t}Z_{e,t}\cdot(Y_{e,t}(D_{i,t}^{e})-Y_{e,t}(D_{i,t}^{\infty})),
\end{align}
where the second equality holds from \eqref{ApeDeq1} and \eqref{ApeDeq2}.\par
From \eqref{ApeDeq6}, as $N \rightarrow \infty$ , we obtain
\begin{align}
&\sum_{e}\sum_{t}N_{e,t}\hat{Z}_{e,t}\left[Y_{e,t}-Y_{e,1}-(Y_{1,t}-Y_{1,1})\right]\notag\\
&\xrightarrow{p} \sum_{e}\sum_{t}E[\hat{Z}_{i,t}|E_i=e]\cdot n_{e,t} \cdot \Bigg\{E[Y_{e,t}(D_{i,t}^{\infty})|E_i=e]-E[Y_{e,1}(D_{i,1}^{\infty})|E_i=e]\notag\\
&-\left(E[Y_{1,t}(D_{i,t}^{\infty})|E_i=1]-E[Y_{1,1}(D_{i,1}^{\infty})|E_i=1]\right)\Bigg\}\notag\\
&+\sum_{e}\sum_{t}E[\hat{Z}_{i,t}|E_i=e]\cdot n_{e,t} \cdot E[Z_{e,t}\cdot(Y_{e,t}(D_{i,t}^{e})-Y_{e,t}(D_{i,t}^{\infty}))|E_i=e]\notag\\
&=\sum_{e}\sum_{t \geq e}E[\hat{Z}_{i,t}|E_i=e]\cdot n_{e,t} \cdot E[(Y_{e,t}(D_{i,t}^{e})-Y_{e,t}(D_{i,t}^{\infty}))|E_i=e]\notag\\
\label{ApeDeq7}
&=\sum_{e}\sum_{t \geq e}E[\hat{Z}_{i,t}|E_i=e]\cdot n_{e,t} \cdot CATT^{1}_{e,t}\cdot CLATT_{e,t},
\end{align}
where $n_{e,t}$ is population share and $E[\hat{Z}_{i,t}|E_i=e]$ in cohort $e$ at time $t$. The first equality follows from Assumption \ref{sec2as1} and Assumption \ref{sec2as7}. The second equality follows from Assumption \ref{sec2as5}.\par
Next, we consider the probability limit of the numerator. We note that the structure in the numerator is same as the one in the numerator. Therefore, by the same argument, we have:
\begin{align}
&\sum_{e}\sum_{t}N_{e,t}\hat{Z}_{e,t}\frac{1}{N_{e,t}}\sum_{i}^{N_{e,t}}D_{e(i),t}\notag\\
\label{ApeDeq8}
&\xrightarrow{p} \sum_{e}\sum_{t \geq e}E[\hat{Z}_{i,t}|E_i=e]\cdot n_{e,t} \cdot CATT^{1}_{e,t}.
\end{align}
Combining the result \eqref{ApeDeq8} with \eqref{ApeDeq7}, we obtain
\begin{align*}
\hat{\beta}_{IV} &\xrightarrow{p} \beta_{IV}\\
&=\frac{\sum_{e}\sum_{t \geq e}E[\hat{Z}_{i,t}|E_i=e]\cdot n_{e,t} \cdot CATT^{1}_{e,t}\cdot CLATT_{e,t}}{\sum_{e}\sum_{t \geq e}E[\hat{Z}_{i,t}|E_i=e]\cdot n_{e,t} \cdot CATT^{1}_{e,t}}\\
&=\sum_{e}\sum_{t \geq e}w_{e,t}\cdot CLATT_{e,t}.
\end{align*}
where the weight $w_{e,t}$ is:
\begin{align*}
w_{e,t}=\frac{E[\hat{Z}_{i,t}|E_i=e]\cdot n_{e,t} \cdot CATT^{1}_{e,t}}{\sum_{e}\sum_{t \geq e}E[\hat{Z}_{i,t}|E_i=e]\cdot n_{e,t} \cdot CATT^{1}_{e,t}}.
\end{align*}\par
Completing the proof.
\subsubsection*{Supplementary of Theorem \ref{sec5thm1}}
We provide another representation of Theorem \ref{sec5thm1}. We assume that there does not exist a never exposed cohort, that is, we have $\infty \notin \mathcal{S}(E_i)$, and define $l=\max\{E_i\}$ to be the last exposed cohort.
\begin{Lemma}
\label{ApeDlemma1}
    Suppose Assumptions \ref{sec2as1}-\ref{sec2as7} hold. Assume a binary treatment and an unbalanced panel setting. If there does not exists a never exposed cohort, i.e., we have $\infty \notin \mathcal{S}(E_i)$, the population regression coefficient $\beta_{IV}$ is:
    \begin{align}
    \label{Supplement1}
   \beta_{IV}=&\sum\limits_{e}\sum\limits_{l-1 \geq t \geq e}w_{e,t}^{1}\cdot CLATT_{e,t}+\sum\limits_{e}\sum\limits_{t \geq l}w_{e,t}^{2} \cdot \Delta_{e,t},
    \end{align}
    where $\Delta_{e,t}$ is:
    \begin{align*}
\displaystyle\frac{CAET^{1}_{e,t}\cdot CLATT_{e,t}-CAET^{1}_{l,t}\cdot CLATT_{l,t}}{CAET^{1}_{e,t}-CAET^{1}_{l,t}},
    \end{align*}
    and the weights $w^{1}_{e,t}$ and $w^{2}_{e,t}$ are:
    \begin{align}
\label{ApeDeq11}
&w_{e,t}^1=\frac{E[\hat{Z}_{i,t}|E_i=e]\cdot n_{e,t} \cdot CAET^1_{e,t}}{\sum\limits_{e}\left(\sum\limits_{l-1 \geq t \geq e}E[\hat{Z}_{i,t}|E_i=e]\cdot n_{e,t}\cdot CAET^1_{e,t}+\sum\limits_{t \geq l}E[\hat{Z}_{i,t}|E_i=e]\cdot n_{e,t}\cdot (CAET^{1}_{e,t}-CAET^{1}_{l,t})\right)},\\
\label{ApeDeq12}
&w_{e,t}^{2}=\frac{E[\hat{Z}_{i,t}|E_i=e]\cdot n_{e,t}\cdot(CAET^{1}_{e,t}-CAET^{1}_{l,t})}{\sum\limits_{e}\left(\sum\limits_{l-1 \geq t \geq e}E[\hat{Z}_{i,t}|E_i=e]\cdot n_{e,t}\cdot CAET^1_{e,t}+\sum\limits_{t \geq l}E[\hat{Z}_{i,t}|E_i=e]\cdot n_{e,t}\cdot (CAET^{1}_{e,t}-CAET^{1}_{l,t})\right)}.
\end{align}
\end{Lemma}
We note that when there is no never exposed cohort, we can only identify each $CLATT_{e,t}$ before the time period $l=\max\{E_i\}$ for cohort $e \neq l$, exploiting the time trends of the unexposed treatment and outcome for cohort $l$. This implies that in equation \eqref{Supplement1}, each $\Delta_{e,t}$ is the bias term occurring from the bad comparisons performed by TWFEIV regressions. In a given application, we can estimate $CLATT_{e,t}$, $CLATT_{e,t}$, and the associated weights $w^{1}_{e,t}$, $w^{2}_{e,t}$ by constructing the consistent estimators, using \eqref{ApeDeq9} and \eqref{ApeDeq10} below.
\begin{proof}
We consider the case where there is no never exposed cohort, i.e., we have $\infty \notin \mathcal{S}(E_i)$. In this case, by using the last exposed cohort $l=\max\{E_i\}$, we obtain
\begin{align*}
\hat{\beta}_{IV}&=\frac{\sum_{e}\sum_{t}N_{e,t}\hat{Z}_{e,t}\left[Y_{e,t}-Y_{e,1}-(Y_{l,t}-Y_{l,1})\right]}{\sum_{e}\sum_{t}N_{e,t}\hat{Z}_{e,t}\left[D_{e,t}-D_{e,1}-(D_{l,t}-D_{l,1})\right]}\\
&=\frac{\sum_{e}\sum_{t}N_{e,t}\hat{Z}_{e,t}\left[D_{e,t}-D_{e,1}-(D_{l,t}-D_{l,1})\right]\cdot \widehat{WDID}_{e,t}}{\sum_{e}\sum_{t}N_{e,t}\hat{Z}_{e,t}\left[D_{e,t}-D_{e,1}-(D_{l,t}-D_{l,1})\right]}.
\end{align*}
where we define 
\begin{align*}
\widehat{WDID}_{e,t} \equiv \frac{\left[Y_{e,t}-Y_{e,1}-(Y_{l,t}-Y_{l,1})\right]}{\left[D_{e,t}-D_{e,1}-(D_{l,t}-D_{l,1})\right]}.
\end{align*}
From the Law of Large Numbers and the same argument in the proof of Theorem \ref{sec4thm2}, we have
\begin{align}
\label{ApeDeq9}
\widehat{WDID}_{e,t} \xrightarrow{p} \left\{
\begin{array}{lll}
0 & (t <e)\\
CLATT_{e,t} & (l-1 \geq t \geq e) \\
\displaystyle\frac{CAET^{1}_{e,t}\cdot CLATT_{e,t}-CAET^{1}_{l,t}\cdot CLATT_{l,t}}{CAET^{1}_{e,t}-CAET^{1}_{l,t}} & (t \geq l)
\end{array}
\right.
\end{align}
Similarly, we obtain
\begin{align}
\label{ApeDeq10}
\left[D_{e,t}-D_{e,1}-(D_{l,t}-D_{l,1})\right] \xrightarrow{p} \left\{
\begin{array}{lll}
0 & (t <e)\\
CAET^{1}_{e,t} & (l-1 \geq t \geq e) \\
CAET^{1}_{e,t}-CATT^{1}_{l,t} & (t \geq l) 
\end{array}
\right.
\end{align}
Combining the result \eqref{ApeDeq9} with \eqref{ApeDeq10} and by the Slutsky's theorem, we obtain
\begin{align*}
\beta_{IV}=&\sum\limits_{e}\sum\limits_{l-1 \geq t \geq e}w_{e,t}^{1}\cdot CLATT_{e,t}+\sum\limits_{e}\Bigg[\sum\limits_{t \geq l}w_{e,t}^{2} \cdot \displaystyle\frac{CAET^{1}_{e,t}\cdot CLATT_{e,t}-CAET^{1}_{l,t}\cdot CLATT_{l,t}}{CAET^{1}_{e,t}-CAET^{1}_{l,t}}\Bigg].
\end{align*}\par
Completing the proof.
\end{proof}\par
\subsection{Random assignment of the instrument adoption date} 
First, we set up the additional notations. We define $LATE_k^{CM}(W)$ and $LATE_k(W)$ analogous to $CLATT_k^{CM}(W)$ and $CLATT_k(W)$ in section \ref{sec4}: 
\begin{align*}
&LATE^{CM}_k(W) \equiv \sum_{t \in W}\frac{AE^{1}_{k,t}}{\sum_{t \in W}AE^{1}_{k,t}}LATE_{k,t},\\
&LATE_k(W) \equiv \frac{1}{T_W}\sum_{t \in W}AE^{1}_{k,t}LATE_{k,t},
\end{align*}
where we replace $CAET^{1}_{k,t}$ and $CLATT_{k,t}$ with $AE^{1}_{k,t}$ and $LATE_{k,t}$ respectively in $CLATT_k^{CM}(W)$ and $CLATT_k(W)$.\par
Theorem \ref{ApeCthm1} below presents the TWFEIV estimand under Assumptions \ref{sec2as1} - \ref{sec2as5} and Assumption \ref{sec5as1} (Random assignment assumption of adoption date $E_i$).
\begin{Theorem}
Suppose Assumptions \ref{sec2as1}-\ref{sec2as5} and \ref{sec5as1} holds. Then, the population regression coefficient $\beta_{IV}$ consists of two terms:
\label{ApeCthm1}
\begin{align*}
\beta_{IV}=WLATE-\Delta LATE.
\end{align*}
where we define:
\begin{align*}
WLATE &\equiv \sum_{k \neq U}w_{IV,kU}LATE^{CM}_{k}(POST(k))+\sum_{k \neq U}\sum_{l >k}w_{IV,kl}^{k}LATE^{CM}_{k}(MID(k,l))\\
&+\sum_{k \neq U}\sum_{l >k}\sigma_{IV,kl}^{l}\cdot LATE_l(POST(l)),\\
\Delta LATE &\equiv \sum_{k \neq U}\sum_{l >k}\sigma_{IV,kl}^{l}\cdot \left[LATE_k(POST(l))-LATE_k(MID(k,l))\right].
\end{align*}
The weights $w_{IV,kU}$, $w_{IV,kl}^{k}$ and $\sigma_{IV,kl}^{l}$ are the same in the proof of Theorem \ref{sec4thm2}.
\end{Theorem}\par
Theorem \ref{ApeCthm1} is analogous to Theorem \ref{sec4thm2}, but $CAET^{1}_{e,l}$ and $CLATT_{e,l}$ are replaced by $AE^{1}_{e,l}$ and $LATE_{e,l}$ respectively because we assume a random assignment of adoption date. If we consider the restrictions on the effects of the instrument on the treatment and outcome as in section \ref{sec4.2}, the similar arguments hold as in Theorem \ref{sec4thm2}, Lemma \ref{sec4lemma2} and Lemma \ref{sec4lemma3}.
\begin{proof}
First, we note that Assumption \ref{sec5as1} implies Assumption \ref{sec2as6} and Assumption \ref{sec2as7}. Therefore, we obtain the result in Theorem \ref{sec4thm2} under Assumptions \ref{sec2as1}-\ref{sec2as5} and \ref{sec5as1}.\par 
By noticing that we have $CATT^{1}_{e,l}=AE^{1}_{e,l}$ and $CLATT_{e,l}=LATE_{e,l}$ under Assumption \ref{sec5as1}, we obtain:
\begin{align*}
&CLATT_k^{CM}(W)=LATE^{CM}_k(W),\\
&CLATT_k(W)=LATE_k(W).
\end{align*}\par
By replacing $CLATT_k^{CM}(W)$ and $CLATT_k(W)$ with $LATE^{CM}_k(W)$ and $LATE_k(W)$ in Theorem \ref{sec4thm2}, we obtain the desired result.
\end{proof}
\section{Proofs and discussions in section \ref{sec7}}\label{ApeD}
In this appendix, we first derive equation \eqref{TWFEIV_between2}. We then discuss the causal interpretation of the covariate-adjusted TWFEIV estimand under staggered DID-IV designs, imposing the additional assumptions.\par
\subsection{Decomposing the between IV coefficient}
Let $\hat{C}^{D,\Tilde{z}}_{b}$ denote the covariance between $D_{i,t}$ and $\Tilde{z}_{k,t}$, the between term of $\Tilde{z}_{i,t}$. The between IV coefficient $\hat{\beta}_{b,IV}^{z}$ is:
\begin{align}
\label{TWFEIV_ApeDeq1}
\hat{\beta}_{b,IV}^{z}=\frac{\hat{C}(Y_{i,t},\Tilde{z}_{k,t})}{\hat{C}(D_{i,t},\Tilde{z}_{k,t})}=\frac{\hat{C}(Y_{i,t},\Tilde{z}_{k,t})}{\hat{C}^{D,\Tilde{z}}_{b}}.
\end{align}
To derive equation \eqref{TWFEIV_between2}, we decompose the covariance between $Y_{i,t}$ and $\Tilde{z}_{k,t}$. To do so, we first split the between term $\Tilde{z}_{k,t}$ into the between term of $Z_{i,t}$ and the between term of $\tilde{p}_{i,t}$:
\begin{align*}
\Tilde{z}_{k,t}&=[(\bar{Z}_{k,t}-\bar{Z}_{k})-(\bar{Z}_{t}-\bar{\bar{Z}})]-[(\bar{p}_{k,t}-\bar{p}_{k})-(\bar{p}_{t}-\bar{\bar{p}})]\\
&\equiv \Tilde{Z}_{k,t}-\Tilde{p}_{k,t}.
\end{align*}
Then, we have
\begin{align}
\hat{C}(Y_{i,t},\Tilde{z}_{k,t})&=\frac{1}{NT}\sum_{i}\sum_{t}Y_{i,t}[(\bar{z}_{k,t}-\bar{z}_{k})-(\bar{z}_{t}-\bar{\bar{z}})]\notag\\
&=\frac{1}{T}\sum_{k}n_k \sum_{t}\bar{Y}_{k,t}[(\bar{z}_{k,t}-\bar{z}_{k})-(\bar{z}_{t}-\bar{\bar{z}})]\notag\\
&=\sum_{k}\sum_{l}n_k n_l \frac{1}{T}\sum_{t}(\bar{Y}_{k,t}-\bar{Y}_{l,t})\Tilde{Z}_{k,t}-\sum_{k}\sum_{l}n_k n_l \frac{1}{T}\sum_{t}(\bar{Y}_{k,t}-\bar{Y}_{l,t})\Tilde{p}_{k,t}\notag\\
\label{TWFEIV_ApeDeq2}
&=\sum_{k}\sum_{l > k}(n_k+n_l)^2[\hat{C}_{kl}^{D,Z}\hat{\beta}_{IV,kl}^{2 \times 2}-\hat{C}_{b,kl}^{p}\hat{\beta}_{b,IV,kl}^{p}].
\end{align}
$\hat{\beta}_{IV,kl}^{2 \times 2}$ is an estimator obtained from an IV regression of $Y_{i,t}$ on $D_{i,t}$ with $\Tilde{Z}_{k,t}$ as the excluded instrument in $(k,l)$ cell subsample. $\hat{C}_{kl}^{D,Z}$ is the covariance between $D_{i,t}$ and $\Tilde{Z}_{k,t}$ in $(k,l)$ cell subsample. $\hat{\beta}_{b,IV,kl}^{p}$ is an estimator obtained from an IV regression of $Y_{i,t}$ on $D_{i,t}$ with $\Tilde{p}_{k,t}$ as the excluded instrument in $(k,l)$ cell subsample. $\hat{C}_{b,kl}^{p}$ is the covariance between $D_{i,t}$ and $\Tilde{p}_{k,t}$ in $(k,l)$ cell subsample.
\par
By combining \eqref{TWFEIV_ApeDeq1} with \eqref{TWFEIV_ApeDeq2}, we obtain equation \eqref{TWFEIV_between2}.
\subsection{Causal interpretation of the covariate-adjusted TWFEIV estimand}
\par
This section considers the causal interpretation of the covariate-adjusted TWFEIV estimand $\beta_{IV}^{X}$. To simplify the analysis, we first make the following assumptions.  \cite{Goodman-Bacon2021-ej} also make similar assumptions to investigate the causal interpretation of the covariate-adjusted TWFE estimand in Appendix B.
\begin{itemize}
\item[(i)] Time-varying covariates $X_{i,t}$ are not affected by instrument (policy shock).
\item[(ii)] Time-varying covariates $X_{i,t}$ do not vary within cohorts.
\item[(iii)] The coefficients obtained from regressing $\Tilde{Z}_{i,t}$ on $\Tilde{X}_{i,t}$ in $(k,l)$ cell subsample are the same regardless of the pair $(k,l)$.
\end{itemize}
\par
Because Assumption (ii) implies that the within term is equal to zero, the covariate-adjusted TWFEIV estimator $\hat{\beta}_{IV}^{X}$ simplifies to
\begin{align*}
\hat{\beta}_{IV}^{X}=\sum_{k}\sum_{l > k}s_{b,kl}\hat{\beta}_{b,IV,kl}^{z}.
\end{align*}\par
Assumption (iii) guarantees that $\hat{\beta}_{b,IV,kl}^{z}$ is equal to the between coefficient obtained from estimating equation \eqref{TWFEIV-covariate1} in $(k,l)$ subsample, which we denote $\hat{\beta}_{b,IV,kl}^{z,X}$ hereafter. To see this formally, let $\Tilde{p}_{i,t}^{kl} \equiv \hat{\Gamma}_{k,l}\Tilde{X}_{i,t}$ denote the linear projection obtained from regressing $\Tilde{Z}_{i,t}$ on $\Tilde{X}_{i,t}$ in $(k,l)$ subsample and let $\Tilde{p}_{j,t}^{kl}$ denote the between term of $\Tilde{p}_{i,t}^{kl}$ in cohort $j$. We note that $\Tilde{p}_{j,t}^{kl} \neq \Tilde{p}_{j,t}$ (the between term of $\Tilde{p}_{i,t}$) holds in general because $\Tilde{p}_{i,t}=\hat{\Gamma}\Tilde{X}_{i,t}$ is estimated using the whole sample. Then, we have
\begin{align*}
\hat{\beta}_{b,IV,kl}^{z}=\frac{\hat{C}(Y_{i,t},\Tilde{z}_{j,t})}{\hat{C}(D_{i,t},\Tilde{z}_{j,t})}&=\frac{\hat{C}(Y_{i,t},\Tilde{Z}_{j,t}-\Tilde{p}_{j,t}^{kl})+\hat{C}(Y_{i,t},\Tilde{p}_{j,t}^{kl}-\Tilde{p}_{j,t})}{\hat{C}(D_{i,t},\Tilde{Z}_{j,t}-\Tilde{p}_{j,t}^{kl})+\hat{C}(D_{i,t},\Tilde{p}_{j,t}^{kl}-\Tilde{p}_{j,t})}\\
&=\frac{\hat{C}(D_{i,t},\Tilde{Z}_{j,t}-\Tilde{p}_{j,t}^{kl})\hat{\beta}_{b,IV,kl}^{z,X}+\hat{C}(D_{i,t},\Tilde{p}_{j,t}^{kl}-\Tilde{p}_{j,t})\hat{\beta}_{b,IV,kl}^{dif}}{\hat{C}(D_{i,t},\Tilde{Z}_{j,t}-\Tilde{p}_{j,t}^{kl})+\hat{C}(D_{i,t},\Tilde{p}_{j,t}^{kl}-\Tilde{p}_{j,t})},\hspace{3mm}j=k,l,
\end{align*}
where $\hat{\beta}_{b,IV,kl}^{dif}$ is an estimator obtained from an IV regression of $Y_{i,t}$ on $D_{i,t}$ with the difference $\Tilde{p}_{j,t}^{kl}-\Tilde{p}_{j,t}$ as the excluded instrument. Because Assumption (iii) ($\Tilde{p}_{j,t}^{kl}=\Tilde{p}_{j,t}$) implies $\hat{C}(D_{i,t},\Tilde{p}_{j,t}^{kl}-\Tilde{p}_{j,t})=0$, we obtain $\hat{\beta}_{b,IV,kl}^{z}=\hat{\beta}_{b,IV,kl}^{z,X}$.\par
Hereafter, we assume the identifying assumptions in staggered DID-IV designs and Assumption (i)-(iii). We focus on the between coefficient $\hat{\beta}_{b,IV,kU}^{z}=\hat{\beta}_{b,IV,kU}^{z,X}$ as it clarifies how covariates affect the interpretation of the TWFEIV estimand:
\begin{align*}
\hat{\beta}_{b,IV,kU}^{z}=\frac{\hat{C}(Y_{i,t},\Tilde{Z}_{j,t})-\hat{C}(Y_{i,t},\Tilde{p}_{j,t}^{kl})}{\hat{C}(D_{i,t},\Tilde{Z}_{j,t})-\hat{C}(D_{i,t},\Tilde{p}_{j,t}^{kl})},\hspace{3mm}j=k,U.
\end{align*}
Then, by the similar calculations in the proof of Theorem \ref{sec4thm2}, we obtain
\begin{align}
\hat{C}(Y_{i,t},\Tilde{Z}_{j,t})&=\hat{V}_{kU}^{z}\hat{D}_{kU}^{2 \times 2}\beta_{kU,IV}^{2 \times 2}\notag\\
\label{TWFEIV_causal1}
&\xrightarrow[]{p}V_{kU}^{z}\cdot CAET_{k}^{1}(POST(k))\cdot CLATT_{k}^{CM}(POST(k)),
\end{align}
and
\begin{align}
\hat{C}(Y_{i,t},\Tilde{p}_{j,t}^{kl})&=\frac{n_{kU}(1-n_{kU})}{T}\sum_{t}(\bar{Y}_{kt}-\bar{Y}_{Ut})\cdot [(\bar{p}_{k,t}^{kU}-\bar{p}_{k}^{kU})-(\bar{p}_{U,t}^{kU}-\bar{p}_{U}^{kU})]\notag\\
&\xrightarrow[]{p} \frac{N_{kU}(1-N_{kU})}{T}\sum_{t}\left\{E[Y_{i,t}(D_{i,t}^{\infty})|E_i=k]-E[Y_{i,t}(D_{i,t}^{\infty})|E_i=U]\right\}[(p_{k,t}^{kU}-p_{k}^{kU})-(p_{U,t}^{kU}-p_{U}^{kU})]\notag\\
\label{TWFEIV_causal2}
&+\frac{N_{kU}(1-N_{kU})}{T-(k-1)}\sum_{t \geq k}\underbrace{CAET_{k}\cdot CLATT_{k,t}}_{CAIET_{k,t}}\cdot[(p_{k,t}^{kU}-p_{k}^{kU})-(p_{U,t}^{kU}-p_{U}^{kU})],
\end{align}
where $N_{kU}$ and $[(p_{k,t}^{kU}-p_{k}^{kU})-(p_{U,t}^{kU}-p_{U}^{kU})]$ are the probability limits of $n_{kU}$ and $[(\bar{p}_{k,t}^{kU}-\bar{p}_{k}^{kU})-(\bar{p}_{U,t}^{kU}-\bar{p}_{U}^{kU})]$, respectively. Equations \eqref{TWFEIV_causal1} and \eqref{TWFEIV_causal2} indicate that covariates affects the causal interpretation of $\hat{\beta}_{b,IV,kU}^{z}$ in two ways. First, it additionally introduce the covariance between the difference in unexposed outcomes and the difference in the variation of the linear projection for cohorts $k$ and $U$ (the first term in equation \eqref{TWFEIV_causal1}). Second, it additionally introduce the covariance between the $CAIET_{k,t}$ and the difference in the variation of the linear projection for cohorts $k$ and $U$ (the second term in equation \eqref{TWFEIV_causal2}).
\end{document}